\documentclass{article}
\usepackage[letterpaper,margin=1in]{geometry}
\usepackage{makecell}

\usepackage[utf8]{inputenc} 
\usepackage[T1]{fontenc}    
\usepackage{hyperref}       
\usepackage{url}            
\usepackage{booktabs}       
\usepackage{amsfonts}       
\usepackage{nicefrac}       
\usepackage{microtype}      
\usepackage{xcolor}         
\usepackage{natbib}
\usepackage{algorithm}
\usepackage{algpseudocode}
\usepackage{colortbl}

\usepackage{authblk}

\usepackage[utf8]{inputenc}
\usepackage{amsmath}
\usepackage{natbib}
\usepackage{amsfonts}
\usepackage{amsthm}

\usepackage{graphicx}
\usepackage{subcaption}
\usepackage{multirow}
\usepackage{wrapfig}
\usepackage{url}
\usepackage{xcolor}

\newcommand{\xhdr}[1]{\vspace{1mm} \noindent\textbf{#1.}}

\newtheorem{theorem}{Theorem}
\numberwithin{theorem}{section}
\newtheorem{proposition}{Proposition}
\numberwithin{proposition}{section}
\newtheorem{definition}{Definition}
\numberwithin{definition}{section}
\newtheorem{lemma}{Lemma}
\numberwithin{lemma}{section}
\newtheorem{corollary}{Corollary}
\numberwithin{corollary}{section}

\newtheorem{observation}{Observation}
\numberwithin{observation}{section}

\usepackage[suppress]{color-edits}
\addauthor[Ben]{bl}{red}
\addauthor[Hoda]{hh}{blue}

\title{Fine-Tuning Games:\\
 Bargaining and Adaptation for General-Purpose Models
}

\author{
Benjamin Laufer\footnote{Cornell Tech, New York, NY, USA}
\ \ \ \ \ \ \ \ \ \ \ 
Jon Kleinberg\footnote{Cornell University, Ithaca, NY, USA}\ $^{, \text{\S}}$
\ \ \ \ \ \ \ \ \ \ \ 
Hoda Heidari\footnote{Carnegie Mellon University, Pittsburgh, PA, USA}\ $^,$\footnote{Joint senior author.}
}

\date{\vspace{-.1in}}

\begin{document}

\maketitle

\begin{abstract}
Recent advances in Machine Learning (ML) and Artificial Intelligence (AI) follow a familiar structure: A firm releases a large, pretrained model. It is designed to be adapted and tweaked by other entities to perform particular, domain-specific functions. The model is described as `general-purpose,’ meaning it can be transferred to a wide range of downstream tasks, in a process known as \textit{adaptation} or \textit{fine-tuning}. Understanding this process – the strategies, incentives, and interactions involved in the development of AI tools – is crucial for making conclusions about societal implications and regulatory responses, and may provide insights beyond AI about general-purpose technologies. We propose a model of this adaptation process. A Generalist brings the technology to a certain level of performance, and one or more Domain specialist(s) adapt it for use in particular domain(s). Players incur costs when they invest in the technology, so they need to reach a bargaining agreement on how to share the resulting revenue before making their investment decisions. We find that for a broad class of cost and revenue functions, there exists a set of Pareto-optimal profit-sharing arrangements where the players jointly contribute to the technology. Our analysis, which utilizes methods based on bargaining solutions and sub-game perfect equilibria, provides insights into the strategic behaviors of firms in these types of interactions. For example, profit-sharing can arise even when one firm faces significantly higher costs than another. After demonstrating findings in the case of one domain-specialist, we provide closed-form and numerical bargaining solutions in the generalized setting with $n$ domain specialists. We find that any potential domain specialization will either \textit{contribute}, \textit{free-ride}, or \textit{abstain} in their uptake of the technology, and provide conditions yielding these different responses.
\end{abstract}






\section{Introduction}

Large-scale AI models have garnered a great deal of excitement because they are considered to be \textit{general purpose} \citep{goldfarb2023could, eloundou2023gpts, crafts2021artificial, pan2020privacy, wei2022emergent, ouyang2022training}. 
Some have referred to these technologies as \textit{foundation models} \citep{bommasani2021opportunities,bommasani2023ecosystem,fei2022towards} because they are designed as massive, centralized models that support potentially many downstream uses. For example, \citet{bommasani2021opportunities} write, ``a foundation model is itself incomplete but serves as the common basis from which many task-specific models are built via adaptation.'' There is palpable excitement about these technologies. But to turn their potential into actual use and impact, one needs to specialize, tweak, and evaluate the technology for particular application domains. This process takes various names, including \textit{adaptation} \citep{peters2019tune} and, in some contexts, \textit{fine-tuning} \citep{tajbakhsh2016convolutional,zhang2020revisiting,kumar2022fine}.

Notably, the process of adapting a technology involves multiple parties. Technology teams developing ML and AI technologies rely on outside entities to adapt, tweak, transfer, and integrate the general-purpose model. This dynamic suggests a latent strategic interaction between producers of a foundational, general-purpose technology and specialists considering whether and how to adopt the technology in a particular context. Understanding this interaction is necessary to study the social, economic, and regulatory consequences of introducing the technology.

This paper employs methods from economic theory to model and analyze this interaction. We put forward a model of fine-tuning where the interaction between two agents, a generalist and a domain-specialist, determines how they'll bring a general-purpose technology to market (Figure \ref{fig:gameplay}). The result of this interaction is a domain-adapted product that offers a certain level of \textit{performance} to consumers, in exchange for a certain level of surplus revenue for the producers. Crucially, the producers must decide how to distribute the surplus, and engage in a bargaining process in advance of making their investment decisions. An immediate intuition might be to divide the surplus based on contribution to the technology --- however, this is one of many potential bargaining solutions, each with different normative assumptions and implications for the technology's performance and the distribution of utility. 

Through this analysis, we uncover several general principles that apply not just to today's AI technologies, but to a potentially wide swath of models that exhibit a similar structure --- i.e., developed for general use and adapted to one or more domains to produce revenue. Thus, even as these technologies improve and develop, our proposed model of fine-tuning may continue to describe how they may be adapted for real-world use(s). 
Further, some of our findings apply to other general-purpose technologies outside the AI context. For example, cloud computing infrastructure enables a number of consumer-facing services that use web hosting, database services, and other on-demand computing resources. Additive manufacturing (e.g., 3D printing) requires the production of a general-purpose technology that other entities use to create valuable products in particular domains. Digital marketplaces, too, are general market-making technologies that enable specialists (vendors) to sell goods, subject to an agreement over surplus.

Our main conceptual contribution is modeling the adaptation process as a \textbf{multi-stage game} consisting of (1) a \textbf{bargaining process} between a general-purpose technology producer ($G$) and one or more domain specialists ($D$), and (2) two additional stages for $G$ and $D$ to invest in performance, respectively (see Figure~\ref{fig:gameplay}). 
Both players bargain over how to share revenue, and each takes a turn contributing to the technology's performance before it reaches the market. Within the set of Pareto-optimal revenue-sharing agreements, we introduce a number of \textit{bargaining solutions} that represent potential arrangements for how entities involved in AI's development should distribute profit and effort. These bargaining solutions can be thought of as diverse normative proposals for how to appropriately distribute welfare.

Our analysis consists in deriving the sub-game perfect equilibrium strategies, identifying the set of Pareto-optimal bargaining agreements, and then solving for various bargaining solutions. Even in the presence of significant cost differentials, we find bargaining leads to profit-sharing agreements because specialists can leverage their power to exit the deal, reducing the reach of the technology --- or, in the case of one specialist, preventing the technology from being produced altogether. For fine-tuning games with a somewhat general set of cost and revenue functions, we develop a method for identifying Pareto-optimal bargains. 
A high-level take-away from our analysis is a characterization of the specialist fine-tuning strategy. We find that any potential adaptor of a technology falls into one of three groups: \textbf{contributors}, who invest effort before selling the technology; \textbf{free-riders}, who sell the technology without investing any additional effort; and \textbf{abstainers}, who do not enter any fine-tuning agreement and opt not to bring the technology to their particular domain. It turns out, using only marginal information about a domain (0th- and 1st-order approximations of cost and revenue), it is possible to reliably determine which strategy the adaptor will take for a notably broad set of scenarios and cost and revenue functions (Section \ref{subsec:specialist-regimes}).


Some have suggested that scholarship on AI and data-driven technologies focuses predominantly on the technical developments without situating these developments in political economy (though notable exceptions exist) \citep{abebe2020roles,trajtenberg2018ai,brevini2020revisiting,cobbe2023understanding, widder2023dislocated}. We propose a model that accounts for the different interests and interactions involved in the development of new, general-purpose AI technology. Our model enables analysis on how these interactions affect market outcomes like performance in practice. Understanding these interactions may also inform future regulation of harms that can arise from large-scale ML technologies.
\subsection{Related Work}

There exists a considerable body of literature on methods for fine-tuning and adapting general ML models. Our work leverages economic theory to understand the incentives and strategies that determine how these general-purpose technologies develop.

\textbf{Approaches to fine-tuning.} New applications of ML often involve leveraging an existing model to a specific task, in a process known as transfer learning \citep{zhuang2020comprehensive}. As a result, a variety of broad and flexible base models have been developed (`pretrained') for downstream adaptation to particular tasks. These include large language models~\citep{brown2020language,howard2018universal,dai2015semi} and visual models~\citep{radford2021learning,yuan2021florence}. Fine-tuning is an approach where new data and training methods are applied to a pretrained base model to improve performance on a domain-specific task~\citep{dodge2020fine}. Fine-tuning often consists of several steps: (1) gathering, processing and labeling domain-specific data, (2) choosing and adjusting the base model's architecture (including number of layers~\citep{wang2017growing} and parameters~\citep{sanh2020movement}) and the appropriate objective function~\citep{gunel2020supervised}, (3) Updating the model parameters using techniques like gradient descent, and (4) evaluating the resulting model and refining if necessary. 

\textbf{Economic models of general-purpose technology production.} Several lines of work in growth economics address the development and diffusion of general-purpose technologies (or GPTs). \citet{bresnahan2010general} provides a general survey of this concept. \citet{jovanovic2005general} offers a historic account of technologies such as electricity and information technology as GPTs with major impacts on the United States economy. Scholars have examined the effects of factors such as knowledge accumulation, entrepreneurial activity, network effects, and sectoral interactions on the creation of GPTs~\citep{helpman1998general}. The model presented here abstracts away the forces giving rise to the invention of general-purpose technologies, and instead focuses on the later-stage decision of when (or at what performance level) to release the GPT to market for domain-specialization. 

Some have suggested that general-purpose technologies create the need for new business models that describe their impact on individual sectors \citep{lipsey2005economic}. \citet{gambardella2010business} propose one such model of domain adaptation for a general-purpose technology that is based on revenue sharing --- however, they do not use bargaining or multi-stage strategy to describe how the technology is developed and brought to market. Our notion of \textit{performance} as it relates to model technologies is inspired by economic models of product innovation \citep{viscusi1993product,cooper1984strategy}. 

\textbf{Bargaining and joint production.} This work draws from a long line of research on welfare economics and cooperative game theory devoted to understanding how agents reach agreements when their interests are intertwined~\citep{driessen2013cooperative}. Methods have been developed for finding an optimal set of agreements (e.g., contract curves~\citep{edgeworth1881mathematical} and cores~\citep{shapley1971cores,aumann1961core}) in exchange economies, where agents can trade goods. When parties must reach an agreement to jointly produce a product, they often engage in a \textit{bargain} -- we discuss bargaining further in Section \ref{subsec:bargaining-primer}. Existing empirical work observes how real people or firms bargain, and measures how close these agreements are to those proposed by theorists~\citep{takeuchi2022bargaining, fischer2019collusion}. A setting with similar models is the development \textit{supply chains} where different firms or entities negotiate over how much effort they each invest and how much profit they each receive~\citep{wu2009bargaining,du2014newsvendor}. A related body of work is referred to as the \textit{hold-up problem} \citep{rogerson1992contractual}. This work analyzes settings where two (or more) agents negotiate over an \textit{incomplete} contract and distribute surplus \citep{hart1988incomplete}. In these models, after an initial agreement, players are able to re-negotiate and alter parts of the contract, yielding shifts in strategy. 

\textbf{Game Theory and ML.} Our paper contributes to a line of work using game-theoretic methods to describe the development of (and responses to) ML models~\citep{hardt2016strategic,liu2022strategic,perdomo2020performative,kleinberg2020classifiers} and their societal implications~\citep{milli2019social,hu2019disparate,laufer2023strategic}. \citet{donahue2021model} explore a game-theoretic setting where agents may voluntarily take part in a federated learning arrangement. Their setting is a coalitional game among parties that all move simultaneously, whereas ours is a sequential game that involves parties with different roles in the process. Focusing on digital platforms, \citet{hardt2023algorithmic} describe interactions between a firm implementing an ML algorithm and collectives of users 
who manipulate their data to influence the algorithm. 

\section{A Model of Fine-Tuning}

In this section, we put forward a model of fine-tuning a data-driven technology for use in a domain-specific context. The technology is developed in two steps: First, a general-purpose producer develops a technology up to a certain level of performance. Then, a domain-specific producer decides whether to adopt the technology, and how much to invest in the technology to further improve its performance beyond the general-purpose baseline. After these steps, the two entities share a payout.

\textbf{Generalist.} Player $G$ (for General-purpose producer) is the first to invest in the technology’s performance, and brings the performance level to $\alpha_0 \in \mathbb{R}^+$. $G$ is motivated to invest in the technology because, ultimately, the technology's performance level determines the revenue $G$ earns.

\textbf{Domain Specialist.} After investing in the technology, $G$ can offer the technology to a domain-specialist, denoted $D$, who fine-tunes the model to their specific use case. If $D$ and $G$ enter an agreement, $D$ will invest in improving the technology's performance from $\alpha_0$ to $\alpha_1 \in \mathbb{R^+}$ where $\alpha_1\geq \alpha_0$.

\textbf{Revenue and costs.} The technology's \textit{performance}, $\alpha_1$, determines the total revenue that can be gained from fine-tuning the technology in that domain. In particular, we assume there is a monotonic function $r: \mathbb{R}^+ \rightarrow \mathbb{R}^+$ such that $r(\alpha_1)$ is the total revenue generated by performance level $\alpha_1$. Unless otherwise specified, we assume $r(\cdot)$ is the identity function, that is, the total revenue brought by the technology is $\alpha_1$. The cost associated with producing $\alpha_1$ requires considering the two steps involved with developing the technology: general production and fine-tuning. 
We say that $G$ faces cost function $\phi_{0}(\alpha_0): \mathbb{R}^+ \rightarrow \mathbb{R}^+$ to produce a general technology at performance-level $\alpha_0$. $D$ faces cost function $\phi_1(\alpha_1; \alpha_0): \mathbb{R}^+ \rightarrow \mathbb{R}^+$ to bring the technology from performance $\alpha_0$ to performance $\alpha_1$. 
We assume these cost functions are publicly known. Unless otherwise specified, we also assume $r(0)=0$, $\phi_0(0)=0$, and $\phi_1(\alpha_1=\alpha_0; \alpha_0)=0$, meaning that not investing in the technology is free and brings in zero revenue.

\begin{figure*}
    \centering
    \includegraphics[width=.9\linewidth]{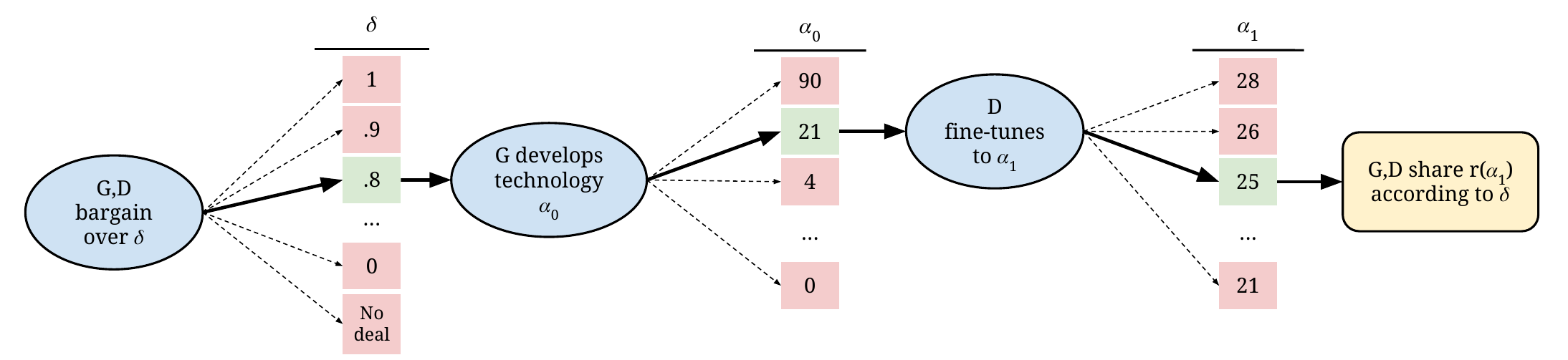}
    \caption{Illustration of the fine-tuning game. In the first step, players bargain over revenue-sharing agreement $\delta$. In this example, they agree that G will receive $80\%$ of the revenue and D will receive $20\%$. In the second step, $G$ develops the technology to performance level $\alpha_0 = 21$. In the third step, $D$ `fine-tunes' the technology to $\alpha_1=25$. If the players collectively receive revenue of $25$, they would share so that $G$ receives $20$ and $D$ receives $5$.}
    \label{fig:gameplay}
    \vspace{-3mm}
\end{figure*}

\textbf{The fine-tuning game.} The players are $G$ and $D$. In deciding whether to purchase the technology, $D$ negotiates revenue sharing with $G$. $G$ and $D$ share revenue $r(\alpha_1)$ according to a bargaining parameter $\delta \in [0,1]$. At the end of the game, $G$ receives $\delta r(\alpha_1)$ in revenue, and $D$ receives $(1-\delta) r(\alpha_1)$. The model fine-tuning game consists in each player deciding their level of investment and collectively bargaining to decide $\delta$. The game proceeds as follows:
\begin{enumerate}
    \item $G$ and $D$ negotiate bargaining coefficient $\delta \in [0,1]$.
    \item $G$ invests in a general-purpose technology, subject to cost $\phi_0(\alpha_0)$, yielding performance $\alpha_0$.
    \item $D$ fine-tunes the technology, subject to cost $\phi_1(\alpha_1;\alpha_0)$, yielding performance $\alpha_1$.    
\end{enumerate}
 The steps of the game are illustrated in Figure \ref{fig:gameplay}. Players earn the following utilities, defined as revenue share minus cost:
\begin{equation}
\label{u-g}
    U_G(\delta) := \delta r(\alpha_1) - \phi_0(\alpha_0), \ \ 
U_{D}(\delta) := (1-\delta) r(\alpha_1) - \phi_1(\alpha_1;\alpha_0).
\end{equation}
If the players do not agree to a feasible bargain $\delta \in [0,1]$, then the bargaining outcome is referred to as \textit{disagreement}. In this scenario, the generalist receives $d_0$ and the specialist receives $d_1$. We assume, unless otherwise specified, that the disagreement scenario is described by $d_0=d_1=0$.

\subsection{Primer on Bargaining Games}
\label{subsec:bargaining-primer}

Bargaining games are a potentially useful method for computer science research. In this section we include a primer on these methods before demonstrating their use in our model.


A bargain is a process for identifying joint agreements between two or more agents on how to share payoff. The \textbf{\textit{Bargaining Problem}}, formalized by \cite{nash1950bargaining}, consists of two players that must jointly decide how to share surplus profit. The problem consists of a set of feasible agreements and a `disagreement’ alternative, which specifies the utilities players receive if they do not come to an agreement.

Bargaining solutions are established ways to select among candidate agreements on how to share surplus. Different bargaining solutions, proposed over the years by mathematicians and economists, aim to satisfy certain desiderata like fairness, Pareto optimality, and utility-maximization. Typically, solving for bargaining solutions consists in defining some measure of \textit{joint utility} between players (e.g. take the sum, product, or minimum of the players’ utilities). The feasible, Pareto-optimal solution that maximizes this joint utility is known as a \textbf{\textit{bargaining solution}}.

Bargaining solutions are normative: they provide guidelines for how surplus payoffs should be distributed. Solutions are inspired by moral theories like utilitarianism (which aims to maximize the sum of utilities) and egalitarianism (which aims to maximize the worst-off agent). We demonstrate the use of bargaining solutions in the subsequent sections.
\subsection{Pareto-Optimal Bargains}

Our model of the fine-tuning process unfolds in two stages: the first stage is a bargain where the players must jointly agree on $\delta$, and the second stage is a sequential game where the players make decisions individually in order (i.e., $G$ moves first and $D$ moves second). Our analysis will identify the players' equilibrium strategies and a variety of bargaining solutions $\delta$ with different welfare implications. In order to derive solutions, it is important to define \textit{Pareto dominance} and \textit{Pareto efficiency}. Since our analysis relies on these concepts, in this section, we state our first result deriving the set of Pareto-optimal solutions for a general set of cost and revenue functions. For completeness, we begin by defining relevant concepts.

\begin{definition}[Pareto-dominant agreements]
    A bargaining agreement $\delta_a$ \textbf{Pareto-dominates} an alternative agreement $\delta_b \neq \delta_a$ iff at least one player gains utility by switching from $\delta_b$ to $\delta_a$, and no players lose utility.
\end{definition}

 \begin{definition}[Pareto-optimal agreements]
 \label{def:pareto-optimal}
     A \textbf{Pareto-optimal} agreement is one where no alternative agreement would improve the utility of one player without decreasing the utility of another player.
 \end{definition}
 
 \begin{definition}[Strictly Unimodal Function]
\label{def:strict-unimodal}
    A function $f:\mathbb{R}\rightarrow\mathbb{R}$ is called a \textbf{strictly unimodal function} over a real domain $x\in \mathbf{D}$ if there exists some value $m\in \mathbf{D}$ such that $f$ is strictly increasing $\forall x\leq m$ and $f$ is strictly decreasing $\forall x\geq m$.
\end{definition}

 When reasoning about how two agents can jointly reach an agreement, it is useful to start by considering the scenario where one player is \textit{all-powerful}, meaning the bargain is determined solely to maximize one player's utility. The formal definition of this sort of bargaining arrangement is provided below.

 \begin{definition}[Powerful-P solution]
 \label{def:powerfulP}
 For a given fine-tuning game player $P\in\{G,D\}$, the powerful-P solution is the revenue-sharing agreement $\delta^{\textit{Powerful }P}\in[0,1]$ that maximizes $P$'s utility:
 $$\delta^{\textit{Powerful }P} = \text{argmax}_{\delta\in[0,1]}U_P(\delta).$$
 \end{definition}

\subsection{Pareto-Optimal Set for Unimodal Utilities}
 We are now in a position to state our first theorem, which characterizes the Pareto-optimal solutions to any fine-tuning game with strictly unimodal utility functions.

\begin{theorem}
\label{thm:unimodal-pareto}
    Consider a fine-tuning game where players bargain over a parameter $\delta$. If the players' utilities are strictly unimodal functions of $\delta$, the set of Pareto-optimal agreements is the interval between their optima $\{\delta^{\text{Powerful }D}, \delta^{\text{Powerful }G}\}$, where both players' utilities are greater than the disagreement scenario. If no such interval exists, then disagreement is Pareto-optimal.
\end{theorem}

The proof is provided in Appendix \ref{app:section-2}. To provide some intuition for the proof, consider the range of agreements $\delta$ between the point which maximizes one player's utility (say, $\delta^{\textit{Powerful }D}$) and the point which maximizes the other ($\delta^{\textit{Powerful }G}$). Agreements within this range exhibit a trade-off between the two utilities. Agreements outside this range, however, leave both players worse-off than, e.g., the nearest powerful-P solution, so they are Pareto-dominated. This intuition is illustrated in Figure \ref{fig:pareto-intuition}.


\begin{figure}
    \centering
    \includegraphics[width=.78\linewidth]{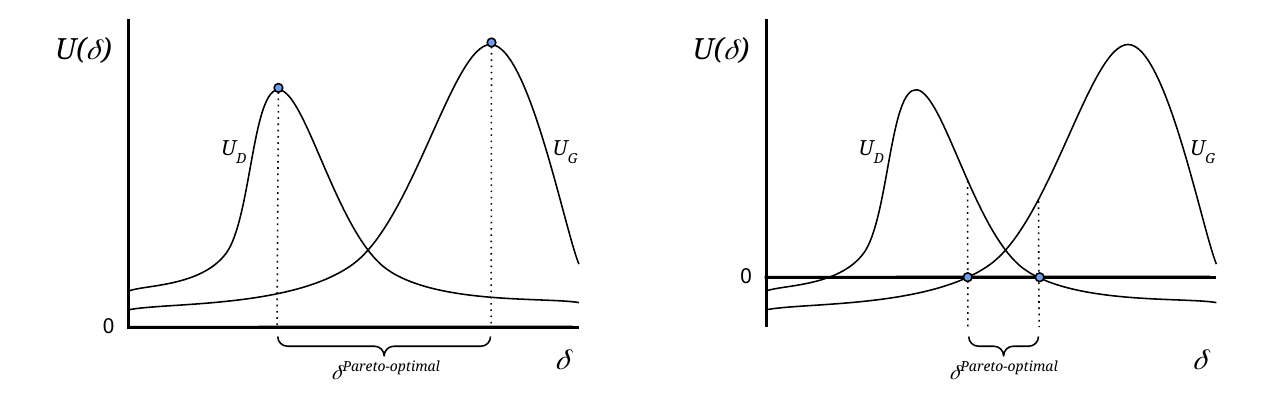}
    \caption{Example to illustrate Theorem \ref{thm:unimodal-pareto}. Left: For two strictly unimodal, positive utility functions over a bargaining parameter $\delta$, the set of Pareto-optimal bargaining agreements is the interval between their optima. Within this interval, players' utilities are characterized by a trade-off, whereas outside this interval, any agreement is Pareto-dominated. Right: The non-numerical bargaining outcome (`disagreement') consists in players receiving 0 utility. Thus, if a potential bargaining agreement yields negative utility for one or both players, they prefer to exit the agreement. The relevant set of Pareto-optimal agreements is constrained to only the interval where both players receive payout that exceeds the disagreement scenario. If no such interval exists, disagreement is Pareto-optimal.}
    \label{fig:pareto-intuition}
\end{figure}

Theorem \ref{thm:unimodal-pareto} applies to a notably broad set of utility functions. For example, any strictly increasing, strictly decreasing, or strictly concave function on the interval $\delta\in [0,1]$ is also strictly unimodal. 

Equipped with the theorem above, solving the fine-tuning game consists of the following steps: 
 (\textbf{1}) Use backward induction to solve for $D$ and $G$'s strategies, represented by $\alpha_1^*$ and $\alpha_0^*$, in terms of $\delta$.
 (\textbf{2}) Find the set of Pareto-optimal bargaining agreements $\delta$ between the powerful-D and powerful-G solutions.
 (\textbf{3}) Within the Pareto set, solve for bargaining agreements that maximize some joint function of the players' utilities.

\section{Analysis for Polynomial Costs} 
\label{sec:quadratic}

Our model applies to general cost and revenue functions, and in Section \ref{sec:multi} we provide results at this general level.
But to understand how the central parameters of the model interact in closed form, it is also useful to study instantiations of the model with specific functional forms.  Accordingly, we show in this section how to solve the model with a set of polynomial cost functions as a paradigmatic instance of convex cost functions, where the marginal costs increase as the technology is improved.
Following this, we show how to draw conclusions about the model with general costs. 
In this section, cost functions take the following polynomial function forms:
\begin{equation}
\label{phi-0-polynomial-single}
     \phi_0(\alpha_0) := c_0\alpha_0^{k_0},\ \  
    \phi_1(\alpha_1;\alpha_0) := c_1(\alpha_1-\alpha_0)^{k_1}.
\end{equation}
Here, $c_0,c_1>0$ since costs should increase with investment, and $k_0,k_1>1$, 
meaning that an incremental improvement grows costlier at higher levels of performance. We will continue to assume that $r(\alpha_1)=\alpha_1$ throughout this section's analysis.


First (\ref{subsec:subgame-quadratic}), we derive the subgame perfect equilibrium strategies $\alpha_0^*,\alpha_1^*$ for fixed $\delta$. Second (\ref{subsec:pareto-quadratic}), we find the set of Pareto-optimal revenue-sharing schemes $\delta^{Pareto}$. Reaching a revenue-sharing agreement $\delta^* \in \delta^{Pareto}$ is modeled as a bargaining problem because the players must decide how to share surplus utility. So, third (\ref{subsec:bargain-quadratic}), we define six potential bargaining solutions: Powerful-$G$, Powerful-$D$, Vertical Monopoly, Egalitarian, Nash Bargaining Solution, and Kalai-Smorodinsky. Where possible, we derive closed-form expressions for these solutions. We end by discussing the implications of these different revenue-sharing schemes.

\subsection{Subgame Perfect Equilibrium for a Given $\delta$}
\label{subsec:subgame-quadratic}

We use backward induction to determine the fine-tuning game's subgame perfect equilibrium (which we will refer to as a `solution' or `equilibrium').  Fixing the outcome of the initial negotiation, $\delta$, it is possible establish the following closed-form solution:
\begin{theorem}
\label{subgame-perfect-eq}
For a fixed $\delta$, the sub-game perfect equilibrium of the fine-tuning game with polynomial costs yields the following best-response strategies: 
$$\alpha_0^*=\left(\frac{\delta}{k_0c_0}\right)^{\frac{1}{k_0-1}}, \ \ \alpha_1^*=\left(\frac{\delta}{k_0c_0}\right)^{\frac{1}{k_0-1}}+\left(\frac{1-\delta}{k_1c_1}\right)^{\frac{1}{k_1-1}}.$$
\end{theorem}

A proof of the above result is provided in Appendix \ref{app:subgame-perfect-polynomial-proof}. Notice that the domain-specific performance, $\alpha_1^*$, is equal to the general-purpose performance, $\alpha_0^*$, plus a term, $(\frac{1-\delta}{k_1c_1})^{\frac{1}{k_1-1}}$, independent of the $G$'s choice over $\alpha_0^*$. This is because the cost of marginal improvements for $D$ only depends on the \textit{difference} $(\alpha_1-\alpha_0)$, and is not affected by a large or small initial investment by $G$. Though we assume, in this section, that $D$'s cost is defined solely in terms of marginal improvement, the generalizations in subsequent sections contain findings that relax this assumption.

As an immediate corollary of Theorem~\ref{subgame-perfect-eq}, we derive players' utilities as a function of $\delta$ alone.
\begin{corollary}
\label{cor:utilities-solved}
For a fixed bargaining parameter $\delta$, the players' utilities are as follows:
\begin{eqnarray}
 U_G(\delta) = \left(\frac{1}{k_0c_0}\right)^{\frac{1}{k_0-1}}\left(1-\frac{1}{k_0}\right)\delta^{\frac{k_0}{k_0-1}} + \left(\frac{1}{k_1c_1}\right)^{\frac{1}{k_1-1}}\delta(1-\delta)^{\frac{1}{k_1-1}}, \label{u-g-polynomial},\\
 U_D(\delta) = \left(\frac{1}{k_1c_1}\right)^{\frac{1}{k_1-1}}\left(1-\frac{1}{k_1}\right)(1-\delta)^{\frac{k_1}{k_1-1}} + \left(\frac{1}{k_0c_0}\right)^{\frac{1}{k_0-1}}(1-\delta)\delta^{\frac{1}{k_0-1}}. \label{u-d-polynomial}
\end{eqnarray}
\end{corollary}


\begin{figure*}[bt]
 	\centering
\hspace{0.06\textwidth}
  \hfill
\includegraphics[width=.271\textwidth]{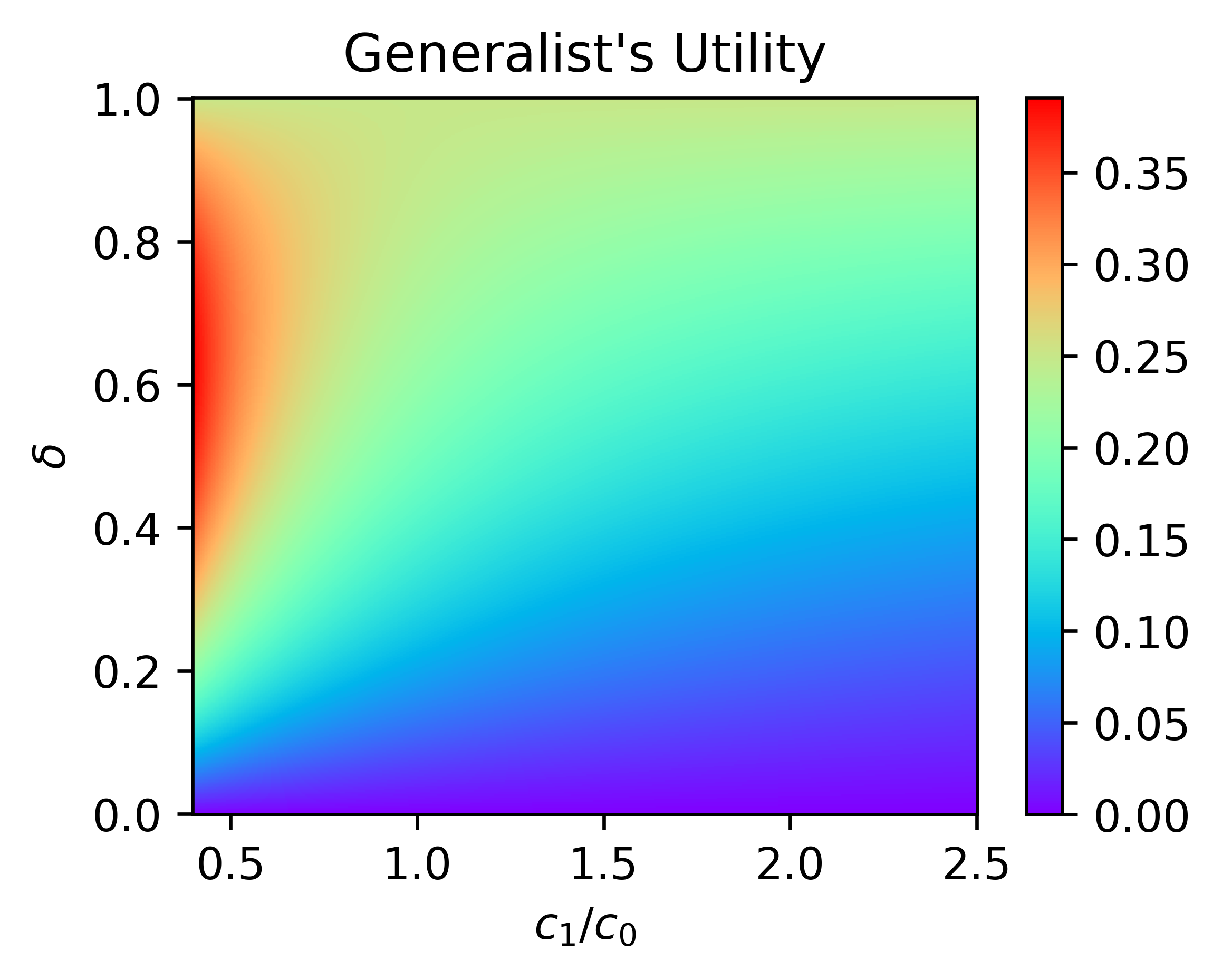}
\hfill
\includegraphics[width=.265\textwidth]{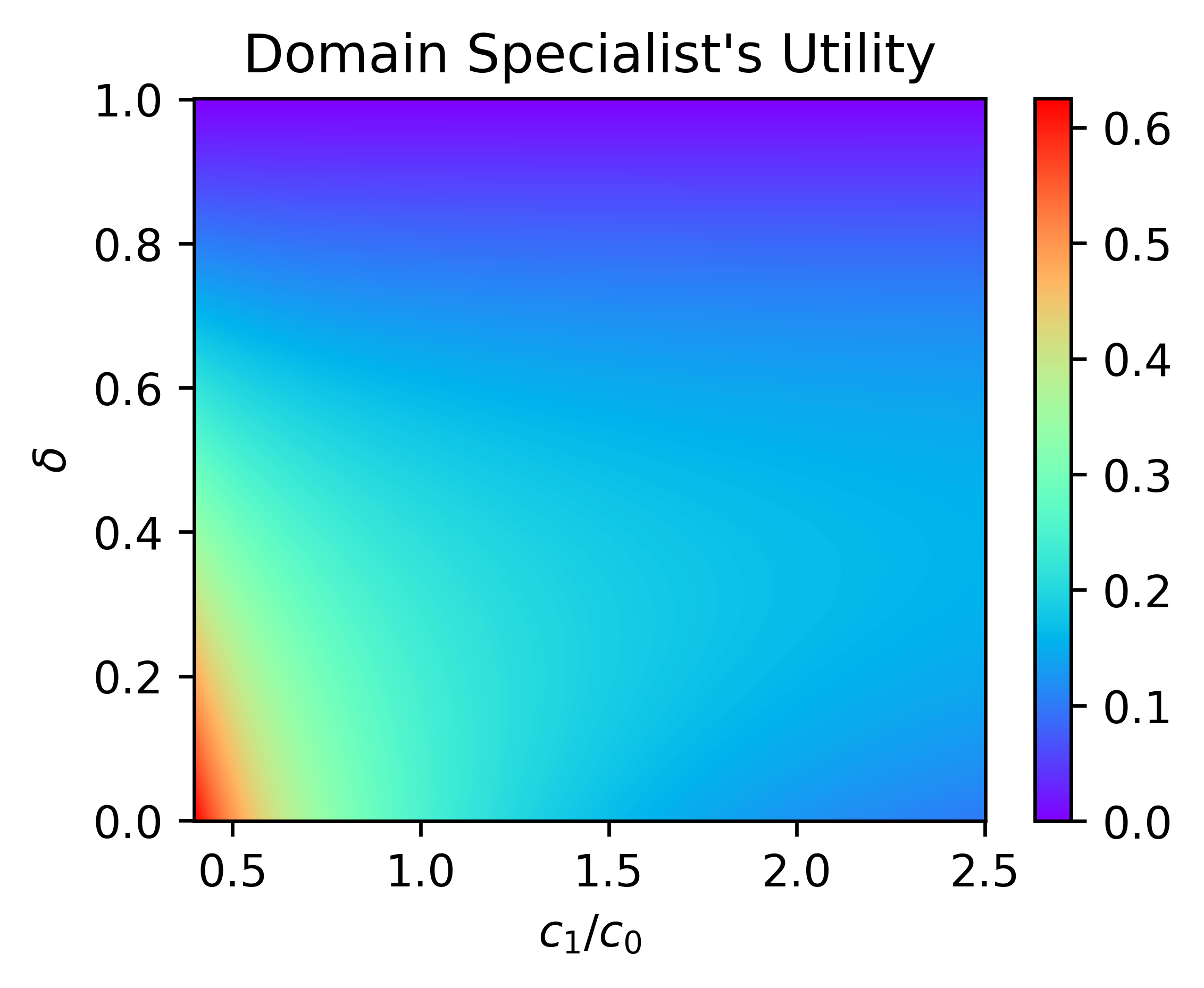}
\hfill
\includegraphics[width=.265\textwidth]{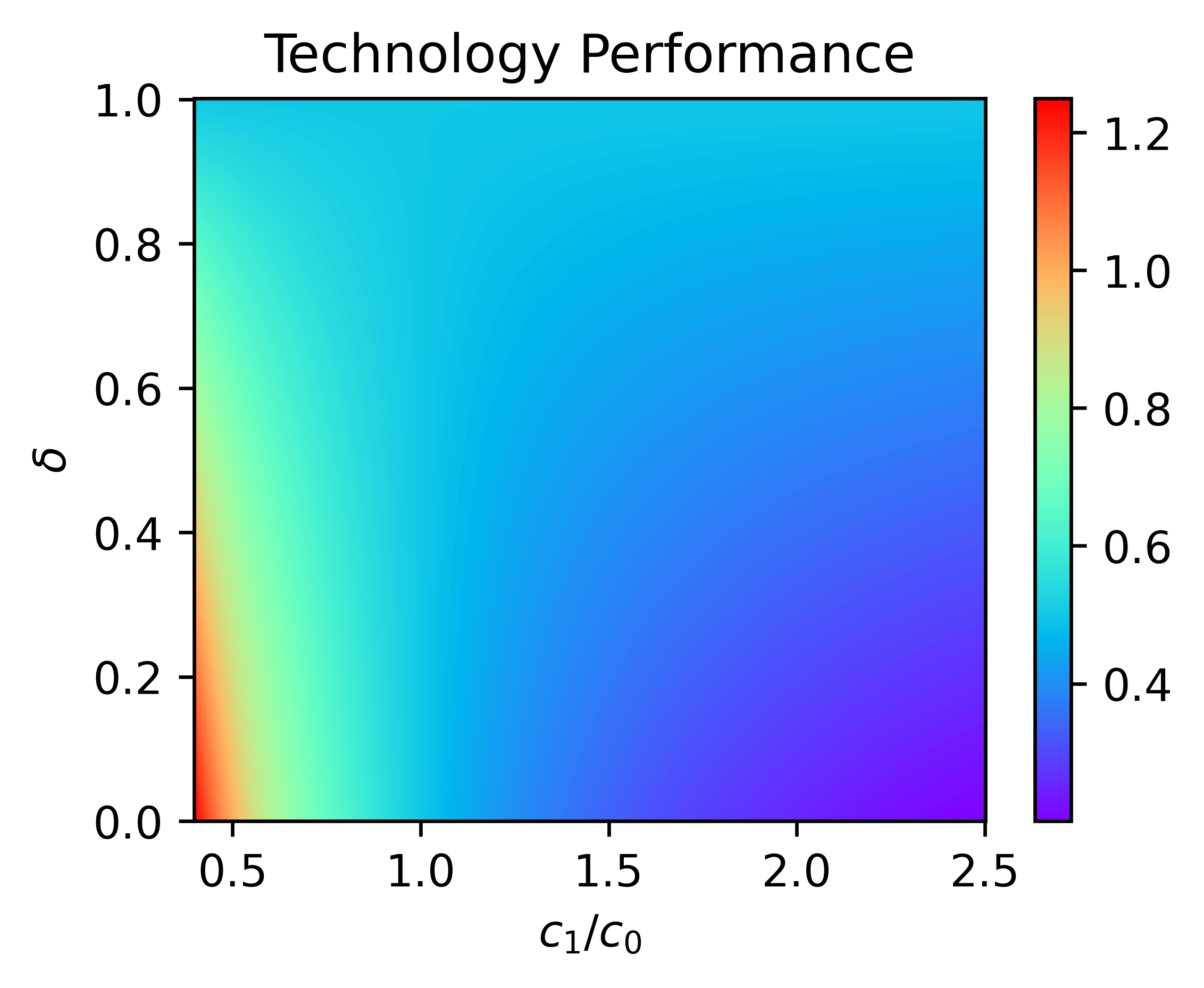}
\hfill
\hspace{0.06\textwidth}
\vspace{-.1in}
 	\caption{Utilities of the Generalist (left) and Domain Specialist (center) for different bargaining parameters and costs. The general-purpose producer prefers to share revenue when $\frac{c_1}{c_0}<1$ and the domain-specific producer prefers to share revenue when $\frac{c_1}{c_0}>1$. The resulting technology performance $\alpha_1^*$ is depicted on the right. Color bar scales are defined assuming $c_0=1$.}
 	\label{fig:capacityempiricaldpr}
 \end{figure*}

In order to determine the set of Pareto-optimal agreements, we first find that the utility functions are strictly unimodal functions of $\delta$ for all $c_0,c_1$ and $k_0,k_1\geq2$.

\begin{proposition}
\label{obs:polynomial-unimodal}
In the fine-tuning game with polynomial costs, if $k_0,k_1\geq2$, then $U_G, U_D$ are strictly unimodal functions of $\delta\in[0,1]$.
\end{proposition}

The above findings are proven in Appendix \ref{app:unimodal-one-specialist}. They suggest that the family of polynomial cost functions yields strictly unimodal utilities over $\delta$. The set of Pareto-optimal solutions to these games can therefore be identified using Theorem \ref{thm:unimodal-pareto}.
It is easy to show that the strict unimodality finding further generalizes to linear combinations of polynomial terms of the form provided in Equations (\ref{phi-0-polynomial-single}), so long as all exponents are greater than or equal to $2$. However, when the condition is not met and $k_0,k_1<2$, numerical simulations suggest that there are counter-examples to the strict unimodality property. When the strict unimodality property does not hold, it is still possible to analyze players' strategies---for example, our analysis in Section \ref{sec:multi} stands even in cases where utility functions are not unimodal in $\delta$.

Solving the powerful-$G$, powerful-$D$, vertical monopoly or other bargaining solutions consists in maximizing players' utilities either separately or combined into a joint utility. This is possible once parameters are specified; however, we cannot produce a closed-form expression for the general polynomial case because doing so would require solving for the zeroes of a polynomial of high degree. Therefore, for the remainder of this section, we will demonstrate the solution steps using parameter values $k_0,k_1 =2$. We call this the case of \textit{quadratic costs}. We choose the quadratic case for clarity and exposition, though we note that other solutions with other parameter values can be calculated using analogous steps.

\subsection{Pareto-optimal Agreements on $\delta$}
\label{subsec:pareto-quadratic}

We've derived both players' optimal strategies for fixed $\delta$. Now, we consider the process where players agree on a particular value of $\delta$. Since both players must enter an agreement in order for the technology to be viable, the determination of $\delta$ is a two-player bargaining game. We start by solving for the set of Pareto-optimal bargaining agreements, which is the interval between the `powerful-player' solutions, defined below.




\textbf{Powerful-Player Solutions.} 
As we showed in Theorem \ref{thm:unimodal-pareto}, identifying the `powerful-player' agreements is important for characterizing the set of Pareto-optimal bargaining solutions. Thus, we begin this section of analysis by solving for the powerful-$G$ and powerful-$D$ solutions (as defined in Definition \ref{def:powerfulP}).
 
\begin{proposition}[Powerful-$G$ Solution]
\label{prop:powerful-g}
    The Powerful-$G$ solution to the model fine-tuning game with quadratic costs is as follows:
    \begin{equation*}
    \delta^{\textit{Powerful }G} = 
        \begin{cases}
        \frac{c_0}{2c_0-c_1} & \textit{ for  } {c_1} < {c_0}, \\
        1 & \textit{ for  } {c_1} \geq c_0.
        \end{cases}
    \end{equation*}
\end{proposition}

\begin{proposition}[Powerful-$D$ Solution]
\label{prop:powerful-d}
    The Powerful-$D$ solution to the model fine-tuning game with quadratic costs is as follows:
    \begin{equation*}
    \delta^{\textit{Powerful }D} = 
        \begin{cases}
        0 & \textit{ for  } c_1 < c_0, \\
        \frac{c_1-c_0}{2c_1-c_0} & \textit{ for  } c_1 \geq c_0.
        \end{cases}
    \end{equation*}
\end{proposition}

The proofs for both solutions above are in Appendix \ref{subsec:powerful-g-single} and \ref{subsec:powerful-d-single}. Now, using Theorem \ref{thm:unimodal-pareto} and Proposition \ref{obs:polynomial-unimodal}, we can define the set of Pareto-optimal solutions as: $\delta^{\textit{Pareto}} \in \{\delta: \delta \leq \delta^{\textit{Powerful }G} \cap \delta \geq \delta^{\textit{Powerful }D}\}$. A visual representation of these solutions for the fine-tuning game with quadratic costs is provided in Figure \ref{fig:my_label}.

\subsection{Bargaining Solutions to Specify $\delta$}
\label{subsec:bargain-quadratic}

\begin{figure*}[bt]
  \vspace{-0.2cm}
 	\centering
\hspace{0.06\textwidth}
 		 \includegraphics[width=.265
    \textwidth]{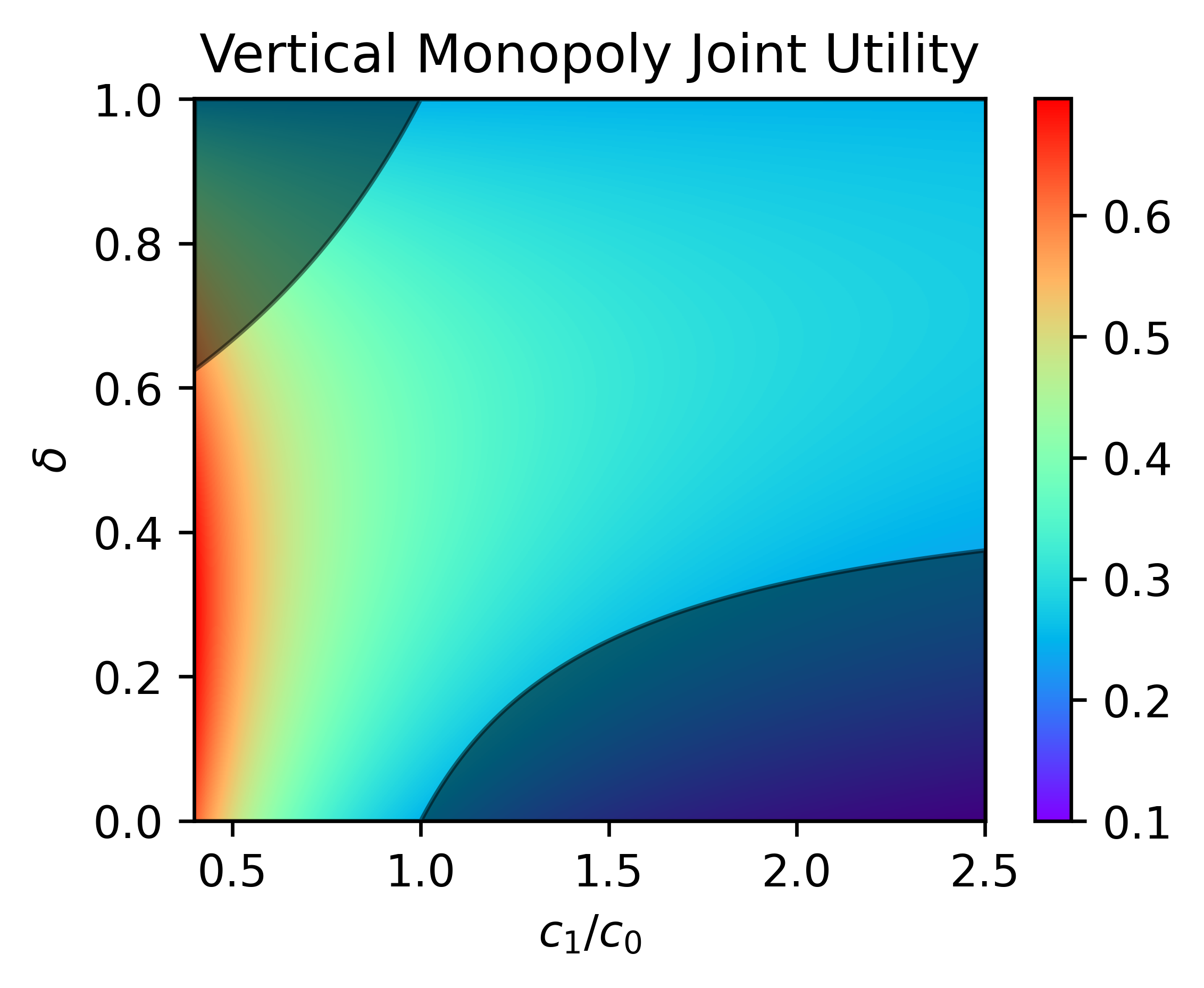}
 	\hfill
 		 \includegraphics[width=.27
    \textwidth]{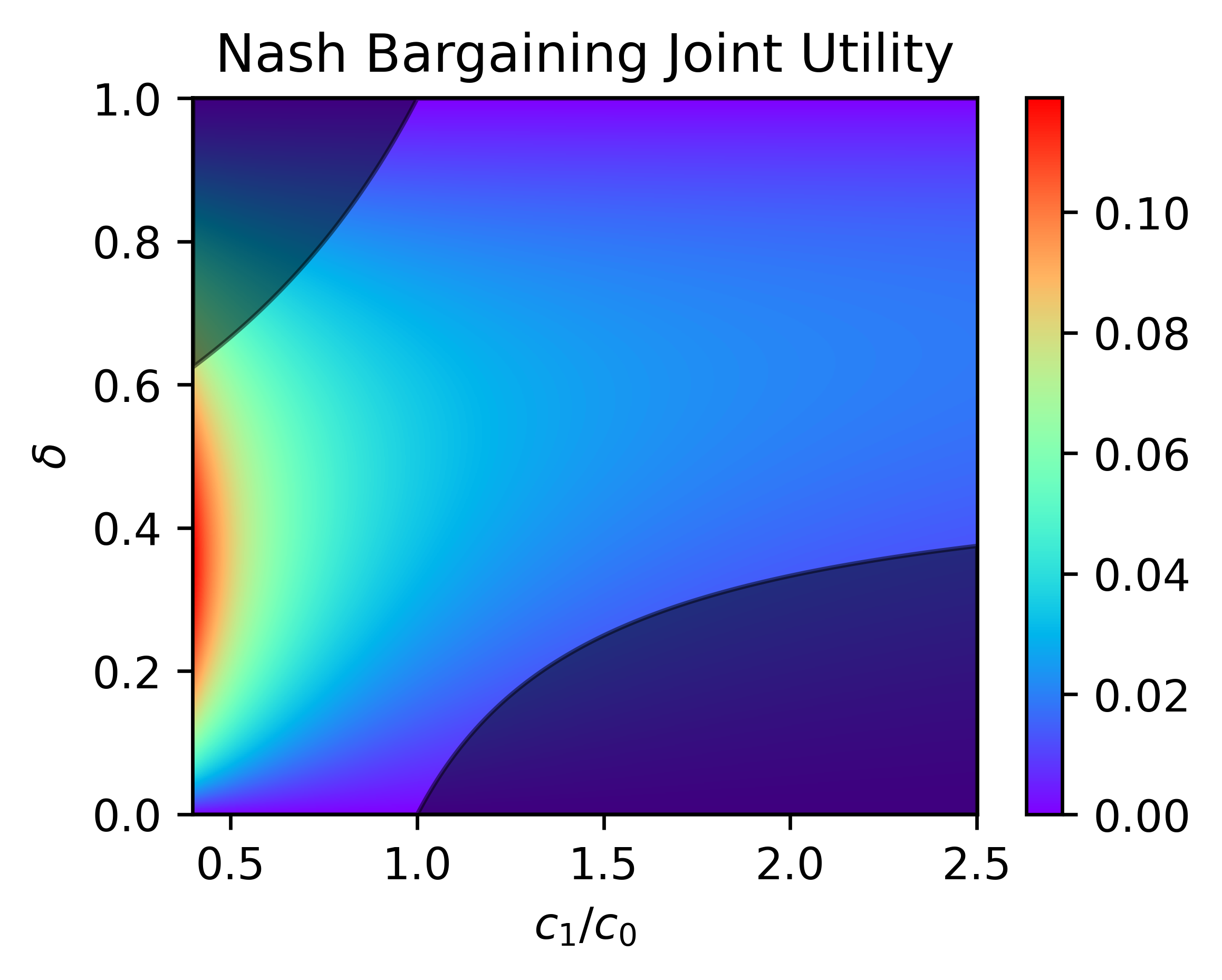}
 	\hfill
 		 \includegraphics[width=.27
    \textwidth]{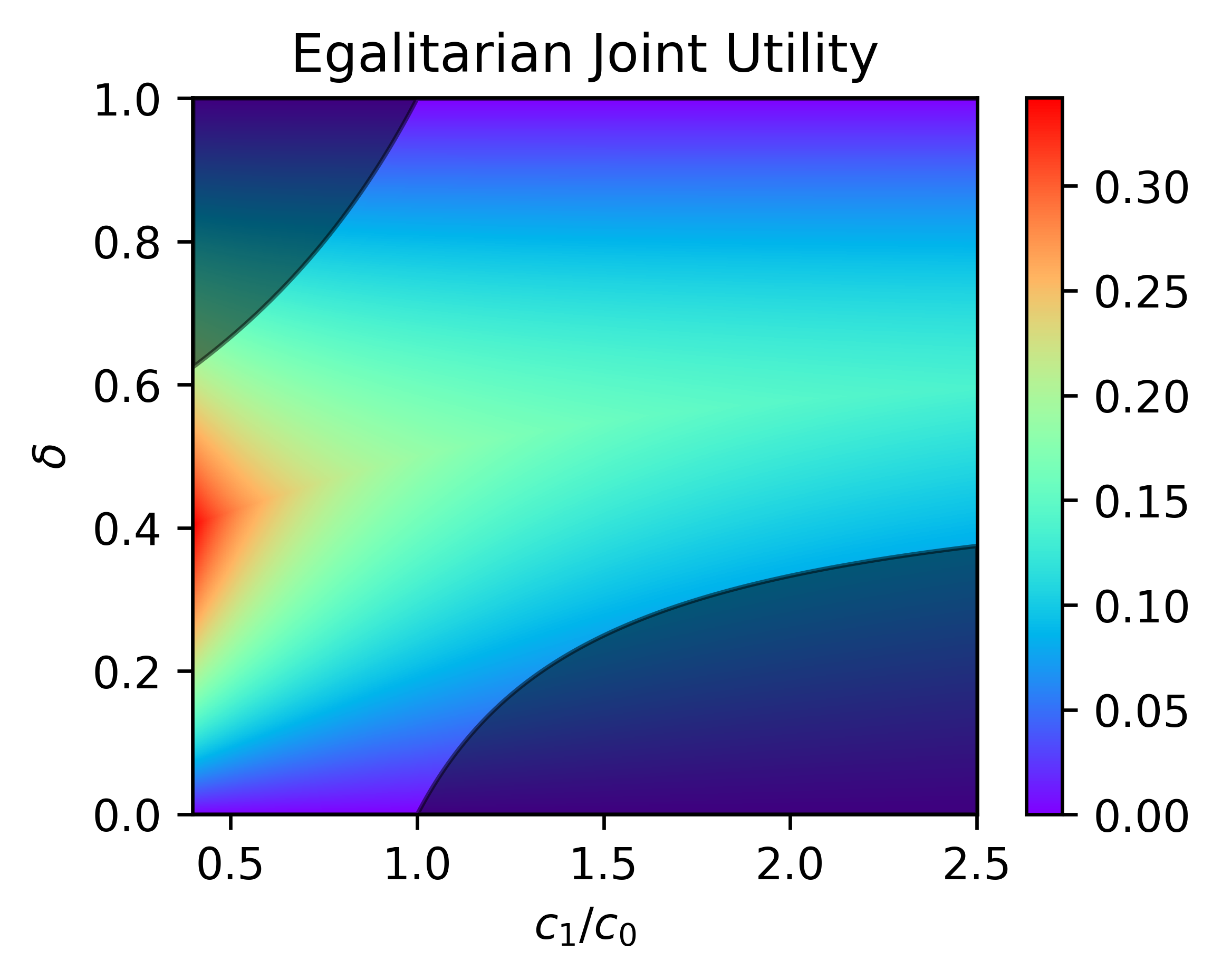}
\hspace{0.06\textwidth}
\vspace{-.1in}
 	\caption{Various joint-utility functions for finding bargaining solutions. Gray regions are $\delta$ values that are not Pareto-optimal and therefore not candidate bargaining solutions. Color bar scales are defined assuming $c_0=1$.}
 	\label{fig:joint-util}
 \end{figure*}

If neither player dominates in a bargain, how do they decide how to share surplus profit? Solutions to bargaining problems identify an agreement that maximizes some joint utility function or satisfies certain desirable properties. In this section, we define the various bargaining solutions that the two players could plausibly arrive at within the set of Pareto-optimal solutions. These solutions mostly use a joint utility function to guide the bargaining agreement, as depicted in Figure \ref{fig:joint-util}. A visual representation of the bargaining solutions is provided in Figure \ref{fig:my_label}. Definitions and closed-form solutions are provided below, and the corresponding proofs can be found in Appendix \ref{app:section-3}.

\textbf{Vertical Monopoly Solution.}
\label{subsubsec:vm}
 A perhaps intuitive approach to bargaining is to choose a revenue-sharing agreement that maximizes the sum of utilities $U_G+U_D$. This solution imagines that the two players are jointly controlled by a single entity who simply wishes to maximize the sum of utility. This solution is known as either the `vertical monopoly' solution or the `utilitarian' solution.

 \begin{definition}[Vertical Monopoly Solution]
 For the fine-tuning game, the Vertical Monopoly (or `Utilitarian') Solution is the feasible revenue-sharing agreement $\delta^{\text{VM}} \in [0,1]$ that maximizes the sum of the players' utilities:
$\delta^{\text{VM}} = \text{argmax}_{\delta \in [0,1]} \left(U_G(\delta)+U_D(\delta)\right)$.
 \end{definition}

 \begin{proposition}[Vertical Monopoly Solution]
\label{prop:vertical-monopoly-quadratic}
   The Vertical Monopoly Bargaining Solution to the fine-tuning game with quadratic costs is as follows:
    \begin{equation*}
    \delta^{\textit{Vertical Monopoly}} = \frac{c_1}{c_1+c_0}.
    \end{equation*}
\end{proposition}

\textbf{Egalitarian Bargaining Solution.}
\label{subsubsec:egal}
An alternative bargaining approach tries to help the worst-off player. This bargaining solutions is known as the `egalitarian' solution.

\begin{definition}[Egalitarian Bargaining Solution]
 For the fine-tuning game, the Egalitarian Bargaining Solution is the feasible agreement $\delta^{\text{Egal.}} \in [0,1]$ that maximizes the minimum of players' utilities:
$\delta^{\text{Egal.}} = \text{argmax}_{\delta \in [0,1]} (\text{min}_{P\in\{G,D\}}\left(U_P(\delta)\right)).$
 \end{definition}

 \begin{proposition}[Egalitarian Bargaining Solution to the fine-tuning game with quadratic costs]
\label{prop:kalai-quadratic}
    The Egalitarian Bargaining Solution to the fine-tuning game with quadratic costs is:
    $$\delta^{\textit{Egal.}} = \frac{-\sqrt{c_0^2-c_0c_1+c_1^2}-c_1+2c_0}{3(c_0-c_1)}.$$
\end{proposition}

\textbf{Nash Bargaining Solution.}
\label{subsubsec:nbs}
The Nash Bargaining solution maximizes the product between the two players' utilities. This arrangement satisfies a number of desiderata, originally laid out by \cite{nash1950bargaining}. 

\begin{definition}[Nash Bargaining Solution]
 For the fine-tuning game, the Nash Bargaining Solution is the feasible revenue-sharing agreement $\delta^{\text{NBS}} \in [0,1]$ that maximizes the product of the players' utilities:
$\delta^{\text{NBS}} = \text{argmax}_{\delta \in [0,1]} \left(U_G(\delta)*U_D(\delta)\right)$.
\end{definition}

Though a closed-form solution for quadratic functions is possible, it involves solving the roots of a cubic function and yields a solution that is clunky and uninterpretable. We refer the reader to our numerical findings on this solution, depicted in Figures \ref{fig:joint-util} and \ref{fig:my_label}.

\textbf{Kalai-Smorodinsky Bargaining Solution.}
\label{subsubsec:ks}
Another solution suggested in economic literature, known as the Kalai-Smorodinsky bargaining solution, equalizes the ratio of maximal gains. Formally: 

\begin{definition}[Kalai-Smorodinsky Bargaining Solution \citep{kalai1975other}]
 For the fine-tuning game, the Kalai-Smorodinsky Bargaining Solution (KSBS) is the feasible revenue-sharing agreement $\delta^{\text{KSBS}} \in [0,1]$ that satisfies the following relation:
$$\frac{U_G(\delta^{\text{KSBS}})}{\max_{\delta \in \delta^{\text{Pareto}}}U_G(\delta)} = \frac{U_D(\delta^{\text{KSBS}})}{\max_{\delta \in \delta^{\text{Pareto}}}U_D(\delta)}.$$
 \end{definition}

Notice the denominators in the above equation are simply the utilities associated with the powerful-G and powerful-D solutions. Despite this simplifying step, the closed form Kalai-Smorodinsky solution is clunky and uninterpretable, so we omit it from this paper. Our numerical findings on this solution are depicted in Figure \ref{fig:my_label}.

\textbf{Solution that maximizes the technology's performance.}
\label{subsubsec:max-a}
Though not, exactly, a bargaining solution, we can reason about the parameter value $\delta$ that maximizes the technology's performance, $\alpha_1^*$. There are a few ways to think of this quantity: It is the performance of the technology, and, equivalently, it is also the amount of revenue the two players collect. The maximum-$\alpha_1^*$ solution can be used as a hypothetical baseline to compare against the relative performance outcomes of other bargaining solutions.


\begin{definition}[Maximum-performance $\delta$]
For the fine-tuning game, the maximum-performance bargain is the feasible revenue-sharing agreement $\delta^{\text{max-}\alpha_1^*} \in [0,1]$ that maximizes the technology's performance $\alpha_1^*$:
$\delta^{\text{max-}\alpha_1^*} = \text{argmax}_{\delta \in [0,1]} \alpha_1^*$.
 \end{definition}

\begin{proposition}[Maximum-$\alpha_1^*$ Solution]
\label{prop:max-alpha-quadratic}
    The bargaining solution that maximizes the technology's performance is given by:
    \begin{equation*}
    \delta^{\textit{Max-}\alpha_1^*} = 
        \begin{cases}
        0 & \textit{ for  } c_1 < c_0, \\
        1 & \textit{ for  } c_1 \geq c_0.
        \end{cases}
    \end{equation*}
\end{proposition}
The maximum-performance solution  favors allocating all revenue to the player facing lower costs. Although inducing effort from the lower-cost player makes sense, it is perhaps counter-intuitive that this solution does not favor revenue-sharing between players as the other bargaining solutions do. Indeed, further analysis reveals that this strict stepwise result is unique to the quadratic case where $k_0=k_1=2$. In polynomial games for which $k>2$,  revenue-sharing yields a higher-performing technology, though most revenue is still allocated to the player with lower costs. 

\begin{figure}
    \centering
    \includegraphics[width=.7\linewidth]{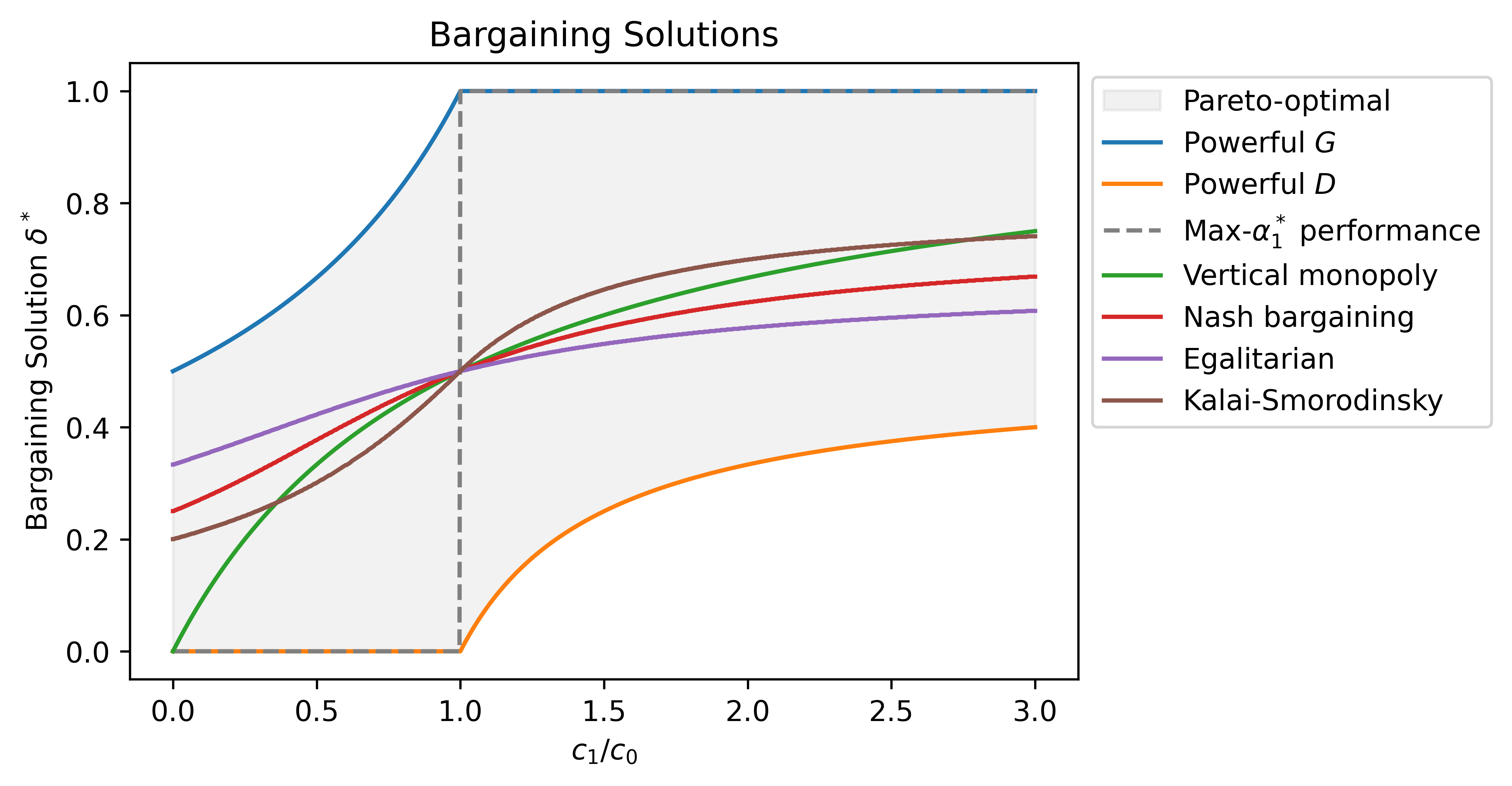}
    \caption{Bargaining agreements for the fine-tuning game with quadratic costs. 
    Most bargaining solutions involve revenue sharing, even when one player faces much higher costs. 
    }
    \label{fig:my_label}
\end{figure}

\subsection{Discussion on Bargaining Solutions}

Above we solve for a number of bargaining solutions revealing different possible configurations of fine-tuning arrangements. The general technology-producer and the domain specialist each have different optimal arrangements, between which any agreement is Pareto-optimal in the case of polynomial costs. 

The first notable take-away is that players do not necessarily opt to maximize their own proportion of the profit. Even if one player has full control over the bargaining solution, depending on the relative cost of production, they may benefit from a profit-sharing agreement in order to encourage investment by the other player. If bargaining is conceptualized as splitting a pie, one player prefers to cede some portion of the pie if it means the entire pie grows to a size that justifies profit-sharing. This phenomenon arises in real-world settings. For instance, Apple allows third party developers to build software on iPhones. Opening up the tasks of application development to third parties improves consumer experience such that consumers are willing to purchase apps or other capabilities within apps. This additional revenue is then shared between Apple and the developer, leaving Apple with higher profits and a better product. Revenue sharing arises, often, because doing so is lucrative.

Profit-sharing is present even when both players have exceedingly different costs of production (i.e., when $c_1/c_0$ approaches $0$ or $\infty$). In these limiting instances, we find that the Nash bargaining solution, Kalai-Smorodinski, and Egalitarian solutions all suggest profit-sharing. Only the Utilitarian solution---which models the two players as a vertical monopoly that is centrally controlled---yields the intuitively performance-optimal bargain, where the player with lower costs receives the entire profit. However, the vertical monopoly solution is not always performance-optimal. It underperforms the KSBS when the players face similar costs ($\sim0.5 < \frac{c_1}{c_0}< 2.5$).

The bargaining solutions are neither binding rules nor descriptive observations; instead, they can be thought of as normative prescriptions. Identifying joint utility functions can help guide agents towards decisions that serve collective interests. For example, utilitarian and egalitarian solutions offer different visions for the appropriate distribution of welfare. In the same vein, one could specify and commit to a \textit{social welfare function} in order to identify a bargaining solution that might be referred to as `socially optimal.' Unsurprisingly, however, specifying social interests in a single function is an ambitious undertaking. In our present case, a social welfare function would need to balance the interests of (at least) 1) the technology's producers 2) consumers who value performance and 3) other external stakeholders. The procedure demonstrated in this section provides a road map for a social welfare analysis of the deployment of general-purpose models. Such an analysis might uncover how fine-tuning processes can be configured to serve collective, societal interests. 
\section{Generalization to Multiple Domain Specialists}
\label{sec:multi}

So far, we've modeled the fine-tuning process as a two-player game between a generalist and a specialist. However, an important feature of general-purpose AI models is that they can be developed for a wide swath of downstream use-cases. To capture the possibly many uses for general-purpose models, in this section, we generalize our model to the case where $n\geq 1$ domain specialists adapt the technology.

\textbf{The multi-specialist fine-tuning game.} Consider a game with $n \geq 1$ specialists. The players are $G$, $D_1$, $D_2$, ... $D_n$ and we use $i$ to index the specialists. $G$ develops a technology to general performance $\alpha_0$, after which every domain specialist $D_i$ invests in the technology, bringing it to performance $\alpha_i$ in their domain. $G$ and $D_i$ share revenue $r_i(\alpha_i)$ according to bargaining parameter $\delta_i \in [0,1]$. At the end of the game, $G$ receives $\sum_i \delta_i r_i(\alpha_i)$ and each specialist $D_i$ receives $(1-\delta_i)r_i(\alpha_i)$. The game involves the following steps:
\begin{enumerate}
    \item Players bargain to decide $\delta_i$ for every domain $i$.
    \item $G$ invests in a general-purpose technology yielding performance level $\alpha_0$ and subject to cost $\phi_0(\alpha_0)$.
    \item Each specialist $D_i$ may fine-tune the technology by choosing a performance level $\alpha_i$ subject to cost $\phi_i(\alpha_i;\alpha_0)$.
\end{enumerate}
Players' utilities are defined as their revenue share minus cost:
\begin{equation}
    U_G({\delta}) := \sum_{i}\delta_i r_i(\alpha_i) - \phi_0(\alpha_0),
\end{equation}
\begin{equation}
    U_{D_i}(\delta) :=(1-\delta_i) r_i(\alpha_i) - \phi_i(\alpha_i;\alpha_0).
\end{equation}
If $G$ does not agree to a feasible bargain, she can instead opt for \textit{disagreement}, where $G$ receives utility $d_0$ and every specialist $D_i$ receives utility $d_i$. If any particular domain specialist $D_i$ does not agree to a feasible bargain, they may opt to receive $d_i$. However, this does not preclude other specialists from reaching a deal or adapting the technology. We assume, unless otherwise specified, that the disagreement scenario is described by $d_0=d_i=0$ for all $i$.

There are two possible ways to generalize the fine-tuning game to multiple specialists. The gameplay described above allows a separate bargain for every domain -- this might correspond to a scenario where the price for a model may vary depending on the user. An alternative generalization would restrict the pricing structure so that \textit{all} players must come to a \textit{single} bargaining parameter. Notice the multi-specialist fine-tuning game we define above is general enough that the single-bargain game is always a special case where $\delta_i=\delta \forall i$. We will demonstrate that both iterations of the multi-specialist game are manageable. We begin by demonstrating that the prior closed-form solutions are indeed attainable for multiple players before turning to general results on different sorts of domain-specialist strategies.




\subsection{Pareto-Optimal Bargains}

Characterizing the Pareto-optimal bargains is somewhat more challenging with multiple players. Below, we offer methods for identifying the Pareto set, starting in the case where the players bargain over a \textit{single} parameter $\delta$, and then describing the more general case where multiple bargaining values are feasible.

\textbf{Closed-form characterization in the case of a single bargaining parameter.} What is the Pareto set when every revenue sharing agreement is determined by a single parameter $\delta \in \mathbb{R}^1$?
First we show that even in cases where all players' utility functions are strictly unimodal, the set of Pareto-optimal bargains is no longer necessarily a single interval---as was the case with one specialist. Generalizing our findings about the set of Pareto-optimal bargaining solutions to multiple-specialists requires additional steps to exclude solutions that are Pareto-dominated within the region between players' optima.


Below we state the theorem that generalizes \ref{thm:unimodal-pareto} for the multi-player fine-tuning game.

\begin{theorem}
\label{thm:multi-specialist-pareto}
Consider the multi-player fine-tuning game where players' utility functions are strictly unimodal over the bargaining parameter $\delta$. The Pareto set of bargaining agreements $\delta^{\textit{Pareto}}$ is the feasible region consisting of all $\delta$ values in $[0,1]$ excluding:
\begin{itemize}
    \item all values of $\delta$ for which $U_G(\delta)<0$.
    \item all values of $\delta$ for which $U_G(\delta)$ is positive and increasing and no values of $U_{D_i}(\delta)$ are positive and decreasing $\forall i$.
    \item all values of $\delta$ for which $U_G(\delta)$ is positive and decreasing and no values of $U_{D_i}(\delta)$ are positive and increasing $\forall i$.
\end{itemize}
Disagreement is Pareto-optimal, if and only if all values of $\delta\in[0,1]$ are excluded through the above criteria.
\end{theorem}

A proof of the above theorem is provided in Appendix \ref{app:multi-pareto}. Figure \ref{fig:multi-specialist-pareto} provides an illustration of a 2-specialist fine-tuning game where the Pareto-optimal region is not a single interval but instead consists of two disconnected intervals. The finding highlights the added complexity of the multi-specialist fine-tuning game, where individual specialists can abstain but a deal can still be reached among other players. 

\begin{figure}
    \centering
    \includegraphics[width=0.46\linewidth]{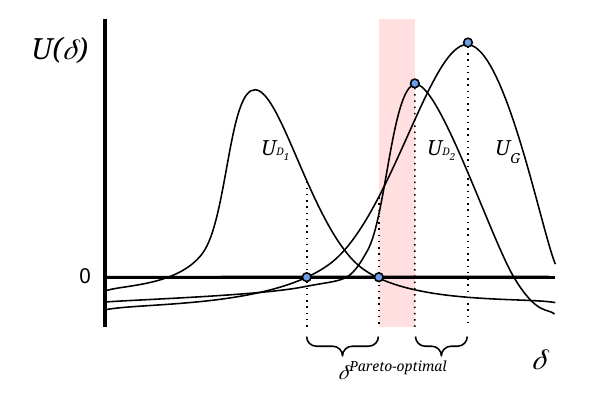}
    \caption{Illustration of Theorem \ref{thm:multi-specialist-pareto}. When there are multiple specialists, the set of Pareto-optimal bargaining agreements is not, necessarily, a single interval, even when utilities are strictly unimodal. In this illustrated example, the region in red represents values of $\delta$ for which all relevant players' utilities are increasing. In the red region, $D_1$ would abstain from the agreement. Any value of $\delta$ in this region is Pareto-dominated by $\delta^{\text{Powerful }D_2}$.}
    \label{fig:multi-specialist-pareto}
\end{figure}

\textbf{Numerical approach in the case of multiple bargaining parameters.}
Notice that the multi-specialist fine-tuning game no longer involves a bargain over a single parameter in general. Instead, the bargaining space may be represented by a vector, $\Vec{\delta}\in \mathbb{R}^n$, where $\Vec{\delta}:=[\delta_1,...,\delta_n]$. Thus, instead of a particular point value on an interval, a bargaining outcome can be represented geometrically as a point in $n$-dimensional space, where the bargaining solutions represent the coordinates. Where, in the one-specialist game, the Pareto set could be expressed as an \textit{interval} of values within $\delta \in [0,1]$, here, the Pareto Set may be enclosed by a more complex shape in higher dimensions. 
Observe that in the fine-tuning game with $n$ specialists, the Pareto-optimal solution can be phrased as the set of solutions to the following weighted optimization problem:
\begin{eqnarray*}
\max_{\Vec{\delta}} & w_0 U_G + \sum_{i=1}^n w_iU_{D_i}
\\
s.t. & \sum_{i=0}^n  w_i = 1 \\ 
\ & w_0,w_1,...,w_n \geq 0
\end{eqnarray*}

Every Pareto-optimal solution corresponds to a point on the standard $n$-simplex in $\mathbb{R}^{n+1}$ representing a weighting of each of the utility functions. The Pareto region can thus be identified numerically by computing the optimal $\vec\delta^*$ over the enumerated range of possible values of $\Vec{w}$ and $(\delta_1,...,\delta_n)$. We refer the reader to Figure \ref{fig:multi-specialist-bargain} for a visualization of the numerically-identified region of Pareto-optimal bargaining agreements.

Note that this technique does not rely on any assumptions about the shape or formal properties of players’ utility functions, and therefore suggests a broad approach for arriving at Pareto-optimal bargaining solutions to fine-tuning games, though in high dimensions (i.e., with many players) the computational burden will grow exponentially. 

While the procedure described above uses a numerical approach, in the remainder of this section, we use the specific case of quadratic cost functions to demonstrate that closed-form solutions are indeed attainable for player strategies and bargaining solutions.

\subsection{Closed-form Strategies and Bargaining Solutions} 

For the multi-specialist game, we demonstrate our results on a set of polynomial cost functions, defined below for $c_0,c_i>0$ and $k_0,k_i \geq 2$:
\begin{equation}
\label{eq:phi-0-multi}
    \phi_0(\alpha_0) := c_0\alpha_0^{k_0},
\end{equation}
\begin{equation}
\label{eq:phi-i-multi}
    \phi_i(\alpha_i;\alpha_0) := c_i(\alpha_i-\alpha_0)^{k_i}.
\end{equation}

The equilibria and bargaining solutions for this set of cost functions follow a similar set of steps to those demonstrated in the one-specialist fine-tuning game (Section \ref{sec:quadratic}). We include the full set of results in Table \ref{tab:multi-specialist-solutions}. We refer the reader to Appendix \ref{app:multi-solutions} for the full description of steps and results. The results suggest that meaningful bargaining solutions are possible for a wide set of potential domain specialties, so long as costs start out reasonably low.

\begin{table}
\centering
\begin{tabular}{l|l|l|l}
\hline
\hline
\textbf{} & \textbf{Single specialist} & \textbf{Single bargain} & \textbf{Multi-bargain} \\ \hline
$G$ Strategy  & \cellcolor{green!20}Theorem \ref{subgame-perfect-eq} & \cellcolor{green!20}Theorem \ref{thm:multi-specialist-equilibrium} & \cellcolor{green!20}Theorem \ref{thm:multi-multi-subgame-perfect-eq} \\ 
$D_i$ Strategy  & \cellcolor{green!20}Theorem \ref{subgame-perfect-eq} & \cellcolor{green!20}Theorem \ref{thm:multi-specialist-equilibrium} & \cellcolor{green!20}Theorem \ref{thm:multi-multi-subgame-perfect-eq} \\ \hline
$G$ Utility   & \cellcolor{green!20}Corollary \ref{cor:utilities-solved} & \cellcolor{green!20}Corollary \ref{cor:multi-specialist-utilities-solved} & \cellcolor{green!20}Corollary \ref{cor:multi-multi-utilities-solved} \\ 
$D_i$ Utility   & \cellcolor{green!20}Corollary \ref{cor:utilities-solved} & \cellcolor{green!20}Corollary \ref{cor:multi-specialist-utilities-solved} & \cellcolor{green!20}Corollary \ref{cor:multi-multi-utilities-solved} \\ \hline
Powerful $G$         & \cellcolor{green!20}Proposition \ref{prop:powerful-g} & \cellcolor{green!20}Proposition \ref{prop:powerful-g-multi} & \cellcolor{green!20}Proposition \ref{prop:multi-multi-powerful-g} \\ 
Powerful $D_i$       & \cellcolor{green!20}Proposition \ref{prop:powerful-d} & \cellcolor{green!20}Proposition \ref{prop:powerful-d-i-multi} & \cellcolor{green!20}Proposition \ref{prop:multi-multi-powerful-d} \\ 
VM (Utilitarian)     & \cellcolor{green!20}Proposition \ref{prop:vertical-monopoly-quadratic} & \cellcolor{green!20}Proposition \ref{prop:vertical-monopoly-multi-specialist} & \cellcolor{green!20}Proposition \ref{prop:vertical-monopoly-multi-multi} \\ 
Egalitarian          & \cellcolor{green!20}Proposition \ref{prop:kalai-quadratic}& \cellcolor{green!20}Proposition \ref{prop:egal-multi-specialist} & \cellcolor{yellow!20}Numerical   \\ 
Nash                 &  \cellcolor{yellow!20}Numerical   & \cellcolor{yellow!20}Numerical   & \cellcolor{yellow!20}Numerical   \\ 
Kalai-Smorodinsky    & \cellcolor{yellow!20}Numerical   & \cellcolor{yellow!20}Numerical   & \cellcolor{yellow!20}Numerical   \\ 
\hline
\hline
\end{tabular}
\caption{Summary of closed-form and numerical results. Each cell shaded in green corresponds to a closed-form solution, while each cell shaded yellow corresponds to numerical results. Though the equilibria strategies and utility functions are solved for general polynomial costs, the bargaining solutions are found in closed-form for cases where $k_0=k_i=2$.}
\label{tab:multi-specialist-solutions}
\end{table}

\begin{figure}
    \centering
    \includegraphics[width=.585\linewidth]{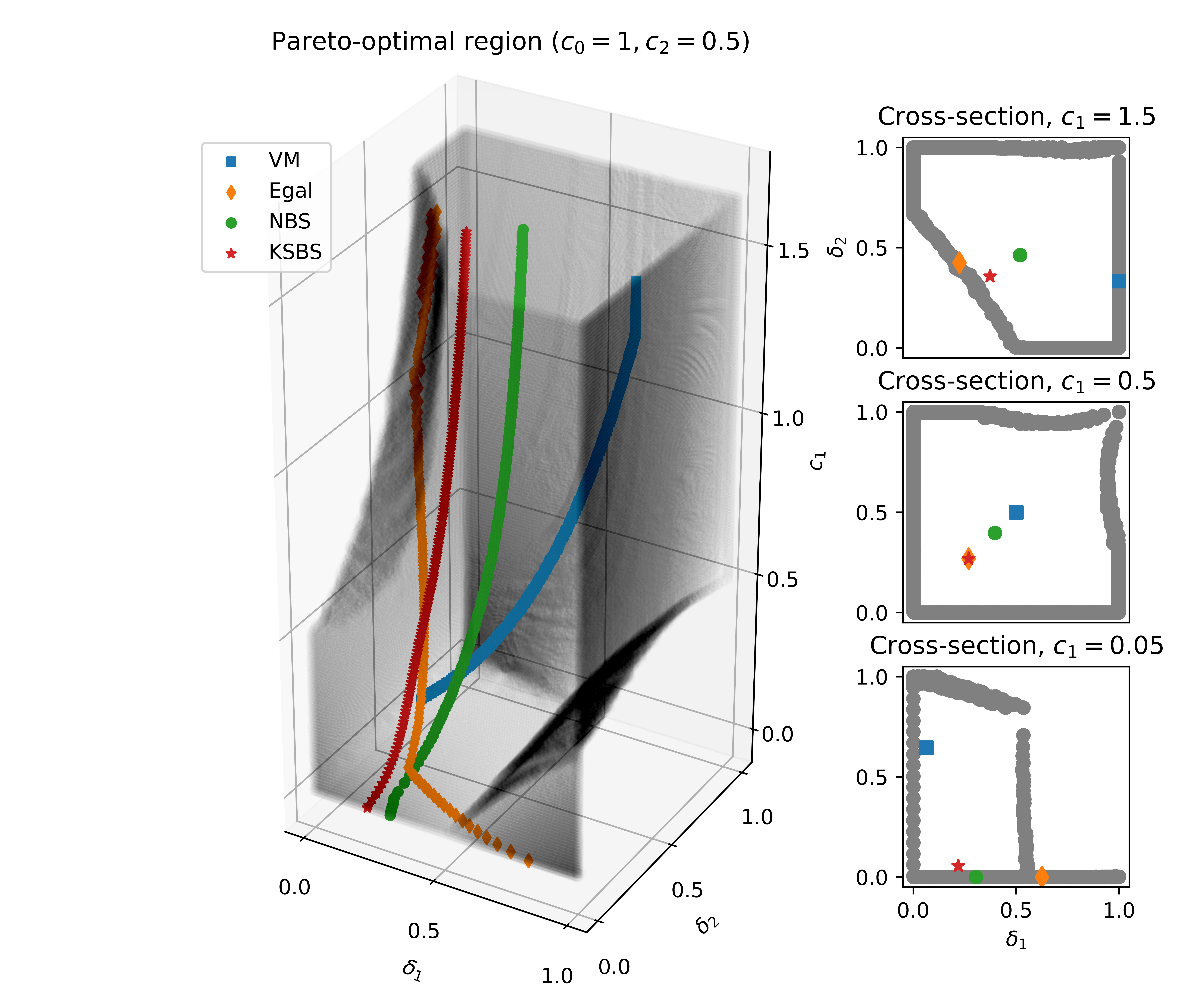}
    \vspace{-5mm}
    \caption{Pareto-optimal region and bargaining solution in a fine-tuning game with two specialists. In this example, the costs are quadratic, and the parameters $c_0=1, c_2=0.5$, and the cost parameter $c_1$ is varied over the interval $(0,1.6)$. The shape in gray bounding the set of bargaining solutions is the region of Pareto-optimal bargains, and the points plotted represent the Veritcal Monopoly (VM), Egalitarian (Egal), Nash (NBS), and Kalai-Smorodinsky (KSBS) bargaining solutions. This figure can be thought of as an extension of Figure \ref{fig:my_label} for the multi-specialist case. Note that the Egalitarian solution resides on the $\delta_2=0$ boundary for low values of $c_1$, and the VM solution resides on the $\delta_1=1$ boundary for high values of $c_1$.}
    \label{fig:multi-specialist-bargain}
\end{figure}

The slate of closed-form and numerical solutions are represented visually in Figure \ref{fig:multi-specialist-bargain}. The figure represents one example set of cost parameters to demonstrate that bargaining solutions can be identified numerically. In this example, we observe that the Pareto-optimal bargains are a connected region in $n$-dimensional space, and that bargaining solutions are contained within the Pareto set. Recall that, in the game with one specialist, all bargaining solutions were contained within $(0,1)$ for all positive finite cost parameters $c_0,c_1$. Now with multiple domain specialists, bargaining solutions including the Vertical Monopoly and Egalitarian solutions yield agreements where some players end up with all or nothing. Namely, the Vertical Monopoly solution arrives at certain agreements for which a specialist receives nothing, and the Egalitarian solution arrives at certain agreements for which the generalist receives nothing.

\subsection{Domain Specialists' Equilibrium Strategies}
\label{subsec:specialist-regimes}

When there are potentially many domains where a technology may prove useful or marketable, different strategies around investment levels and fine-tuning can arise. In some domains, a technology may be adopted `as-is' without significant additional investment or specialization. In other domains, it might be in everyone's interest for a technology to receive significant investment and specialization. Of course, in other domains, a technology might not be viable for any use at all. 
In this section, we explore the different sorts of cooperation (or non-cooperation) that can arise in domains with different characteristics. Our next general finding is a theorem on the various regimes of domain specialist strategies, depending on certain attributes of revenue and cost functions. 

First, we will offer a set of relevant definitions to help characterize the different possible regimes of strategies for the specialist. Then, we will state the formal theorem.

\begin{definition}[Contributor]
    A domain specialist $D_i$ is a \textbf{contributor} at the profit-sharing agreement $\delta_i$ if, given the generalist's optimal investment $\alpha_0$ at $\delta_i$, $D_i$'s optimal strategy is to bring the technology to performance $\alpha_i^*>\alpha_0$. 
\end{definition}
\begin{definition}[Free-rider]
    A  domain specialist $D_i$ is a \textbf{free-rider} at the profit-sharing agreement $\delta_i$ if, given the generalist's optimal investment $\alpha_0$ at $\delta_i$, $D_i$'s optimal strategy is to enter the deal without improving the technology's performance, so $\alpha_i^*=\alpha_0$.
\end{definition}
\begin{definition}[Abstainer]
    A domain specialist $D_i$ is an \textbf{abstainer} at the profit-sharing agreement $\delta$ if, given the generalist's optimal investment $\alpha_0$ at $\delta$, $D_i$'s optimal strategy is to exit the deal and opt for disagreement. 
\end{definition}

Notice that any specialist is inevitably either a contributor, a free-rider, or an abstainer. These three regimes span the possible strategies for $D_i$. Below, we outline conditions that characterize $D_i$'s strategy depending on their domain's cost and revenue $\{r_i,\phi_i\}$.

\begin{theorem}
    Suppose $G$ has produced a general-purpose technology operating at performance $\alpha_0$ and available at profit-sharing parameter $\delta_i$. For any specialist with utility unimodal in $\alpha_i$, the following conditions characterize their strategy, as shown in Table \ref{tab:specialist-regimes}.

    \begin{itemize}
        \item \textit{``Fixed Costs Under Control'' (FCUC)}: At zero investment ($\alpha_i=\alpha_0$), the domain specialist $i$'s cost is less than its share of the revenue. Formally, $r_i(\alpha_0)>\frac{1}{1-\delta_i}\phi_i(\alpha_0; \alpha_0)$.
        \item \textit{``Marginally Profitable Investment'' (MPI)}: At zero investment ($\alpha_i=\alpha_0$), a marginal investment from the domain specialist $i$ increases its revenue share more than its costs. Formally, $r_i'(\alpha_0) > \frac{1}{1-\delta_i}\phi_i'(\alpha_0; \alpha_0)$.
    \end{itemize}
    \vspace{-1mm}
    
    \begin{table}[h]
    \centering
    \begin{tabular}{c|c|c}
    \hline
    \hline
       ``Fixed Costs Under Control'' & ``Marginally Profitable Investment''   & Type of \\
       $r_i(\alpha_0)>\frac{1}{1-\delta}\phi_i(\alpha_0)$ & $r_i'(\alpha_0) > \frac{1}{1-\delta}\phi_i'(\alpha_0)$  & Specialist \\
       \hline
       T & T & Contributor \\
       T & F & Free-rider \\
       F & T & Contributor or Abstainer* \\
       F & F & Abstainer\\
       \hline
       \hline
    \end{tabular}
    \caption{Types of specialists. In the third case (*), marginal conditions alone do not determine whether the specialist contributes or abstains.}
    \label{tab:specialist-regimes}
\end{table}
    \label{thm:specialist-regimes}
\end{theorem}
A proof of the above theorem is provided in Appendix \ref{app:regimes-proof}. The requirement that specialist utility is unimodal in $\alpha_i$ is, in our view, quite natural. It covers three possible scenarios: 1) utility is increasing with investment, 2) utility is decreasing with investment, or 3) utility increases with investment up to a certain point, beyond which any further investment is not cost-justified.

It is important to note that the three regimes defined in this section can describe a specialist's strategy in either the 1-specialist or multi-specialist fine-tuning game. In the 1-specialist case, the potential strategies describe counterfactual outcomes that depend on the particular cost and revenue functions of the specialist. In the multi-specialist game, the strategies are ways of grouping the domains and all can exist simultaneously.

One scenario portrayed in Table \ref{tab:specialist-regimes} does not determine cleanly which regime the specialist falls into. In the scenario labeled with an asterisk (*), fixed costs are not under control but it is marginally profitable to invest in the technology. At zero investment, the technology is not ready to bring to market profitably, and it is unclear only from the marginal return on an initial investment whether it is worthwhile for the specialist to invest. In this case, the technology is potentially viable with some non-zero effort or, alternatively, not viable for the domain at any level of investment. Though the marginal conditions do not tell us whether the specialist would contribute or abstain, we can identify their strategy as follows: If $(1-\delta_i)r_i(\alpha_i)-\phi_i(\alpha_i)$ has positive real roots (for values of $\alpha_i$ greater than $\alpha_0$), then $D_i$ would contribute. Otherwise, $D_i$ would abstain.

\begin{figure}
    \centering
    \includegraphics[width=.6\linewidth]{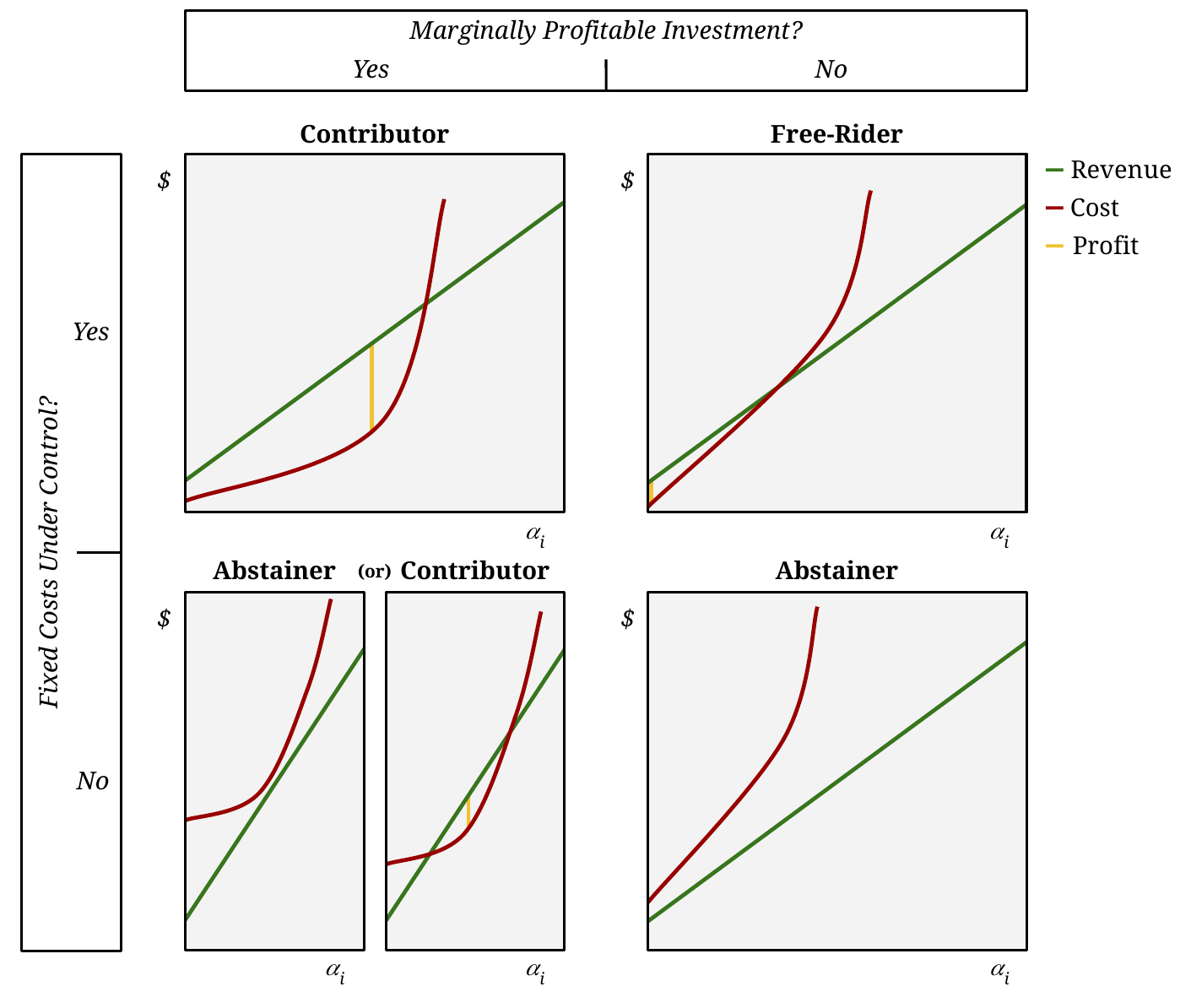}
    \vspace{-5mm}
    \caption{Examples illustrating Theorem \ref{thm:specialist-regimes}. Depending on the characteristics of the cost and revenue curves, a domain specialist might engage in different types of strategy. For instance, when fixed costs are under control but investment is not marginally profitable (upper right quadrant), the firm will free-ride. When fixed costs are too high but investment yields marginal returns (lower left), the firm either abstains, or contributes if revenue exceeds cost at any level of investment.}
    \label{fig:specialist-regimes}
    \vspace{-4mm}
\end{figure}

An illustration of Theorem \ref{thm:specialist-regimes} is provided in Figure \ref{fig:specialist-regimes}. 
A noteworthy feature of this result is that it allows us to identify particular adaptation strategies using only the attributes about the domain $i$ around $\alpha_i=\alpha_0$. In this setting, much of the information about the viability of a technology can be learned from only the 0th- and 1st-order approximations of $U_{D_i}\big\vert_{\alpha_i=\alpha_0}$, when the domain has invested minimal effort in the technology. This result perhaps coheres with the belief that a `minimum viable product' (MVP) can provide an important signal about the profitability of a technology \cite{blank2018lean, ries2011lean}. 

Our analysis helps explain why technologies see significant uptake in some domains and not others. It characterizes domains that are particularly suitable or unsuitable to adopt a general-purpose technology. It also may explain why some technologies are re-sold without additional investment while others require fine-tuning.

\section{Discussion}
\label{sec:discussion}

Our theoretical model of fine-tuning provides a way of understanding the process whereby firms adapt a general-purpose technology to a particular domain. Notably, the model highlights some of the \textit{normative} considerations in these arrangements. There are many potential mechanisms for distributing surplus between technological developers. Different bargaining solutions can highlight the ways that developers share upsides, with implications for a technology's performance (i.e., the value provided to the consumer). 

Our model provides a starting point for considering the different interests and choices involved in the development of general-purpose models. By putting forward this model, we attempt to invoke the political economy of the development of general AI technologies. These technologies are produced by a number of entities with different interests, and may potentially affect many individuals. This paper models agents' different interests explicitly, and proposes different methods for weighing between them in light of societal values. 

Further, our model may prove useful in suggesting economic and incentive-based levers that might be available to the regulatory process. Interventions and incentives could help align the production of general-purpose models with social welfare.

\subsection{Conclusion} 

In this work, we introduced the notion of a {\em fine-tuning game} as a way of modeling the economic interaction between a firm creating a general-purpose technology and one or more additional firms that specialize the technology to specific domains. We have shown how this formalism leads to an interesting class of bargaining games with rich structure, with insights into how the relative capabilities of the different parties translate into the effort invested in development and the eventual performance of the technology. 

The work suggests a number of interesting directions for further research. One direction is to identify further general existence results for bargaining solutions with general functions in this model. More broadly, we also believe that formalizing the societal interests involved in AI regulation is an important direction; such a formalism would need to build on an underlying model that contains the economic interests of the firms producing the AI technology. Our model may therefore help form the foundation for such work. 
\section{Acknowledgements}



The authors would like to thank the members of the AI, Policy and Practice group (AIPP) at Cornell University, the Center for Data Science (CDS) at NYU, the Center for Discrete Mathematics and Theoretical Computer Science (DIMACS) at Rutgers University, the Digital Life Initiative (DLI) at Cornell Tech, the EconCS group at Harvard University, and the Human + Machine Decisions Group at MIT for providing us the opportunity to present this work and for their feedback. In particular, we thank Sarah Cen, Yiling Chen, Nikhil Garg, James Grimmelmann, Karen Levy, Helen Nissenbaum, Kenny Peng, Manish Raghavan, Yoav Wald, and Lily Xu for illuminating conversations and suggestions. 

The work is supported in part by a grant from the John D. and Catherine T. MacArthur Foundation. Ben Laufer is additionally supported by a LinkedIn-Bowers CIS PhD Fellowship, a doctoral fellowship from DLI, and a SaTC NSF grant CNS-1704527. Jon Kleinberg is additionally supported by a Vannevar Bush Faculty Fellowship, AFOSR award FA9550-19-1-0183, and a grant from the Simons Foundation. Hoda Heidari acknowledges support from NSF (IIS-2040929 and IIS-2229881) and PwC (through the Digital Transformation and Innovation Center at CMU). Any opinions, findings, conclusions, or recommendations expressed in this material are those of the authors and do not reflect the views of NSF or other funding agencies.

\bibliographystyle{ACM-Reference-Format}
\bibliography{sample}

\section{Section 2 Materials}
\label{app:section-2}

\subsection{Pareto set characterization and Theorem \ref{thm:unimodal-pareto}}
\begin{proof}{Proof of Theorem \ref{thm:unimodal-pareto}.}
    Consider three non-overlapping intervals that collectively span the feasible set $\delta \in [0,1]$. These intervals are: 
    \begin{enumerate}
        \item $0\leq \delta < \min(\delta^{\text{Powerful }D},\delta^{\text{Powerful }G})$
        \item $\min(\delta^{\text{Powerful }D},\delta^{\text{Powerful }G}) \leq \delta \leq \max(\delta^{\text{Powerful }D},\delta^{\text{Powerful }G})$
        \item $\max(\delta^{\text{Powerful }D},\delta^{\text{Powerful }G}) < \delta \leq 1$
    \end{enumerate}
    We will characterize each of these intervals in turn, finding that intervals (1) and (3) are always Pareto dominated, and interval (2) is characterized by a trade-off in utilities.
    \begin{enumerate}
        \item  Within interval (1), the domain is characterized by $\delta < \min(\delta^{\text{Powerful }D},\delta^{\text{Powerful }G}) \Rightarrow \delta < \delta^{\text{Powerful }D} $ and $ \delta < \delta^{\text{Powerful }G}$. By the definition of a strictly unimodal function (\ref{def:strict-unimodal}), this means that both utility functions $\{U_D,U_G\}$ are strictly increasing over interval 1. Thus, there exists some quantity $\epsilon>0$ such that, for any value $\delta$ in interval (1), $U_D(\delta+\epsilon)>U_D(\delta)$ and $U_G(\delta+\epsilon)>U_G(\delta)$. Thus, every potential agreement in interval (1) is Pareto-dominated.
        \item Within interval (2), the domain is characterized by $\min(\delta^{\text{Powerful }D},\delta^{\text{Powerful }G}) \leq \delta$, and also $\delta \leq \max(\delta^{\text{Powerful }D},\delta^{\text{Powerful }G})$. If $\delta^{\text{Powerful }D}=\delta^{\text{Powerful }G}$, then the value $\delta=\delta^{\text{Powerful }D}=\delta^{\text{Powerful }G}$ is the unique Pareto-optimal agreement because it is optimal for both players. Otherwise if $\delta^{\text{Powerful }D} \neq \delta^{\text{Powerful }G}$, then interval (2) can be characterized as follows: For one player $P\in\{G,D\}$, the utility $U_P$ one utility function is strictly decreasing because $\delta \geq \delta^{\text{Powerful }P}$ and $U_P(\delta)$ is a strictly unimodal function. For the other player $\{G,D\}\setminus P$, the utility $U_{\{G,D\}\setminus P}$ is strictly increasing because $\delta \leq \delta^{\text{Powerful }\{G,D\}\setminus P}$ and $U_{\{G,D\}\setminus P}(\delta)$ is a strictly unimodal function. Since one player's utility is strictly increasing and the other's is strictly decreasing, any perturbation of $\delta$ within interval (2) constitutes a utility gain for one player and a utility loss for the other. For any value of $\delta$ within this interval, if both players' utilities exceed the disagreement payoff (i.e., positive utility), then $\delta$ is Pareto-optimal.
        \item  Within interval (3), the domain is characterized by $\delta > \max(\delta^{\text{Powerful }D},\delta^{\text{Powerful }G}) \Rightarrow \delta > \delta^{\text{Powerful }D} $ and $ \delta > \delta^{\text{Powerful }G}$. By the definition of a strictly unimodal function (\ref{def:strict-unimodal}), this means that both utility functions $\{U_D,U_G\}$ are strictly decreasing over interval (3). Thus, there exists some quantity $\epsilon>0$ such that, for any value $\delta$ in interval (3), $U_D(\delta-\epsilon)>U_D(\delta)$ and $U_G(\delta-\epsilon)>U_G(\delta)$. Thus, every potential agreement in interval (3) is Pareto-dominated.
    \end{enumerate}
Thus interval (2) is Pareto-efficient among the set of feasible bargaining agreements.
\end{proof}

\section{Section 3 Materials (Polynomial fine-tuning game)}
\label{app:section-3}

\subsection{Subgame perfect equilibrium findings}

\begin{proof}{Proof of Theorem \ref{subgame-perfect-eq}.}
\label{app:subgame-perfect-polynomial-proof}
We solve the game using backward induction as follows: 

First, starting with the last stage (3), we solve for $\alpha_1^*$ given $\alpha_0,\delta,c_1$: 
\begin{eqnarray}\label{eq:alpha_i^*}
&& \alpha_1^* = \text{argmax}_{\alpha_1} U_{D}(\alpha_1,\alpha_0,\delta) \nonumber \\
\Rightarrow && \frac{\partial U_D}{\partial \alpha_1}  \bigg\vert_{\alpha_1= \alpha_1^*}= 0 \nonumber \\
\Rightarrow &&\frac{\partial}{\partial \alpha_1} \left((1-\delta) \alpha_1 - c_1(\alpha_1-\alpha_0)^{k_1}\right) \bigg\vert_{\alpha_1= \alpha_1^*}= 0 \nonumber \\
\Rightarrow && (1-\delta)-k_1c_1(\alpha_1^*-\alpha_0)^{k_1-1}=0 \nonumber\\
\Rightarrow && \alpha_1^* = \alpha_0 + \left(\frac{1-\delta}{k_1c_1}\right)^{\frac{1}{k_1-1}}.
\end{eqnarray}
Note that $\frac{\partial^2 U_D}{\partial \alpha_1^2}=-k_1(k_1-1)c_1(\alpha_1-\alpha_0)^{k_1-2}$. This quantity is negative as long as $k>1$, which is assumed. Thus, the $\alpha_1^*$ derived above yields a global maximum of $U_D$.

Second, knowing $D$'s choice of $\alpha_1^*$ above, we solve for $\alpha_0^*$ as follows:
\begin{eqnarray*}
    && \alpha_0^* = \text{argmax}_{\alpha_0} U_{G}(\alpha_0, \delta) \\
    \Rightarrow && \frac{\partial U_{G}}{\partial \alpha_0}\bigg\vert_{\alpha_0= \alpha_0^*} = 0 \\
    \Rightarrow &&  \frac{\partial }{\partial \alpha_0}\left(\delta \alpha_1^* - c_0\alpha_0^{k_0}\right) \bigg\vert_{\alpha_0= \alpha_0^*} = 0 \\
    \Rightarrow &&  \frac{\partial }{\partial \alpha_0}\left(\delta \left(\alpha_0 + \left(\frac{1-\delta}{k_1c_1}\right)^{\frac{1}{k_1-1}}\right) - c_0\alpha_0^{k_0}\right) \bigg\vert_{\alpha_0= \alpha_0^*} = 0 \\
    \Rightarrow &&  \frac{\partial }{\partial \alpha_0}\left(\delta \alpha_0 + [\text{const}] - c_0\alpha_0^{k_0}\right) \bigg\vert_{\alpha_0= \alpha_0^*} = 0 \\
\Rightarrow &&  \delta - k_0c_0(\alpha_0^*)^{k_0-1} = 0 \\
\Rightarrow &&  \alpha_0^* = \left(\frac{\delta}{k_0c_0}\right)^{\frac{1}{k_0-1}}. \\
\end{eqnarray*}
The second derivative $\frac{\partial^2 U_{G}}{\partial \alpha^2}=-k_0(k_0-1)c_0(\alpha_0)^{k_0-2}$. This quantity is negative as long as $k>1$, which is assumed. Thus, the value of $\alpha_0^*$ derived above yields a global maximum of $U_G$.

Finally, plugging in  $ \alpha_0^* = \left(\frac{\delta}{k_0c_0}\right)^{\frac{1}{k_0-1}}$ into Equation \ref{eq:alpha_i^*}, we obtain the following expression for $\alpha_1^*$ as a function of $\delta$ only: 
$$\alpha_1^* =  \left(\frac{\delta}{k_0c_0}\right)^{\frac{1}{k_0-1}} + \left(\frac{1-\delta}{k_1c_1}\right)^{\frac{1}{k_1-1}}.$$
\end{proof}

\subsection{Utilities as a function of $\delta$}
\begin{proof}{Proof of Corollary \ref{cor:utilities-solved}}
Plugging the formulas from Theorem \ref{subgame-perfect-eq} into Equation \ref{u-g}, we obtain:
\begin{eqnarray*}
    U_G = && \delta \alpha_1 - \phi_0(\alpha_0) \\
    = && \delta \left(\left(\frac{\delta}{k_0c_0}\right)^{\frac{1}{k_0-1}} + \left(\frac{1-\delta}{k_1c_1}\right)^{\frac{1}{k_1-1}}\right) - c_0\left(\left(\frac{\delta}{k_0c_0}\right)^{\frac{1}{k_0-1}} \right)^{k_0} \\
    = && \delta\left(\frac{\delta}{k_0c_0}\right)^{\frac{1}{k_0-1}}+\delta\left(\frac{1-\delta}{k_1c_1}\right)^{\frac{1}{k_1-1}}-c_0\left(\frac{\delta}{k_0c_0}\right)^{\frac{k_0}{k_0-1}} \\
    = && \left[\left(\frac{1}{k_0c_0}\right)^{\frac{1}{k_0-1}}-c_0\left(\frac{1}{k_0c_0}\right)^{\frac{k_0}{k_0-1}}\right]\delta^{\frac{k_0}{k_0-1}}+ \left(\frac{1}{k_1c_1}\right)^{\frac{1}{k_1-1}}\delta(1-\delta)^{\frac{1}{k_1-1}} \\
    = && \left(\frac{1}{k_0c_0}\right)^{\frac{1}{k_0-1}}\left(1-\frac{1}{k_0}\right)\delta^{\frac{k_0}{k_0-1}} + \left(\frac{1}{k_1c_1}\right)^{\frac{1}{k_1-1}}\delta(1-\delta)^{\frac{1}{k_1-1}}.\\
\end{eqnarray*}

Plugging the formulas from Theorem \ref{subgame-perfect-eq} into Equation \ref{u-g}, we obtain:  
\begin{eqnarray*}
    U_D = && (1-\delta) \alpha_1 - \phi_i(\alpha_1;\alpha_0) \\
    = && (1-\delta)\left[\left(\frac{\delta}{k_0c_0}\right)^{\frac{1}{k_0-1}}+\left(\frac{1-\delta}{k_1c_1}\right)^{\frac{1}{k_1-1}}\right] - c_1\left[\left(\frac{\delta}{k_0c_0}\right)^{\frac{1}{k_0-1}} + \left(\frac{1-\delta}{k_1c_1}\right)^{\frac{1}{k_1-1}} - \left(\frac{\delta}{k_0c_0}\right)^{\frac{1}{k_0-1}}\right]^{k_1}\\
     = && (1-\delta)\left(\frac{\delta}{k_0c_0}\right)^{\frac{1}{k_0-1}}+(1-\delta)\left(\frac{1-\delta}{k_1c_1}\right)^{\frac{1}{k_1-1}} - c_1\left(\frac{1-\delta}{k_1c_1}\right)^{\frac{k_1}{k_1-1}} \\
     = && \left(\frac{1}{k_0c_0}\right)^{\frac{1}{k_0-1}}(1-\delta)\delta^{\frac{1}{k_0-1}} 
     + \left(\frac{1}{k_1c_1}\right)^{\frac{1}{k_1-1}}(1-\delta)^{\frac{k_1}{k_1-1}} 
     - \left(\frac{1}{k_1c_1}\right)^{\frac{1}{k_1-1}}\left(\frac{1}{k_1}\right)(1-\delta)^{\frac{k_1}{k_1-1}} \\
     = && 
     \left(\frac{1}{k_1c_1}\right)^{\frac{1}{k_1-1}}\left(1-\frac{1}{k_1}\right)(1-\delta)^{\frac{k_1}{k_1-1}}
     + \left(\frac{1}{k_0c_0}\right)^{\frac{1}{k_0-1}}(1-\delta)\delta^{\frac{1}{k_0-1}}. \\ 
\end{eqnarray*}
\end{proof}

\subsection{Utilities are strictly unimodal functions of $\delta$}
\label{app:unimodal-one-specialist}

\begin{proof}{Proof of Proposition \ref{obs:polynomial-unimodal}.} We'll start by proving $U_G$ is strictly unimodal, and then extend the results to $U_D$. 

\textbf{Beginning with $U_G$:} The proof relies on the following Lemma:

    \begin{lemma}
    \label{unimodal-lemma}
        A differentiable continuous function f($\delta$) is strictly unimodal over $\delta \in [a,b]$ if the following conditions are met: \textbf{i)} $f'(a) > 0$, \textbf{ii)} $f'$ is concave over the domain.
    \end{lemma}

\begin{proof}{Proof of Lemma \ref{unimodal-lemma}.}
    By the definition of strict unimodality, we can conclude a function $f(\delta)$ is strictly unimodal over an interval $[a,b]$ if one of the following properties hold: 1) $f'(\delta)>0\ \forall\ \delta \in [0,1]$, meaning the function is strictly increasing over the interval, or 2) For some value $c \in (a,b)$, the function is strictly increasing for values $[a,c)$ and strictly decreasing for values $(c,b]$. Notice that, so long as condition (\textbf{i}) holds (i.e., the function starts out strictly increasing at $a$), the function is strictly unimodal as long as its derivative crosses the $f'=0$ axis at \textit{no more than one} point in $[a,b]$. So, the remainder of the proof finds a contradiction when we assume conditions (\textbf{i}) and (\textbf{ii}) and that there are two values in $(a,b)$ for which $f'=0$. 

    \begin{figure}[h]
        \centering
        \includegraphics[width=0.65\linewidth]{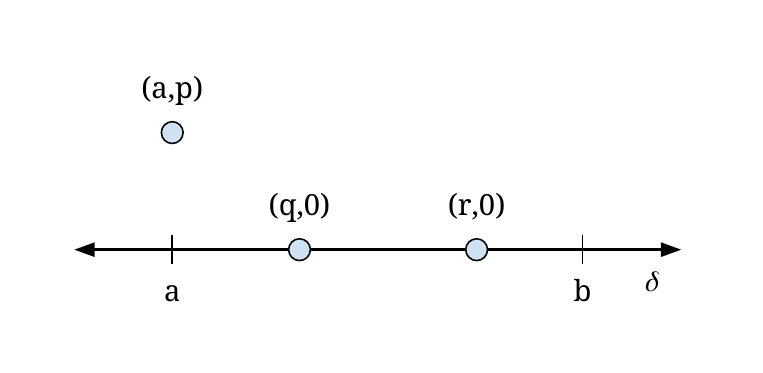}
        \caption{Illustration of the proof for Lemma \ref{unimodal-lemma}. If the derivative of a function $f$ is positive at $a$ and concave, it cannot cross the axis more than 1 time, meaning $f$ is strictly unimodal.}
        \label{fig:enter-label}
    \end{figure} 
    
    Consider the curve $f'=\frac{df}{d\delta}$. We specify a point on this curve using $\{ (x,y) | x= \delta , y = \frac{df}{d \delta}\}$. Given condition (\textbf{i}), There is some point $x_0 = (a,p)$ where $p>0$ and $x_0 \in \frac{d f}{d \delta}$. Assume for the sake of contradiction that there are two points $x_1 = (q,0), x_2 = (r,0)$ where $x_1,x_2 \in \frac{d f}{d \delta}$ and $a<q<r<b$ (without loss of generality). We can plug these three points into the definition of concavity and find our contradiction: First, notice $l=\frac{q-a}{r-a} \in (0,1)$ because $a<q<r$. Next, plugging in the definition of concavity: 
    \begin{eqnarray*}
        \frac{d f}{d \delta}((1-l)a + lr) & \geq & (1-l)\frac{d f}{d \delta}(a)+l\frac{d f}{d \delta}(r) \\
        \frac{d f}{d \delta}\left(\left(1-\frac{q-a}{r-a}\right)a + \frac{q-a}{r-a}r\right) & \geq & \left(1-\frac{q-a}{r-a}\right)[p] + \frac{q-a}{r-a}[0] \\
        \frac{d f}{d \delta}\left(a-\frac{q-a}{r-a}a + \frac{q-a}{r-a}r\right) & \geq & p-\frac{q-a}{r-a}p \\
        \frac{d f}{d \delta}\left(a + (r-a)\frac{q-a}{r-a}\right) & \geq & p-\frac{q-a}{r-a}p \\
        \frac{d f}{d \delta}\left(a + (r-a)\frac{q-a}{r-a}\right) & \geq & p-\frac{q-a}{r-a}p \\
        \frac{d f}{d \delta}\left(a + (q-a)\right) & \geq & p-\frac{q-a}{r-a}p \\
        \frac{d f}{d \delta}\left(q\right) & \geq & p-\frac{q-a}{r-a}p \\
         0 & \geq & p-\frac{q-a}{r-a}p > 0. \\
    \end{eqnarray*} Hence the contradiction: We know $p-\frac{q-a}{r-a}p>0$ is strictly greater than 0 because $p>0$ and $l\in(0,1)$. Thus a function characterized by conditions (\textbf{i}) and (\textbf{ii}) cannot contain these three points. This concludes Lemma \ref{unimodal-lemma}'s proof: If  $f'$ is concave, and starts positive at $a$, it must cross the axis at most once meaning $f$ is unimodal.
\end{proof}

\noindent
    Now, we prove the utilities are unimodal by showing that $U_G(\delta)$ satisfies both conditions in Lemma \ref{unimodal-lemma} for the domain $\delta \in [0,1]$.

    We first differentiate $U_G$ with respect to $\delta$. 

    \begin{eqnarray*}
        \frac{\partial U_G}{\partial \delta} = & \frac{\partial}{\partial \delta}\left[\left(\frac{1}{k_0c_0}\right)^{\frac{1}{k_0-1}}\left(1-\frac{1}{k_0}\right)\delta^{\frac{k_0}{k_0-1}} + \left(\frac{1}{k_1c_1}\right)^{\frac{1}{k_1-1}}\delta(1-\delta)^{\frac{1}{k_1-1}}\right] \\
        = & \frac{\partial}{\partial \delta}\left[A\delta^{\frac{k_0}{k_0-1}} + B \delta(1-\delta)^{\frac{1}{k_1-1}}\right],
    \end{eqnarray*}
    where $A := \left(\frac{1}{k_0c_0}\right)^{\frac{1}{k_0-1}}\left(1-\frac{1}{k_0}\right) > 0$ and $B := \left(\frac{1}{k_1c_1}\right)^{\frac{1}{k_1-1}} > 0$. $A$ is positive as long as $k_0>1$ and $c_0>0$, which is given. $B$ is positive as long as $k_1>1$ and $c_1>0$, which is given. Continuing:
    \begin{eqnarray*}
    = & A\left(\frac{k_0}{k_0-1}\right)\delta^{\frac{1}{k_0-1}} + B \left[-\delta \left(\frac{1}{k_1-1}\right)(1-\delta)^{\frac{1}{k_1-1}-1} + (1-\delta)^{\frac{1}{k_1-1}}\right]\\
    = & A\left(\frac{k_0}{k_0-1}\right)\delta^{\frac{1}{k_0-1}} + \frac{(1-\delta)^{\frac{1}{k_1-1}-1}}{k_1-1}\left[-\delta+(k_1-1)(1-\delta)\right]\\
    = & A\left(\frac{k_0}{k_0-1}\right)\delta^{\frac{1}{k_0-1}} + \frac{(1-\delta)^{\frac{1}{k_1-1}-1}}{k_1-1}\left[k_1 - k_1\delta - 1 + \delta - \delta\right]\\
    = & A\left(\frac{k_0}{k_0-1}\right)\delta^{\frac{1}{k_0-1}} + \frac{(1-\delta)^{\frac{1}{k_1-1}-1}}{k_1-1}\left[k_1(1-\delta) - 1\right]\\
    \frac{\partial U_G}{\partial \delta} = & A\left(\frac{k_0}{k_0-1}\right)\delta^{\frac{1}{k_0-1}} + 
    B \left(\frac{k_1}{k_1-1}\right) (1-\delta)^{\frac{1}{k_1-1}} -
    B \left(\frac{1}{k_1-1}\right) (1-\delta)^{\frac{1}{k_1-1}-1}.\\
    \end{eqnarray*}
    Now we can show the first condition (a) in Lemma \ref{unimodal-lemma} holds: 

    $$\frac{\partial U_G}{\partial \delta} \bigg\vert_{\delta=0} = [0] + B\left(\frac{k_1}{k_1-1}\right) - B\left(\frac{1}{k_1-1}\right) = B\left(\frac{k_1-1}{k_1-1}\right) = B > 0.$$

    To show the second condition (b) in Lemma \ref{unimodal-lemma} holds, we perform the second-derivative test, which requires differentiating the function two more times:

    \begin{eqnarray*}
        \frac{\partial^3 U_G}{{\partial \delta}^3} =  & A\left(\frac{k_0}{k_0-1}\right)\left(\frac{1}{k_0-1}\right)\left(\frac{1}{k_0-1}-1\right)\delta^{\frac{1}{k_0-1}-2} \\
        & + B \left(\frac{k_1}{k_1-1}\right) \left(\frac{1}{k_1-1}\right) \left(\frac{1}{k_1-1}-1\right) (1-\delta)^{\frac{1}{k_1-1}-2} \\
        & - B \left(\frac{1}{k_1-1}\right) \left(\frac{1}{k_1-1}-1\right) \left(\frac{1}{k_1-1}-2\right) (1-\delta)^{\frac{1}{k_1-1}-3}.
    \end{eqnarray*}
    
    The above expression is never positive. First, notice all three coefficients are less than or equal to zero:
    \begin{itemize}
        \item $A(\frac{k_0}{k_0-1})(\frac{1}{k_0-1})(\frac{1}{k_0-1}-1)$ is the product of one negative and otherwise non-negative numbers: Given $k_0\geq 2$, observe $A>0$, $(\frac{k_0}{k_0-1})>0$, $(\frac{1}{k_0-1})>0$, $(\frac{1}{k_0-1}-1)\leq0$. 
        \item $B(\frac{k_1}{k_1-1})(\frac{1}{k_1-1})(\frac{1}{k_1-1}-1)$ is the product of one negative and otherwise non-negative numbers: Given $k_1\geq 2$, observe $B>0$, $(\frac{k_1}{k_1-1})>0$, $(\frac{1}{k_1-1})>0$, $(\frac{1}{k_1-1}-1)\leq0$.
        \item $-B(\frac{1}{k_1-1})(\frac{1}{k_1-1}-1)(\frac{1}{k_1-1}-2)$ is the product of three negative and otherwise non-negative numbers: Given $k_1\geq 2$, observe $-B<0$, $(\frac{1}{k_1-1})>0$, $(\frac{1}{k_1-1}-1)\leq0$,$(\frac{1}{k_1-1}-1)<0$.
    \end{itemize}
    Second, notice all three expressions of $\delta$ are defined and positive on the interval $(0,1)$:
    \begin{itemize}
        \item $\delta^{\frac{1}{k_0-1}-2}$ is positive and defined $\forall \delta>0$.
        \item $(1-\delta)^{\frac{1}{k_1-1}-2}$ is positive and defined $\forall \delta<1$.
        \item $(1-\delta)^{\frac{1}{k_1-1}-3}$ is positive and defined $\forall \delta<1$.
    \end{itemize}
    Every term in our derived expression for $\frac{\partial^3 U_G}{{\partial U_G}^3}$ is non-positive. Thus the function $\frac{\partial U_G}{\partial U_G}$ is concave satisfying condition (b) in Lemma \ref{unimodal-lemma}. This completes the proof that $U_G$ is unimodal.

    \textbf{Moving on to $U_D$:} Notice the formulation of $U_G$ in Equation \ref{u-g-polynomial} is almost exactly the same functional form as $U_D$ in Equation \ref{u-d-polynomial}. If we define a variable $\gamma=(1-\delta)$, we can use the identical proof completed above to show $U_D$ is unimodal in $\gamma$. Since we prove unimodality on the interval $[0,1]$, a function defined over $\gamma \in [0,1]$ is simply a function of $\delta\in[0,1]$ reflected over the vertical line $\delta=0.5$. A transform that reflects a univariate function over the vertical line passing through the midpoint of its domain preserves strict unimodality.
\end{proof}

\subsection{Powerful-G Bargaining Solution}
\label{subsec:powerful-g-single}
\begin{proof}{Proof of Proposition \ref{prop:powerful-g}.}
The powerful-$G$ solution is the solution $\delta^{\textit{Powerful }G}$ that maximizes $U_G$ over the feasible set of $\delta \in [0,1]$:
\begin{eqnarray*}
    & \delta^{\textit{Powerful }G} = \text{argmax}_{\delta} U_G(\delta) \\
    \Rightarrow & \frac{\partial U_G}{\partial \delta} = 0 \\
    \Rightarrow & \frac{\partial}{\partial \delta} \left[\frac{\delta^2}{4c_0}+\frac{\delta}{2c_1}-\frac{\delta^2}{2c_1}\right]= 0 & \textit{Corr. \ref{cor:utilities-solved}}\\
    \Rightarrow & \frac{\delta}{2c_0} + \frac{1}{2c_1} - \frac{\delta}{c_1} = 0 \\
    \Rightarrow & \delta \left(\frac{1}{2c_0}-\frac{1}{c_1}\right) = -\frac{1}{2c_1} \\
    \Rightarrow & \delta = -\frac{1}{2c_1} \left(\frac{c_1-2c_0}{2c_1c_0}\right)^{-1} \\
    \Rightarrow & \delta = \frac{c_0}{2c_0-c_1}.
\end{eqnarray*}

The second partial derivative $\frac{\partial^2 U_G}{\partial \delta^2} = \frac{1}{2c_0} - \frac{1}{c_1}$, which is negative as long as $0 < \frac{c_1}{c_0} < 2$. Since there is only one root, the derived equation is a global maximum for $0 < c_1 < 2{c_0}$. However, notice that the derived expression is only feasible for the values ${c_1}\leq{c_0} $, since the value $\delta$ must be in the range $[0,1]$ (Specialist would not take a negative share of the profit). Thus, $\delta^{\textit{Powerful }G} = \frac{1}{2-c_1}$ for $0 < {c_1}<{c_0}$.

The remainder of the proof will show that, for $c_1 \geq c_0$, within the feasible set $0 \leq \delta \leq 1$, $\delta=1$ maximizes $U_G$. We'll do so by showing that the partial derivative $\frac{\partial U_G}{\partial \delta}$ is non-negative for all $c_1 \geq c_0$ and $0 \leq \delta \leq 1$. Assume for sake of contradiction:

\begin{eqnarray*}
    & \frac{\partial U_G}{\partial \delta} < 0 \\
    \Rightarrow & \frac{\delta}{2c_0} + \frac{1}{2c_1} - \frac{\delta}{c_1}  < 0 \\
     \Rightarrow & \delta\left(\frac{1}{2c_0}-\frac{1}{c_1}\right)+\frac{1}{2c_1}  < 0 \\
    \Rightarrow &  \delta\left(\frac{c_1-2c_0}{2c_1c_0}\right)+\frac{c_0}{2c_1c_0}  < 0 \\
     \Rightarrow & \delta(c_1-2c_0)+c_0  < 0 &  \textit{because }c_1,c_0>0. \\
    \Rightarrow &  \delta(c_1-2c_0)  < -c_0. \\
\end{eqnarray*}
Notice the resulting expression is only met when $\delta<0$ or $c_1 \leq c_0$. However, we're given $\delta \in [0,1] \cap c_1 \geq c_0$.
\end{proof}

\subsection{Powerful-D Bargaining Solution}
\label{subsec:powerful-d-single}

\begin{proof}{Proof of Proposition \ref{prop:powerful-d}.}
The powerful-$D$ solution is the solution $\delta^{\textit{Powerful }D}$ that maximizes $U_D$ over the feasible set:
\begin{eqnarray*}
    & \delta^{\textit{Powerful }D} = \text{argmax}_{\delta} U_D(\delta) \\
    \Rightarrow & \frac{\partial U_D}{\partial \delta} = 0 \\
    \Rightarrow & \frac{\partial}{\partial \delta} \left[\left(\frac{1}{4c_1}\right)+\left(\frac{1}{2c_0}-\frac{1}{2c_1}\right)\delta + \left(\frac{1}{4c_1}-\frac{1}{2c_0}\right)\delta^2\right]= 0 & \textit{Corr. \ref{cor:utilities-solved}}\\
    \Rightarrow & \left(\frac{1}{2c_0}-\frac{1}{2c_1}\right) + 2\left(\frac{1}{4c_1}-\frac{1}{2c_0}\right)\delta = 0 \\
    \Rightarrow & \left(\frac{1}{2c_1}-\frac{1}{c_0}\right)\delta = \frac{1}{2c_1}-\frac{1}{2c_0} \\
    \Rightarrow & \delta = \frac{\left(\frac{c_0-c_1}{2c_0c_1}\right)}{\left(\frac{c_0-2c_1}{2c_0c_1}\right)}\\
    \Rightarrow & \delta = \frac{c_1-c_0}{2c_1-c_0}.
\end{eqnarray*}

The second partial derivative $\frac{\partial^2 U_D}{\partial \delta^2} = \frac{1}{2c_1}-\frac{1}{c_0}$, which is negative as long as $2c_1 > c_0$. Since there is only one root, the derived equation is a global maximum for $2c_1 > c_0$. However, notice that the derived expression is only feasible for the values $c_1 \geq c_0$, since the value $\delta$ must be in the range $[0,1]$ (Generalist would not take a negative share of the profit). Thus, $\delta^{\textit{Powerful }D} = \frac{c_1-c_0}{2c_1-c_0}$ for $c_1 \geq c_0$.

The remainder of the proof will show that, for $c_1 < c_0$, within the feasible set $0 \leq \delta \leq 1$, $\delta=0$ maximizes $U_D$. We'll do so by showing that the partial derivative $\frac{\partial U_D}{\partial \delta} \leq 0$ for all $0 < c_1 < 1$ and $0 \leq \delta \leq 1$. Assume for sake of contradiction:

\begin{eqnarray*}
    & \frac{\partial U_D}{\partial \delta}  > 0 \\
     \Rightarrow & \left(\frac{1}{2c_0}-\frac{1}{2c_1}\right) + \left(\frac{1}{2c_1}-\frac{1}{c_0}\right)\delta  > 0 \\
      \Rightarrow &  \left(\frac{1}{2c_1}-\frac{1}{c_0}\right)\delta  > \frac{1}{2c_1} - \frac{1}{2c_0} \\
      \Rightarrow &  \left(\frac{c_0-2c_1}{2c_0c_1}\right)\delta  > \frac{c_0-c_1}{2c_0c_1} \\
      \Rightarrow &  \left(c_0-2c_1\right)\delta  > c_0-c_1 & \textit{because }c_0,c_1>0.\\
\end{eqnarray*}
We show the contradiction $\forall \frac{c_1}{c_0} \in (0,1)$ (equivalently, every scenario where $0<c_1\leq c_0$):
\begin{itemize}
    \item For $\frac{1}{2} < \frac{c_1}{c_0} <1$, the final step implies $\delta  > \frac{c_1-c_0}{2c_1-c_0}$, which contradicts the global optimum finding above.
    \item For $0 < \frac{c_1}{c_0}  < \frac{1}{2}$, the final step implies $\delta  \leq \frac{c_1-c_0}{2c_1-c_0}$. But notice the right-hand-side must be negative, contradicting the given range $\delta \in [0,1]$.
    \item Finally, for $\frac{c_1}{c_0}=\frac{1}{2}$, the final step implies $0>\frac{1}{2}$. 
\end{itemize}
Thus we've established the contradiction. 
\end{proof}

\subsection{Vertical monopoly bargaining solution}

\begin{proof}{Proof of Proposition \ref{prop:vertical-monopoly-quadratic}.}
    The vertical monopoly or ``utilitarian'' solution is the solution that maximizes the sum of utilities:

    \begin{eqnarray*}
        & \delta^{Vertical Monopoly} = \text{argmax}_{\delta}U_G(\delta)+U_D(\delta) \\
        \Rightarrow & \frac{\partial}{\partial \delta}\left(U_G(\delta)+U_D(\delta)\right)=0 \\
        \Rightarrow & \frac{\delta}{2c_0} + \frac{1}{2c_1} - \frac{\delta}{c_1} + \left(\frac{1}{2c_0}-\frac{1}{2c_1}\right) + 2\left(\frac{1}{4c_1}-\frac{1}{2c_0}\right)\delta =0 & \textit{Corr. \ref{cor:utilities-solved}} \\
        \Rightarrow & \delta\left(\frac{1}{2c_0} - \frac{1}{c_1} +\frac{1}{2c_1}-\frac{1}{c_0}\right) = -\frac{1}{2c_1}-\frac{1}{2c_0}+\frac{1}{2c_1} \\
        \Rightarrow & \delta\left(c_1-2c_0+c_0-2c_1\right)=-c_1 \\
        \Rightarrow &  \delta = \frac{c_1}{c_1+c_0}.
    \end{eqnarray*}

The second partial derivative is $\frac{\partial^2 U_G}{\partial\delta^2}=-\frac{1}{2c_1}-\frac{1}{2c_0}$ which is negative for any $c_0,c_1>0$, meaning $\delta^{Vertical Monopoly}=\frac{c_1}{c_1+c_0}$ maximizes the sum of utilities.
\end{proof}

\subsection{Egalitarian bargaining solution}

\begin{proof}{Proof of Proposition \ref{prop:kalai-quadratic}.}
    The Kalai (egalitarian) solution $\delta^{\textit{Egal.}}$ is the solution that maximizes the minimum utility among players. 

    First, observe that if there exists a point in the Pareto solution set where the two utilities are equal, this point must be the egalitarian solution. Pareto means that an increase in any player's utility must correspond to a decrease in another player's utility. If a solution within the Pareto set equalizes utilities, then any other solution must inevitably trade off one player's utility for the other's, meaning any alternative solution would yield a lower utility for at least one player. 
    
    So, setting the utilities equal we get:
    \begin{eqnarray*}
        U_G(\delta)=U_D(\delta) \\
        \frac{\delta^2}{4c_0}+ \frac{\delta}{2c_1}-\frac{\delta^2}{2c_1} = \frac{1}{4c_1}+\frac{\delta}{2c_0}-\frac{\delta}{2c_1}+\frac{\delta^2}{4c_1}-\frac{\delta^2}{2c_0} & \textit{Corr. \ref{cor:utilities-solved}}\\
        \delta^2(3c_1-3c_0)+\delta(4c_0-2c_1)-c_0=0
    \end{eqnarray*}

    Plugging into the quadratic formula, we get two candidate solutions:

    \begin{eqnarray*}
        \delta^{\textit{Egal.}} \stackrel{?}{=} \left\{\frac{\sqrt{c_0^2-c_0c_1+c_1^2}-c_1+2c_0}{3(c_0-c_1)},\frac{-\sqrt{c_0^2-c_0c_1+c_1^2}-c_1+2c_0}{3(c_0-c_1)}\right\}
    \end{eqnarray*}

    Notice that the first of these solutions, for $c_0,c_1>0$, is not in the feasible set $0 \leq \delta \leq 1$. Thus, the Egalitarian solution is given by:

    \begin{equation*}
        \delta^{\textit{Egal.}} = \frac{-\sqrt{c_0^2-c_0c_1+c_1^2}-c_1+2c_0}{3(c_0-c_1)}.
    \end{equation*}
\end{proof}

\subsection{Maximum-performance bargaining solution}

\begin{proof}{Proof of Proposition \ref{prop:max-alpha-quadratic}.}
    We will show that within the feasible region $\delta \in [0,1], c_1>0$, the following three properties hold:
\begin{enumerate}
    \item $\frac{\partial \alpha_1}{\partial \delta}(c_1) <0\ \ \forall\ \ c_1 < c_0$
    \item $\frac{\partial \alpha_1}{\partial \delta}(c_1) >0\ \ \forall\ \ c_1 > c_0$
    \item $\frac{\partial \alpha_1}{\partial \delta}(c_1) = 0\ \text{ for }\ c_1=c_0$
\end{enumerate}

First, we differentiate our expression for $\alpha_1(\delta;c_0,c_1)$ with respect to $\delta$, using the expression attained in Theorem \ref{subgame-perfect-eq}:
\begin{eqnarray*}
    & \frac{\partial \alpha_1}{\partial \delta} \\
    = & \frac{\partial}{\partial \delta}\left[\frac{\delta}{2c_0}+\frac{1-\delta}{2c_1}\right] \\
    = & \frac{1}{2c_0}-\frac{1}{2c_1}.
\end{eqnarray*}
    Now notice each of the three properties are satisfied in turn:
\begin{enumerate}
    \item For $c_1 < c_0$, $\frac{\partial \alpha_1}{\partial \delta} = \frac{1}{2c_0}-\frac{1}{2c_1} <0$.
    \item For $c_1 > c_0$, $\frac{\partial \alpha_1}{\partial \delta} = \frac{1}{2c_0}-\frac{1}{2c_1} >0$.
    \item For $c_1 = c_0$, $\frac{\partial \alpha_1}{\partial \delta} = \frac{1}{2c_0}-\frac{1}{2c_1} =0$.
\end{enumerate} 
\end{proof}


\bldelete{
\section{Section 4 Materials}
\subsection{Existence theorem}
\label{app:general}

\begin{proof}{Proof of Theorem \ref{thm:general}.}
We prove this theorem via a sequence of Lemmas.

\begin{lemma}
\label{lemma:general-zero}
    If $r'(0) > \lambda_0 \phi_0'(0)$ for a constant $\lambda_0\geq2$, then there exists a set $A^*\subseteq(0,1)$ such that $\frac{1}{2} \in A^*$ and for all $\delta \in A^*$, $\alpha_0^*(\delta)>0$, $U_G(\delta)>0$.
\end{lemma}

Let's presume $\delta^*=\frac{1}{2}$. If we can show that this solution yields $\alpha_0^*>0$, then this means the generalist's utility is greater than 0 for investing some non-zero effort spend. $\alpha_0^*$ is the value of $\alpha_0$ that maximizes $U_G$, so if $U_G$ has positive slope at $\alpha_0=0$, then $\alpha_0^*>0$.

Notice that this positive-utility outcome is met as long as $\frac{\partial U_G}{\partial \alpha_0}\big\vert_{\alpha_0=0}>0$. This condition would necessarily mean that there exists some positive $\alpha_0>0$ which maximizes $U_G$. So, formally:

\begin{eqnarray*}
    r'(0) > 2 \phi_0'(0) \\
    \frac{1}{2} r'(0) - \phi_0'(0) > 0 \\
    \frac{\partial}{\partial \alpha_0}\left(\delta r - \phi_0\right)\bigg\vert_{\alpha_0=0} > 0 \\
    \frac{\partial U_G}{\partial \alpha_0}\bigg\vert_{\alpha_0=0}>0. \\
\end{eqnarray*} 

Notice that this inequality is met as long as $r'(0) > \lambda_0 \phi_0'(0) $ where $\lambda_0>2$. Thus, there is a non-empty set $A^*$ of solutions with $\frac{1}{2} \in A^*$ that yield positive $\alpha_0$ and positive $U_G$.

\bldelete{
\begin{lemma}
\label{lemma:general-one}
    If $\frac{\partial r}{\partial \alpha_1}\big\vert_{\alpha_1=\alpha_0} > \lambda_1 \frac{\partial \phi_1}{\partial \alpha_1}\big\vert_{\alpha_1=\alpha_0}$ where $\lambda_1\geq2$, then there is a non-empty set of agreements $B^*\in(0,1)$ such that $\alpha_1^* > \alpha_0^*$, $U_D>0$, and $\frac{1}{2} \subseteq B^*$.
\end{lemma}
}

\begin{lemma}
\label{lemma:general-one}
    If $r'(\alpha_0) > \lambda_1 \phi_1'(\alpha_0)$ for a constant $\lambda_1\geq2$, then there exists a set $B^*\subseteq(0,1)$ such that $\frac{1}{2} \in B^*$ and for all $\delta \in B^*$, $\alpha_1^*(\delta) > \alpha_0^*$, $U_D(\delta)>0$.
\end{lemma}

Let's presume $\delta^*=\frac{1}{2}$. If we can show that this solution yields $\alpha_1^*>\alpha_0^*$, then this means the domain specialist's utility is greater than 0 for investing some non-zero effort spend. $\alpha_1^*$ is the value of $\alpha_1$ that maximizes $U_D$, so if $U_D$ has positive slope at $\alpha_1=\alpha_0^*$, then $\alpha_1^*>\alpha_0^*$.

Notice that this positive-utility outcome is met as long as $\frac{\partial U_D}{\partial \alpha_1}\big\vert_{\alpha_1=\alpha_0^*}>0$ -- this condition would necessarily mean that there exists some $\alpha_1^*>\alpha_0^*$ which maximizes $U_D$. So, formally: 
 \begin{eqnarray*}
    r'(\alpha_0) > 2 \phi_1'(\alpha_0) \\
    \frac{1}{2} r'(\alpha_0) - \phi_1'(\alpha_0) > 0 \\
    \frac{\partial}{\partial \alpha_1}\left((1-\delta) r - \phi_1\right)\bigg\vert_{\alpha_1=\alpha_0} > 0 \\
    \frac{\partial U_D}{\partial \alpha_1}\bigg\vert_{\alpha_1=\alpha_0}>0. \\
\end{eqnarray*} 

Notice that this inequality is met as long as $r'(\alpha_0) > \lambda_1 \phi_1'(\alpha_0)$ where $\lambda_1\geq2$. Thus, there is a non-empty set $B^*$ of solutions with $\frac{1}{2} \in B^*$ that yield $\alpha_1^*>\alpha_0^*$ and positive $U_D$.
\bldelete{
\begin{lemma}
    $A^* \cap B^*$ is a non-empty set where any solution $\delta^* \in A^* \cap B^*$ satisfies the three properties: $0<\delta_i<1$; $0<\alpha_0^*<\alpha_1^*$; and $U_D,U_G>0$.
\end{lemma}
}
\begin{corollary}
    $A^* \cap B^*$ is a non-empty set where any solution $\delta^* \in A^* \cap B^*$ satisfies the three properties: $0<\delta_i<1$; $0<\alpha_0^*<\alpha_1^*$; and $U_D,U_G>0$.
\end{corollary}
This Corollary follows from the findings that have already been shown in the former Lemmas.

First, $A^* \cap B^*$ is non-empty because $\frac{1}{2}\in A^*$ (as shown in Lemma \ref{lemma:general-zero}) and $\frac{1}{2}\in B^*$ (as shown in Lemma \ref{lemma:general-one}) so it follows that $\frac{1}{2} \in A^* \cap B^*$. 

The three properties are each met for the following reasons:
\begin{enumerate}
    \item Property 1 ($0<\delta^*<1$) is met because $\frac{1}{2} \in A^* \cap B^*$ and $0<\frac{1}{2} < 1$ so the existence finding is satisfied for a non-extreme value of $\delta$. This means, for $\delta^* \in A^* \cap B^*$, players share revenue.
    \item Property 2 ($0<\alpha_0^*<\alpha_1^*$) is met because any solution in $A^*$ yields $\alpha_0^* > 0$ (as shown in Lemma \ref{lemma:general-zero}) and any solution in $B^*$ yields $\alpha_1^* > \alpha_0^*$ (as shown in Lemma \ref{lemma:general-one}). This means, for $\delta^* \in A^* \cap B^*$, players do not free-ride (they both act to improve the technology).
    \item Property 3 ($U_G,U_D>0$) is met because any solution in $A^*$ yields $U_G>0$ (as shown in Lemma \ref{lemma:general-zero}) and any solution in $B^*$ yields $U_D>0$ (as shown in Lemma \ref{lemma:general-one}). This means any solution $\delta^* \in A^* \cap B^*$ yields positive utility for both players, which Pareto-dominates the disagreement scenario in which both players have zero utility.
\end{enumerate}    
\end{proof}
}

\section{Section \ref{sec:multi} Materials}
\label{app:multi-solutions}

\subsection{Proof of Pareto Set Characterization in the Multi-specialist Game}
\label{app:multi-pareto}
\begin{proof}{Proof of Theorem \ref{thm:multi-specialist-pareto}.}
Say there are $n \geq 2$ specialists. For each player in the game, at a given value of $\delta$, the utility curve can exhibit one of three regimes: 1) negative utility 2) increasing utility 3) decreasing utility. Thus, at a particular value of $\delta$, the utilities as a function of delta can exhibit $3^{n+1}$ different combinations of regimes. For example, if there are $n=2$ domain specialists, a given value of $\delta$ might exhibit one of $3^3=27$ regimes, including $\{D_1 \text{ negative}, D_2 \text{ decreasing}, G \text{ increasing}\}$ as one example.

We systematically analyze all possible regimes:
\begin{enumerate}
    \item All regimes where the generalist's utility are negative $U_G<0$ would lead the generalist to opt for the no-deal solution -- this captures a total of $3^n$ combinations. In the $n=2$ example, these scenarios are described by $\{\cdot,\cdot,G \text{ negative}\}$, so a total of $3^2=9$ out of $27$ possible regimes are accounted for.
    \item All regimes where the $U_G$ is positive-increasing and no specialist is positive-decreasing are Pareto-dominated. This is because, if no specialist is decreasing, then for some small increment of $\delta$, all players are better off. This captures a total of $2^n$ combinations. In the $n=2$ example, these scenarios account for a total of $2^2 = 4$ combinations. The four combinations would be: $\{D_1 \text{ increasing},D_2 \text{ increasing},G \text{ increasing}\}$,
$\{D_1 \text{ negative},D_2 \text{ increasing},G \text{ increasing}\}$,\\
$\{D_1 \text{ increasing},D_2 \text{ negative},G \text{ increasing}\}$,
$\{D_1 \text{ negative},D_2 \text{ negative},G \text{ increasing}\}$.
    \item All regimes where the $U_G$ is positive-decreasing and no specialist is positive-increasing are Pareto-dominated. This is because, if no specialist is decreasing, then for some small decrease of $\delta$, all players are better off. This captures a total of $2^n$ combinations. In the $n=2$ example, these scenarios account for a total of $2^2 = 4$ combinations.
    \item All regimes where $U_G$ is positive-increasing and at least one specialist is positive-decreasing are Pareto-optimal. This is because, for unimodal curves, it is impossible to improve one players' utility without a corresponding decrease in the other's. This secnario is exactly the set outside scenario (2) above. Thus, this accounts for $3^n-2^n$ cases. In the case where $n=2$, this would be 5 cases.
    \item All regimes where $U_G$ is positive-decreasing and at least one specialist is positive-increasing are Pareto-optimal. This is because, for unimodal curves, it is impossible to improve one players' utility without a corresponding decrease in the other's. This secnario is exactly the set outside scenario (3) above. Thus, this accounts for $3^n-2^n$ cases. In the case where $n=2$, this would be 5 cases.
\end{enumerate}
In total, we've defined 5 non-overlapping cases, collectively accounting for $3^n+2^n+2^n$ cases where the value of $\delta$ is not Pareto-optimal, and $(3^n-2^n)+(3^n-2^n)$ cases where the value of $\delta$ is Pareto-optimal. Overall these cases sum to $3^n+2^n+2^n+(3^n-2^n)+(3^n-2^n)=3*3^n = 3^{n+1}$, so we've accounted for every possible regime.
\end{proof} 


\subsection{Analysis for Single Bargain $\delta$}

\subsubsection{Derivation of the Subgame Perfect Equilibrium for a Given $\delta$}

We use backward induction to determine the multi-specialist fine-tuning game's subgame perfect equilibrium, as we did in Section \ref{subsec:subgame-quadratic}.

\begin{theorem}
\label{thm:multi-specialist-equilibrium}
    For a fixed $\delta$, the sub-game perfect equilibrium of the multi-specialist, single-bargain fine-tuning game with polynomial costs yields the following best-response strategies: $\alpha_0^* = \sqrt[k_0-1]{\frac{n\delta}{c_0k_0}}$; $\alpha_i^* = \sqrt[k_0-1]{\frac{n\delta}{c_0k_0}}+\sqrt[k_i-1]{\frac{1-\delta}{c_ik_i}}$.
\end{theorem}

\begin{proof}{Proof of Theorem \ref{thm:multi-specialist-equilibrium}.}
    We solve the game using backward induction as follows:

    First, starting with the last stage (3), we solve for $\alpha_i^*$ given $\alpha$ and $\delta$:
    \begin{eqnarray}\label{eq:alpha_i-star-multi-specialist}
        & \alpha_i^* = \text{argmax}_{\alpha_i}U_{D_i}(\alpha_i,\alpha_0,\delta) \nonumber \\
        \Rightarrow & \frac{\partial U_{D_i}}{\partial \alpha_i}\bigg\vert_{\alpha_i=\alpha_i^*}=0 \nonumber \\
        \Rightarrow & (1-\delta) - \frac{\partial \phi_i}{\partial \alpha_i} = 0 \nonumber \\
        \Rightarrow & (1-\delta) - c_ik_i(\alpha_i-\alpha_0)^{k_i-1} \nonumber \\
        \Rightarrow & \alpha_i^* = \alpha_0 + \sqrt[k_i-1]{\frac{1-\delta}{c_ik_i}}.
    \end{eqnarray}

    Note that $\frac{\partial^2 U_{D_i}}{\partial \alpha_i^2}= -c_ik_i(k_i-1)(\alpha_i-\alpha_0)^{k_i-2}$ is negative for all $c_i>0$ and $k_i>1$ meaning that $\alpha_i^*$ derived above yields a global maximum of $U_{D_i}$. 

    Second, knowing $D_i$'s choice of $\alpha_i^*$ above, we solve for $\alpha_0^*$ as follows:
    \begin{eqnarray*}
         & \alpha_0^* = \text{argmax}_{\alpha_0}U_G(\alpha_0,\delta) \\
         \Rightarrow & \frac{\partial U_G}{\partial \alpha_0} \bigg\vert_{\alpha_0 = \alpha_0^*}=0\\
         \Rightarrow & \frac{\partial}{\partial \alpha_0} \left(\sum_{i=1}^n\delta\alpha_i - c_0(\alpha_0)^{k_0}\right)\bigg\vert_{\alpha_0 = \alpha_0^*}=0\\
         \Rightarrow & n\delta - c_0k_0\alpha_0^{k_0-1}=0\\
         \Rightarrow & \alpha_0^* = \sqrt[k_0-1]{\frac{n\delta}{c_0k_0}}.\\
    \end{eqnarray*}
    The second derivative $\frac{\partial^2 U_{G}}{\partial \alpha_i^2}=-c_0k_0(k_0-1)\alpha_0^{k_0-2}$ is negative for all $c_i>0$ and $k_i>1$ meaning that $\alpha_0^*$ derived above yields a global maximum of $U_G$.

    Finally, plugging in $\alpha_0^*= \sqrt[k_0-1]{\frac{n\delta}{c_0k_0}}$ into Equation \ref{eq:alpha_i-star-multi-specialist}, we obtain the following expression for $\alpha_i^*$ as a function of $\delta,c_0,c_i,k_0,k_i$ only:
    $$\alpha_i^* = \sqrt[k_0-1]{\frac{n\delta}{c_0k_0}} + \sqrt[k_i-1]{\frac{1-\delta}{c_ik_i}}.$$
\end{proof}

As a corollary to Theorem \ref{thm:multi-specialist-equilibrium}, we can specify the utilities of each player as a function of $\delta$ alone:

\begin{corollary}
\label{cor:multi-specialist-utilities-solved} 
For a fixed bargaining parameter $\delta$, the players' utilities are as follows:
    \begin{equation}
    \label{eq:u-g-multi}
        U_G(\delta)=n\left(\frac{n}{c_0k_0}\right)^{\frac{1}{k_0-1}}\left(1-\frac{1}{k_0}\right)\delta^{\frac{k_0}{k_0-1}}+\delta\sum_i \left(c_ik_i\right)^{\frac{-1}{k_i-1}}(1-\delta)^{\frac{1}{k_i-1}},
    \end{equation}
    \begin{equation}
    \label{eq:u-d-i-multi}
    U_{D_i}(\delta) = \left(\frac{n}{c_0k_0}\right)^{\frac{1}{k_0-1}}\delta^{\frac{1}{k_0-1}}(1-\delta)
     +\left(c_ik_i\right)^{\frac{-1}{k_i-1}}\left(1-\frac{1}{k_i}\right)(1-\delta)^{\frac{k_i}{k_i-1}}.
    \end{equation}
\end{corollary}
\begin{proof}{Proof of Corollary \ref{cor:multi-specialist-utilities-solved}.}
    Plugging in the formulas from Theorem \ref{thm:multi-specialist-equilibrium} into equation \ref{eq:phi-0-multi}, we obtain:
\begin{eqnarray*}
    U_G = & \sum_i\delta r_i(\alpha_i)-\phi_0(\alpha_0) \\
    = & \delta \sum_i \left( \sqrt[k_0-1]{\frac{n\delta}{c_0k_0}} + \sqrt[k_i-1]{\frac{1-\delta}{c_ik_i}} \right) - c_0\left(\sqrt[k_0-1]{\frac{n\delta}{c_0k_0}}\right)^{k_0} \\
    = & \delta \sum_i \left(\frac{n\delta}{c_0k_0}\right)^{\frac{1}{k_0-1}} + \delta\sum_i\left(\frac{1-\delta}{c_ik_i}\right)^{\frac{1}{k_i-1}} - c_0\left(\frac{n\delta}{c_0k_0}\right)^{\frac{k_0}{k_0-1}} \\
    = & n\left(\frac{n}{c_0k_0}\right)^{\frac{1}{k_0-1}}\delta^{\frac{k_0}{k_0-1}}-c_0\left(\frac{n}{c_0k_0}\right)^{\frac{k_0}{k_0-1}}\delta^{\frac{k_0}{k_0-1}} + \delta\sum_i\left(\frac{1}{c_ik_i}\right)^{\frac{1}{k_i-1}}(1-\delta)^{\frac{1}{k_i-1}}\\
    = & n\left(\frac{n}{c_0k_0}\right)^{\frac{1}{k_0-1}}\left(1-\frac{1}{k_0}\right)\delta^{\frac{k_0}{k_0-1}}+\delta\sum_i \left(c_ik_i\right)^{\frac{-1}{k_i-1}}(1-\delta)^{\frac{1}{k_i-1}}.
\end{eqnarray*}

Plugging in the formulas from Theorem \ref{thm:multi-specialist-equilibrium} into equation \ref{eq:phi-i-multi}, we obtain:
\begin{eqnarray*}
    U_{D_i} = & (1-\delta) r_i(\alpha_i)-\phi_i(\alpha_i;\alpha_0) \\
    = & (1-\delta)\left( \sqrt[k_0-1]{\frac{n\delta}{c_0k_0}} + \sqrt[k_i-1]{\frac{1-\delta}{c_ik_i}}\right)-c_i\left( \sqrt[k_0-1]{\frac{n\delta}{c_0k_0}} + \sqrt[k_i-1]{\frac{1-\delta}{c_ik_i}}-\sqrt[k_0-1]{\frac{n\delta}{c_0k_0}}\right)^{k_i} \\
    = & (1-\delta) \left(\frac{n\delta}{c_0k_0}\right)^{\frac{1}{k_0-1}} + (1-\delta)\left(\frac{1-\delta}{c_ik_i}\right)^{\frac{1}{k_i-1}} - c_i\left(\frac{1-\delta}{c_ik_i}\right)^{\frac{k_i}{k_i-1}}\\
    = & \left(\frac{n}{c_0k_0}\right)^{\frac{1}{k_0-1}}\delta^{\frac{1}{k_0-1}}(1-\delta) + (c_ik_i)^{\frac{-1}{k_i-1}}\left(1-\frac{1}{k_i}\right)(1-\delta)^{\frac{k_i}{k_i-1}}. \\ 
\end{eqnarray*}
\end{proof}

\subsubsection{Deriving the Set of Pareto-Optimal Solutions}

In order to determine the set of Pareto-optimal bargaining agreements, first notice that when $k_0=2$ and $k_i=2\ \forall i$ (i.e., the case of quadratic costs), the utilities are strictly unimodal.

\begin{observation}
\label{obs:multi-quadratic-unimodal}
    In the multi-specialist fine-tuning game with quadratic costs, $U_G$ and $U_{D_i}$ are strictly unimodal functions of $\delta$, for all domain-specialists $i$.
\end{observation}
\begin{proof}{Proof of Observation \ref{obs:multi-quadratic-unimodal}.}
    Notice that the utility functions take the same form as the one-specialist case, whose analogous observation is proven in Appendix \ref{app:unimodal-one-specialist}. There, we proved that a utility of the form $A\delta^{\frac{k_0}{k_0-1}} + B \delta(1-\delta)^{\frac{1}{k_1-1}}$ is unimodal for $A,B>0$, $k_0,k_1\geq2$, $\delta\in[0,1]$.

    For simplicity, in the multi-specialist case, we prove unimodality for the case of quadratic costs (this is all we need to arrive at the bargaining solutions reported in the paper). We show that the same proof holds for both the generalist and specialists' utilities:
    \begin{itemize}
        \item For the specialist, the utility function given in Equation \ref{eq:u-d-i-multi} is already of the form proven unimodal in Proposition \ref{obs:polynomial-unimodal}.
        \item For the generalist, using Equation \ref{eq:u-g-multi}, $U_G(\delta)=n\left(\frac{n}{c_0k_0}\right)\left(1-\frac{1}{2}\right)\delta^{2}+\delta\sum_i \left(c_ik_i\right)^{-1}(1-\delta)=A'\delta^2 + B' \delta (1-\delta)$. This matches the form already proven unimodal in Proposition \ref{obs:polynomial-unimodal}. 
    \end{itemize}
 \end{proof}

This suggests that this multi-specialist fine-tuning game has a set of Pareto-optimal solutions that are characterized by Theorem \ref{thm:multi-specialist-pareto}.

\subsubsection{Powerful-Player Solutions}

\begin{proposition}[Powerful-$G$ Solution]
    The Powerful-$G$ solution to the multi-specialist fine-tuning game with quadratic costs is as follows:
\begin{equation*}
    \delta^{\text{Powerful }G}= 
    \begin{cases}
        \frac{\sum_i \frac{1}{c_i}}{2\sum_i\frac{1}{c_i}-\frac{n^2}{c_0}} & \textit{ for  } \frac{1}{n}\sum_i\frac{c_0}{c_i}>n, \\
        1 & \textit{ else.}  
    \end{cases}
\end{equation*}
\label{prop:powerful-g-multi}
\end{proposition}
\begin{proof}{Proof of Proposition \ref{prop:powerful-g-multi}.}
    The powerful-G solution is the solution that maximizes $G$'s utility:

\begin{eqnarray*}
    & \delta^{\text{Powerful }G} = \arg\max_\delta U_G (\delta) \\
    \Rightarrow & \frac{\partial U_G}{\partial \delta} = 0 \\
    \Rightarrow &  \frac{\partial}{\partial \delta}\left[n\left(\frac{n}{2c_0}\right)\left(1-\frac{1}{2}\right)\delta^{2}+\delta\sum_i \left(2c_i\right)^{-1}(1-\delta)\right] = 0 \\
    \Rightarrow & \frac{n^2}{4c_0}(2\delta) + \sum_i\frac{1}{2c_i}(1-2\delta)=0 \\
    \Rightarrow & \frac{n^2}{2c_0}(\delta) + \sum_i\frac{1}{2c_i}-\sum_i\frac{1}{c_i}(\delta)=0\\
    \Rightarrow & \delta\left(\frac{n^2}{2c_0}-\sum_i\frac{1}{c_i}\right)=-\sum_i\frac{1}{2c_i}\\
    \Rightarrow & \delta^* = \frac{\sum_i\frac{1}{2c_i}}{\sum_i\frac{1}{c_i}-\frac{n^2}{2c_0}}.
\end{eqnarray*}

The second partial derivative $\frac{\partial^2 U_G}{\partial \delta ^2}=\frac{n^2}{2c_0}-\sum_i\frac{1}{c_i}$ suggesting the critical point identified represents a maximum as long as $\frac{1}{n}\sum_i\frac{c_0}{c_i}>\frac{n}{2}$. 

We therefore have the following three cases:
\begin{itemize}
    \item For $\frac{1}{n}\sum_i \frac{c_0}{c_i}\geq n$, the expression maximizes $U_{G}$ and $\delta^{*}\in [0,1]$ so $\delta^{\text{Powerful }G}=\delta^*$ for this case.
    \item For $n>\frac{1}{n}\sum_i \frac{c_0}{c_i}>\frac{n}{2}$, the expression maximizes $U_{G}$ but $\delta^{*}>1.$ So, notice the expression is strictly increasing in $\delta \in [0,1]$, so $\delta^{\text{Powerful }G}=1$ for this case.
    \item For $\frac{1}{n}\sum_i \frac{c_0}{c_i}\leq \frac{n}{2}$, the expression \textit{minimizes} $U_{G}$ and the expression is strictly decreasing over $\delta \in [0,1]$ so $\delta^{\text{Powerful }G}=1$ for this case.
\end{itemize}
Thus:
$$\delta^{\text{Powerful }G}= 
    \begin{cases}
        \frac{\sum_i \frac{1}{c_i}}{2\sum_i\frac{1}{c_i}-\frac{n^2}{c_0}} & \textit{ for  } \frac{1}{n}\sum_i\frac{c_0}{c_i}>n, \\
        1 & \textit{ else.}  
    \end{cases}
$$
\end{proof}

\begin{proposition}[Powerful-$D_i$ Solution]
    The Powerful-$D_i$ solution to the multi-specialist fine-tuning game with quadratic costs is as follows:
\begin{equation*}
    \delta^{\text{Powerful }D_i}= 
    \begin{cases}
        \frac{2c_0-2nc_i}{c_0-2nc_i} & \textit{ for  } \frac{c_0}{c_i}<n, \\
        0 & \textit{ else.}  
    \end{cases}
\end{equation*}
\label{prop:powerful-d-i-multi}
\end{proposition}

\begin{proof}{Proof of Proposition \ref{prop:powerful-d-i-multi}.}
The powerful-$D_i$ solution maximizes $D_i$'s utility:
\begin{eqnarray*}
    & \delta^{\text{Powerful }D_i}= \arg\max_\delta U_{D_i}(\delta) \\
    \Rightarrow & \frac{\partial U_{D_i}}{\partial \delta} = 0\\
    \Rightarrow & \frac{\partial}{\partial \delta}\left[\left(\frac{n}{2c_0}\right)\delta(1-\delta)+\frac{1}{2c_i}\left(1-\frac{1}{2}\right)(1-\delta)^2\right] = 0\\
    \Rightarrow & \left(\frac{n}{2c_0}\right)(1-2\delta)-2\frac{1}{2c_i}\left(\frac{1}{2}\right)(1-\delta) = 0\\
    \Rightarrow & \frac{n}{2c_0}-\frac{n}{c_0}\delta-\frac{1}{2c_i}+\frac{1}{2c_i}\delta = 0\\
    \Rightarrow & \left(\frac{1}{2c_i}-\frac{n}{c_0}\right)\delta = \frac{1}{2c_i}-\frac{n}{2c_0}\\
    \Rightarrow & \delta^* = \frac{\frac{1}{2c_i}-\frac{n}{2c_0}}{\frac{1}{2c_i}-\frac{n}{c_0}}\\
    \Rightarrow & \delta^* = \frac{2c_0-2nc_i}{c_0-2nc_i}.\\
\end{eqnarray*}

The second partial derivative $\frac{\partial^2 U_{D_i}}{\partial \delta ^2}=\frac{1}{2c_i}-\frac{n}{c_0}$, meaning the critical point identified represents a maximum so long as $\frac{c_0}{c_i}<2n$.

We therefore have the following three cases:
\begin{itemize}
    \item For $c_i\geq \frac{c_0}{n}$, the expression maximizes $U_{D_i}$ and $\delta^{*}\in [0,1]$, so $\delta^{\text{Powerful }D_i}=\delta^{*}$ for this case.
    \item For $ \frac{c_0}{n} > c_i > \frac{c_0}{2n}$, the expression maximizes $U_{D_i}$ but $\delta^{*}<0.$ So, notice the expression is strictly decreasing in $\delta \in [0,1]$, so $\delta^{\text{Powerful }D_i}=0$ for this case.
    \item For $\frac{c_0}{2n}\geq c_i$, the expression \textit{minimizes} $U_{D_i}$ and the expression is strictly decreasing in $\delta \in [0,1]$ so $\delta^{\text{Powerful }D_i}=0$ for this case.
\end{itemize}
Thus:
$$\delta^{\text{Powerful }D_i}= 
    \begin{cases}
        \frac{2c_0-2nc_i}{c_0-2nc_i} & \textit{ for  } \frac{c_0}{c_i}<n, \\
        0 & \textit{ else.}  
    \end{cases}$$
\end{proof}

\subsubsection{Vertical Monopoly Solution}

 \begin{proposition}[Vertical Monopoly Solution]
\label{prop:vertical-monopoly-multi-specialist}
   The Vertical Monopoly Bargaining Solution to the multi-specialist fine-tuning game with quadratic costs is as follows:
    \begin{equation*}
    \delta^{\textit{Vertical Monopoly}} = \frac{n^2}{n^2+\sum_i \frac{c_0}{c_i}}.
    \end{equation*}
\end{proposition}

\begin{proof}{Proof of Proposition \ref{prop:vertical-monopoly-multi-specialist}.}
    The vertical monopoly bargaining solution maximizes the sum of utilities:
    \begin{eqnarray*}
        & \delta^{VM} = \arg\max_\delta\left(U_G+\sum_iU_{D_i}\right)\\
        \Rightarrow & \frac{\partial}{\partial \delta}\left(U_G+\sum_iU_{D_i}\right) = 0\\
        \Rightarrow & \frac{\partial}{\partial \delta}\left[\frac{n^2}{4c_0}\delta^2 + \delta(1-\delta)\sum_i\frac{1}{2c_i}+\sum_i\left(\frac{1}{4c_i}+\left(\frac{n}{2c_0}-\frac{1}{2c_i}\right)\delta + \left(\frac{1}{4c_i}-\frac{n}{2c_0}\right)\delta^2\right)\right] = 0\\
        \Rightarrow & \frac{n^2}{2c_0}\delta + (1-2\delta)\left(\sum_i\frac{1}{2c_i}\right) + \sum_i\left(\frac{n}{2c_0}-\frac{1}{2c_i} + 2\left(\frac{1}{4c_i}-\frac{n}{2c_0}\right)\delta\right)=0 \\
        \Rightarrow & \frac{n^2}{2c_0}\delta + (1-2\delta)\left(\sum_i\frac{1}{2c_i}\right) + \frac{n^2}{2c_0} - \frac{2n^2\delta}{2c_0} + \sum_i\left(\frac{1}{2c_i}\delta - \frac{1}{2c_i}\right) = 0 \\
        \Rightarrow & (\frac{1}{2}-\delta)\left(\sum_i\frac{1}{c_i}\right) + (1-\delta)\frac{n^2}{2c_0} + \frac{1}{2}(\delta-1)\left(\sum_i\frac{1}{c_i}\right) = 0 \\
        \Rightarrow & -\frac{1}{2}\delta\left(\sum_i\frac{1}{c_i}\right) + (1-\delta)\frac{n^2}{2c_0} = 0 \\
        \Rightarrow & \frac{1}{2}\delta\left(\sum_i\frac{1}{c_i}\right) + \delta\frac{n^2}{2c_0} = \frac{n^2}{2c_0} \\
        \Rightarrow & \delta^* = \frac{\frac{n^2}{c_0}}{\left(\sum_i\frac{1}{c_i}\right) + \frac{n^2}{c_0}} \\
        \Rightarrow & \delta^* = \frac{n^2}{{n^2} + \sum_i\frac{c_0}{c_i}}. \\
    \end{eqnarray*}

    The second partial derivative $\frac{\partial^2}{\partial \delta ^2}\left(U_G+\sum_iU_{D_i}\right) = -\frac{1}{2}\left(\sum_i\frac{1}{c_i}\right) - \frac{n^2}{2c_0}$ which is negative $\forall c_0,c_i>0$. Thus, $\delta^{VM}=\delta^*$ maximizes the sum of utilities.
\end{proof}

\subsubsection{Egalitarian Bargaining Solution}

 \begin{proposition}[Egalitarian Solution]
\label{prop:egal-multi-specialist}
   The Egalitarian Bargaining Solution to the multi-specialist fine-tuning game with quadratic costs is given by:
    \begin{equation*}
    \delta^{\textit{Egal.}} = \frac{\sqrt{c_0^2K^2-2c_0nK+\frac{c_0}{c_{\text{max}}}n^2+n^2} + n - \frac{c_0}{c_{\text{max}}}n-c_0K}{n^2+2n-\frac{c_0}{c_{\text{max}}}-2c_0K}.    \end{equation*}
    Where $K = \sum_{i}\frac{1}{c_i}$ and $c_{\text{max}}=\max_{i}c_i$.
\end{proposition}
\begin{proof}{Proof of Proposition \ref{prop:egal-multi-specialist}.}

The egalitarian solution is given by:
$$
\delta^{\text{Egal.}} = \arg\max_\delta \left[\min\left(U_G, \min_i\left(U_{D_i}\right)\right)\right].
$$
    Finding a closed-form solution for the egalitarian bargaining agreement proves more challenging in the multi-specialist case because many players in the fine-tuning game, could, potentially, be worst-off. Instead of simply equalizing utilities, the multi-specialist case might have an egalitarian solution where only the two worst-off players have equal utility. 
    
    The proof for quadratic costs starts off by observing that the specialist with the lowest utility must be the specialist $j$ with the highest cost coefficient $c_j$. This is because revenue function $r(\cdot) = I(\cdot)$ is constant across all domains, as is the bargaining parameter $\delta$ and coefficient $k_i=2$. Thus, the highest-cost-coefficient domain specialist is given by $j = \arg\max_i c_i$. This observation leaves us with two possibilities for the lowest-utility player: 1) Player $D_j$ 2) Player $G$. Thus, the egalitarian solution is the bargaining agreement that equalizes the utilities between the generalist and the costliest player: $$U_G(\delta^{\text{Egal.}}) = U_{D_j}(\delta^{\text{Egal.}}).$$

    Plugging in, we get:
    \begin{eqnarray*}
        & \frac{n^2}{4c_0}\delta^2 + \delta \sum_i \frac{1}{2c_i}(1-\delta) = \frac{1}{4c_{\text{max}}} + \left(\frac{n}{2c_0}-\frac{1}{2c_{\text{max}}}\right)\delta + \left(\frac{1}{4c_{\text{max}}}-\frac{n}{2c_0}\right)\delta^2\\
        \Rightarrow & \left[\frac{n^2}{4c_0} - \sum_i\frac{1}{2c_i} + \frac{n}{2c_0}-\frac{1}{4c_{\text{max}}}\right]\delta^2 + \left[\sum_i \frac{1}{2c_i} + \frac{1}{2c_{\text{max}}}-\frac{n}{2c_0}\right]\delta - \frac{1}{4c_{\text{max}}} = 0.
    \end{eqnarray*}
    Plugging in to the quadratic formula, we get two candidate solutions:
    $$
    \delta^{\textit{Egal.}} \stackrel{?}{=} \left\{
    \frac{\pm\sqrt{c_0^2(\sum_{i}\frac{1}{c_i})^2-2c_0n(\sum_{i}\frac{1}{c_i})+\frac{c_0}{c_{\text{max}}}n^2+n^2} + n - \frac{c_0}{c_{\text{max}}}n-c_0(\sum_{i}\frac{1}{c_i})}{n^2+2n-\frac{c_0}{c_{\text{max}}}-2c_0(\sum_{i}\frac{1}{c_i})}
    \right\}.
    $$
As a last step, we select the candidate solution that yields a feasible $\delta \in [0,1]$:
$$
    \delta^{\textit{Egal.}} = 
    \frac{\sqrt{c_0^2(\sum_{i}\frac{1}{c_i})^2-2c_0n(\sum_{i}\frac{1}{c_i})+\frac{c_0}{c_{\text{max}}}n^2+n^2} + n - \frac{c_0}{c_{\text{max}}}n-c_0(\sum_{i}\frac{1}{c_i})}{n^2+2n-\frac{c_0}{c_{\text{max}}}-2c_0(\sum_{i}\frac{1}{c_i})}.
    $$
\end{proof}

\subsubsection{Solution that maximizes the technology's performance}

\begin{definition}[Maximum-performance solution]
    For the multi-specialist fine-tuning game, the maximum-performance bargaining solution is the feasible revenue-sharing agreement $\delta^{\text{max-}\sum_i\alpha_i^*} \in [0,1]$ that maximizes the sum of the technology's domain-specific performances: $\delta^{\text{max-}\sum_i\alpha_i^*} = \arg\max_{\delta \in [0,1]}\sum_i\alpha_i^*$.
    \label{def:max-performance-multi}
\end{definition}

\begin{proposition}[Maximum-$\sum_i\alpha_i^*$ Solution]
    A bargaining solution that maximizes the technology's domain-specific performances is given by:
    \begin{equation*}
        \delta^{\text{Max-}\sum_i\alpha_i^*} = \begin{cases}
            0 & \textit{  for } \frac{1}{n} \sum_i \frac{c_0}{c_i}>n, \\
            1 & \textit{  else.}
        \end{cases}
    \end{equation*}
    \label{prop:max-performance-multi}
\end{proposition}

\begin{proof}{Proof of Proposition \ref{prop:max-performance-multi}.}
    The maximum-performance solution to the multi-specialist fine-tuning game maximizes the sum of performances:

    \begin{eqnarray*}
        & \delta^{\text{Max-}\sum_i\alpha_i^*} = \arg\max_\delta \sum_i \alpha_i^*\\
        \Rightarrow & \frac{\partial}{\partial \delta} \sum_i \alpha_i \\
        = & \frac{\partial}{\partial \delta} \sum_i \left[\frac{n}{2c_0}\delta + (1-\delta)\frac{1}{2}\sum_j\frac{1}{c_j}\right]\\
        = & \frac{n^2}{2c_0}+ (-1)\frac{1}{2}\sum_i\frac{1}{c_i}\\
        = & \frac{n^2}{2c_0} -\sum_i\frac{1}{2c_i}. \\
    \end{eqnarray*}
    Notice this quantity is constant over the domain $\delta\in[0,1]$ and is either non-positive or non-negative. We therefore have the following cases:
    \begin{itemize}
        \item If $\sum_i\frac{1}{c_i}> \frac{n^2}{c_0}$, then the sum of downstream performances $\sum_i\alpha_i$ is decreasing in $\delta$ so $\delta^{\text{Max-}\sum_i\alpha_i^*}=0$.
        \item Otherwise, the sum of downstream performances $\sum_i\alpha_i$ is increasing in $\delta$ so $\delta^{\text{Max-}\sum_i\alpha_i^*}=1$. 
    \end{itemize}
\end{proof}

\subsection{Analysis for Multiple Bargaining Parameters $\vec\delta \in \mathbb{R}^n$}

\subsubsection{Derivation of the Subgame Perfect Equilibria Strategies}

We use backward induction to determine the multi-specialist fine-tuning game's subgame perfect equilibrium.

\begin{theorem}
\label{thm:multi-multi-subgame-perfect-eq}
    Holding $\vec\delta$ fixed, the sub-game perfect equilibrium of the multi-specialist multiple-bargain fine-tuning game with polynomial costs yields the following best-response strategies: 
    $\alpha_0^* = \sqrt[k_0-1]{\frac{\sum_j\delta_j}{c_0k_0}}$; $\alpha_i^* = \sqrt[k_0-1]{\frac{\sum_j\delta_j}{c_0k_0}}+\sqrt[k_i-1]{\frac{1-\delta_i}{c_ik_i}}$.
\end{theorem}

\begin{proof}{Proof of Theorem \ref{thm:multi-multi-subgame-perfect-eq}.}
We solve the game using backward induction. 

First, we solve for $\alpha_i^*$ given $\alpha$ and $\delta$:
    \begin{eqnarray}\label{eq:alpha_i-star-multi-multi}
        & \alpha_i^* = \text{argmax}_{\alpha_i}U_{D_i}(\alpha_i,\alpha_0,\vec\delta) \nonumber \\
        \Rightarrow & \frac{\partial U_{D_i}}{\partial \alpha_i}\bigg\vert_{\alpha_i=\alpha_i^*}=0 \nonumber \\
        \Rightarrow & (1-\delta_i) - \frac{\partial \phi_i}{\partial \alpha_i} = 0 \nonumber \\
        \Rightarrow & (1-\delta_i) - c_ik_i(\alpha_i-\alpha_0)^{k_i-1} \nonumber \\
        \Rightarrow & \alpha_i^* = \alpha_0 + \sqrt[k_i-1]{\frac{1-\delta_i}{c_ik_i}}.
    \end{eqnarray}

    Note that $\frac{\partial^2 U_{D_i}}{\partial \alpha_i^2}= -c_ik_i(k_i-1)(\alpha_i-\alpha_0)^{k_i-2}$ is negative for all $c_i>0$ and $k_i>1$ meaning that $\alpha_i^*$ derived above yields a global maximum of $U_{D_i}$. 

    Second, knowing $D_i$'s choice of $\alpha_i^*$ above, we solve for $\alpha_0^*$ as follows:
    \begin{eqnarray*}
         & \alpha_0^* = \text{argmax}_{\alpha_0}U_G(\alpha_0,\delta) \\
         \Rightarrow & \frac{\partial U_G}{\partial \alpha_0} \bigg\vert_{\alpha_0 = \alpha_0^*}=0\\
         \Rightarrow & \frac{\partial}{\partial \alpha_0} \left(\sum_{i=1}^n\delta_i\alpha_i - c_0(\alpha_0)^{k_0}\right)\bigg\vert_{\alpha_0 = \alpha_0^*}=0\\
         \Rightarrow & \sum_i\delta_i - c_0k_0\alpha_0^{k_0-1}=0\\
         \Rightarrow & \alpha_0^* = \sqrt[k_0-1]{\frac{\sum_i\delta_i}{c_0k_0}}.\\
    \end{eqnarray*}
    The second derivative $\frac{\partial^2 U_{G}}{\partial \alpha_i^2}=-c_0k_0(k_0-1)\alpha_0^{k_0-2}$ is negative for all $c_i>0$ and $k_i>1$ meaning that $\alpha_0^*$ derived above yields a global maximum of $U_G$.

    Finally, plugging in $\alpha_0^*= \sqrt[k_0-1]{\frac{\sum_i\delta_i}{c_0k_0}}$ into Equation \ref{eq:alpha_i-star-multi-multi}, we obtain the following expression for $\alpha_i^*$ as a function of $\delta,c_0,c_i,k_0,k_i$ only:
    $$\alpha_i^* = \sqrt[k_0-1]{\frac{\sum_j\delta_j}{c_0k_0}} + \sqrt[k_i-1]{\frac{1-\delta_i}{c_ik_i}}.$$
    This finishes the proof.
\end{proof}

\begin{corollary}
\label{cor:multi-multi-utilities-solved} 
For fixed bargaining parameters $\vec\delta$, the players' utilities are as follows:
    \begin{equation}
    \label{eq:u-g-multi-multi}
        U_G(\vec\delta)=\left(1-\frac{1}{k_0}\right)  \left(\frac{1}{c_0k_0}\right)^{\frac{1}{k_0-1}} \left(\sum_i \delta_i\right)^{\frac{k_0}{k_0-1}}+\sum_i \left(c_ik_i\right)^{\frac{-1}{k_i-1}}\delta_i(1-\delta_i)^{\frac{1}{k_i-1}},
    \end{equation}
    \begin{equation}
    \label{eq:u-d-i-multi-multi}
    U_{D_i}(\vec\delta) = \left(\frac{1}{c_0k_0}\right)^{\frac{1}{k_0-1}}\left(\sum_j\delta_j\right)^{\frac{1}{k_0-1}}(1-\delta_i)
     +\left(c_ik_i\right)^{\frac{-1}{k_i-1}}\left(1-\frac{1}{k_i}\right)(1-\delta_i)^{\frac{k_i}{k_i-1}}.
    \end{equation}
\end{corollary}

The arithmetic required to arrive at the above closed-form solutions is omitted for brevity. The steps are nearly identical to those demonstrated in Corollary \ref{cor:multi-specialist-utilities-solved}.

\subsubsection{Powerful-Player Solutions}

\begin{proposition}[Powerful-$G$ Solution] The Powerful-$G$ solution to the multi-specialist, multi-bargain fine-tuning game with quadratic costs is given by the following algorithm: 

\begin{algorithm}
\caption{Powerful-$G$ Bargaining Solution}\label{alg:PG-solution-new}
\begin{algorithmic}
\State \textbf{Given: } Generalist cost $c_0$, and list of specialist costs $\vec{c}$ in descending order of magnitude.
\State $i \leftarrow 1$
\While{$i\leq n$ and $\frac{c_i}{c_0} > \frac{2 - \sum_{j=i+1}^n \frac{c_j}{c_0}}{n+i}$} \Comment{Identify which solutions are on the constraint}
    \State $i\leftarrow i + 1$
\EndWhile 

\State $C \in \mathbb{R}^{n\times n} \leftarrow \text{diag}(\vec{c} )$

\State
$Q \in \mathbb{R}^{n\times n} \leftarrow \frac{1}{2c_0}\left[\vec1 \cdot \vec1^T\right]- C^{-1}$ \Comment{The quadratic term in the objective, $U_G$}

\State $l \in \mathbb{R}^{n}  \leftarrow \frac{1}{2}C^{-1}\vec{1}$ \Comment{The linear term in the objective, $U_G$}


\State $E \in \mathbb{R}^{n\times i}  \leftarrow \left[\begin{array}{cc}
    I^{i \times i} & 0^{(n-i)\times i}
\end{array}\right]$

\State $S \in \mathbb{R}^{i\times i} \leftarrow$ the $i\times i$ upper left sub-matrix of $-Q^{-1}$ \Comment{\small{Schur complement of KKT matrix }$\left[\begin{array}{cc}
Q & E^T\\
E & 0
\end{array}
\right]$}

\State \textbf{Return}: $$\vec\delta^* = \left(Q^{-1} + Q^{-1} E^T S^{-1} E Q^{-1} \right)(-l) - Q^{-1} E^T S^{-1} \vec{1}$$

\end{algorithmic}
\end{algorithm}
\label{prop:multi-multi-powerful-g}
\end{proposition}



\begin{proof}{Proof of Proposition \ref{prop:multi-multi-powerful-g}.}

The solution to the powerful-G bargaining solution is given by the following optimization problem:

\begin{eqnarray}
\label{eq:opt-convex-multi-multi}
\arg\max_{\Vec{\delta}} & U_G \\
\text{s.t.} & \delta_j \in [0,1]\ \forall j \in \{1,...,n\} \nonumber 
\end{eqnarray}

In its current form, this optimization appears intractable: It is a quadratic program and its objective is not, necessarily, concave. Here we show that this optimization can be reduced to one that has a closed form, efficient and exact solution. Our broad strategy is to create a simple algorithm for determining which domains will reside on the constraint boundary ($\delta_i=1$). Once we consider only those domains for which the solution is unconstrained, we'll see that the objective enjoys concavity properties as a function of those variables, and so we'll be left with a solvable quadratic program.

The proof proceeds with a sequence of Lemmas that will help us. All Lemmas are proven in Appendix \ref{subsec:deferred-lemmas-multi-multi-powerful-G}.
\begin{lemma}
    \label{lemma:powerful-g-multiple-bargaining-individual-soln}
        In domain $i$, taking all other bargains $\delta_{k\neq i}$ as fixed, the powerful-G solution is:
        $$\delta_i^{P. G} = \begin{cases}
            1 & \text{for }\frac{c_i}{c_0} \geq \frac{1}{\sum_{k:k\neq i}\delta_k + 1},\\
            \frac{1}{2-\frac{c_i}{c_0}}\left(\frac{c_i}{c_0}\sum_{k:k\neq i}\delta_k + 1\right) & \text{else.}
        \end{cases}
        $$
    \end{lemma}
\begin{lemma}
    \label{lemma:powerful-G-non-decreasing-ci}
        In domain $i$, taking all other bargains $\delta_{k\neq i}$ as fixed, the powerful-G solution $\delta_i^{P. G}$ is a non-decreasing function of $c_i$.
    \end{lemma}
\begin{lemma}
    \label{lemma:powerful-G-ci-at-least-half}
        For any domain $i$, $\delta_i^{P. G}\geq \frac{1}{2}$. 
\end{lemma}

Now, the crucial task is to determine which domains $i$ have $\delta_i=1$. Observe that from a set of potential domains, if at least one domain has $\delta_i=1$, then the domain $i$ for which $c_i>c_j \ \forall j$ must have $\delta_i=1$. For ease and without loss of generality, we'll therefore assume the set of domains is in descending order of value $c_i$, and we'll consider each domain, one at a time, to determine whether the domain has $\delta_i=1$. We can stop once we come across a value $c_i$ for which $\delta_i<1$. The condition for the solution being on the constraint is given below:

\begin{lemma}
\label{lemma-boundary-procedure-pg-multi-multi}
Given specialist costs $\vec{c}$ in descending order of magnitude. If $i=1$ or $\delta_{i-1}^{PG}=1$, $\delta_i^{PG} = 1 \iff c_i(n+i) \geq 2c_0- \sum_{j>i}c_j$.
\end{lemma}

Now that we've specified the set of domains which will end up on the boundary where $\delta_i=1$, our optimization can be stated as an equality-constrained quadratic optimization. If domains $\{1,...,i^*\}$ have $\delta_i=1$ and domains $\{i^*+1,...,n\}$ have $\delta_i<1$, as established using the procedure in Lemma \ref{lemma-boundary-procedure-pg-multi-multi}, the powerful-G bargain is reduced to the following optimization: 
$\arg\max_{\Vec{\delta}} U_G \ \text{s.t.} \  \delta_j = 1 \ \forall \ j \in \{1,...,i^*\}$. The remainder of the proof consists in showing that the objective is concave in the remaining variables $\{\delta_{i^* + 1},...,\delta_n\}$. 
\begin{lemma}
\label{lemma:powerful-g-multi-multi-concave}
    $U_G$ is a concave quadratic function of all variables $\{\delta_{i^*+1},..,\delta_n\}$ for which $\delta_i^{PG}<1$. 
\end{lemma}

Thus we've established $U_G$ is a concave quadratic function of the remaining decision variables, meaning that the critical points of the remaining quadratic program represent a maximum. Thus, we solve the constrained quadratic program using the typical (KKT) constrained approach. We use Lagrangian optimization, where we maximize the objective $U_G$ less a weighted term for each equality constraint: $\max_{\delta} \mathcal{L}$ where $\mathcal{L} := U_G - \vec{\lambda}^T( \vec{\delta} - \vec{1}) = \frac{1}{4c_0}\vec\delta^T [\vec{1} \vec{1}^T]\vec\delta + \frac{1}{2}\vec{1}^T C^{-1}\vec\delta - \frac{1}{2}\vec{\delta}^T C^{-1}\vec{\delta} - \vec{\lambda} (\vec\delta - \vec{1})$.
Using the notation introduced in the Proposition statement, the critical point is found using the following system of linear equations:
\begin{eqnarray*}
    \left[\begin{array}{cc}
        Q & E^T \\
        E & 0
    \end{array}\right]
    \left[\begin{array}{c}
        \vec{\delta}^{PG} \\
        \vec{\lambda}
    \end{array}\right]
    = \left[\begin{array}{c}
        -l  \\
         \vec{1}
    \end{array}\right].
\end{eqnarray*}
Omitting the algebra, we get the following solution for the critical point, $\vec{\delta}^*$, which the above analysis implies is the unique global maximum of $U_G$:
$$\vec\delta^* = \left(Q^{-1} + Q^{-1} E^T S^{-1} E Q^{-1} \right)(-l) - Q^{-1} E^T S^{-1} \vec{1}.$$
\end{proof}

\begin{proposition}[Powerful-$D_i$ Solution] The Powerful-$D_i$ solution to the multi-specialist, multi-bargain fine-tuning game with quadratic costs is as follows:    
$$\delta_j^{\textit{Powerful }D_i}=\begin{cases}
    1 \ \ \textit{for }j\neq i, \\
    0 \ \ \textit{for }j=i.
\end{cases}$$
\label{prop:multi-multi-powerful-d}
\end{proposition}

\begin{proof}{Proof of Proposition \ref{prop:multi-multi-powerful-d}.}
    The powerful-$D_i$ solution is the solution that maximizes $D_i$'s utility. Assuming $D_i$ has control over every bargain, $\Vec\delta$, the solution is given by the values that solve the following optimization: $$ \Vec\delta^{P.D_i}=\arg\max_{\Vec\delta} U_{D_i}$$
    We will first handle the values $\delta_{j\neq i}^{P.D_i}$, that is, all bargains other than $i$. We will show that the utility of $D_i$ is non-decreasing with respect to all $\delta_{j\neq i}$, which implies $\delta_{j}^{P.D_i}=1 \forall j\neq i$. Next, we will solve the resulting one-dimensional optimization problem for $\delta_{i}^{P.D_i}$ assuming $\delta_{j\neq i}=1$. 

\begin{lemma}
\label{lemma-u-d-non-decreasing}
    $U_{D_i}$ is non-decreasing in $\delta_{k}$ for all $k \neq i$.
\end{lemma}

See Appendix \ref{subsec:deferred-lemmas-multi-multi-powerful-D} for a proof of the Lemma. An immediate corollary to Lemma \ref{lemma-u-d-non-decreasing} is that for all bargains $k \neq i$, $\delta_k^{P.D_i}=1$.

\begin{corollary}
\label{corollary-u-d-non-decreasing}
    For all domains $k \neq i$, $\delta_k^{P.D_i}=1$.
\end{corollary}
    
We've now solved the Powerful-$D_i$ bargaining solutions for all domains except $i$. This greatly simplifies our task: instead of solving a many-dimensional optimization problem over $\Vec\delta$, we now must solve a one-dimensional optimization problem over $\delta_i$.

\begin{lemma}[Powerful-$D_i$ Solution for bargain $i$ alone]
\label{powerful-d-multi-quadratic-fixed-j}
Holding all other bargains $\delta_{j\neq i}$ fixed, the Powerful-$D_i$ solution to the fine-tuning game with quadratic costs is as follows:
$$\delta_i^{\textit{Powerful }D_i} = 
        \begin{cases}
        \max \left[0,  \frac{c_0-c_i\left(1-\sum_{j:j\neq i}\delta_j\right)}{c_0-2c_i}\right] & \textit{ for  } c_0 < 2c_i, \\
        0 & \textit{ for  } c_0 \geq 2c_i.
        \end{cases}
        $$
\end{lemma}

The proof of the above Lemma is provided in Appendix \ref{subsec:deferred-lemmas-multi-multi-powerful-D}. The final step that remains is to plug in our result from Corollary \ref{corollary-u-d-non-decreasing} to our formulation for $\delta_i$ in Lemma \ref{powerful-d-multi-quadratic-fixed-j}:
\begin{itemize}
    \item When $c_0 \geq 2c_i$, we've already established $\delta_i^{P.D_i}=0$.
    \item When $c_0 < 2c_i$, we've established $\delta_i^{P.D_i} = \max \left[0,  \frac{c_0-c_i\left(1-\sum_{j:j\neq i}\delta_j\right)}{c_0-2c_i}\right]$. As long as there exists at least one other specialist $k\neq i$, notice that $\sum_{j:j\neq i}\delta_j = \sum_{j:j\neq i}1 \geq 1$. Notice the second value has negative denominator, and the numerator must be positive. Thus, $\delta_i^{P.D_i}=0$.
\end{itemize}

Corollary \ref{corollary-u-d-non-decreasing} establishes that in any case where there exists some $j:j\neq i$, $\delta_j^{P.D_i}=1$. The two bullet points above establish that in these cases, $\delta_i^{P.D_i}=0$. Thus we've proven the proposition.
\end{proof}

\subsubsection{Vertical Monopoly (Utilitarian) Solution}

\begin{proposition}[Vertical Monopoly Solution]
\label{prop:vertical-monopoly-multi-multi}
   In the fine-tuning game with multiple specialists and multiple bargaining parameters, the Vertical Monopoly Bargaining Solution to the multi-specialist, multi-bargain fine-tuning game with quadratic costs is is given by Algorithm \ref{alg:VM-solution-multi-multi}.

\begin{algorithm}
\caption{Vertical Monopoly Bargaining Solution}\label{alg:VM-solution-multi-multi}
\begin{algorithmic}
\State \textbf{Given: } Generalist cost $c_0$, and list of specialist costs $\vec{c}$ in descending order of magnitude.
\State $i \leftarrow 1$
\While{$i\leq n$ and $c_i(n-i) \geq c_0 + \sum_{j=i}^nc_j$} \Comment{Identify which solutions are on the constraint}
    \State $i\leftarrow i + 1$ 
\EndWhile 

\State $C \in \mathbb{R}^{n\times n} \leftarrow \text{diag}(\vec{c} )$

\State
$Q \in \mathbb{R}^{n\times n} \leftarrow \frac{-1}{2c_0}\left[\vec1 \cdot \vec1^T\right]- \frac{1}{2} C^{-1}$ 
\Comment{The quadratic term in the objective, $U_G + \sum_i U_{D_i}$}
\State $l \in \mathbb{R}^{n}  \leftarrow \frac{n}{2c_0}\vec{1}$ \Comment{The linear term in the objective, $U_G + \sum_i U_{D_i}$}
\State $E \in \mathbb{R}^{n\times i}  \leftarrow \left[\begin{array}{cc}
    I^{i \times i} & 0^{(n-i)\times i}
\end{array}\right]$

\State $S \in \mathbb{R}^{i\times i} \leftarrow$ the $i\times i$ upper left sub-matrix of $-Q^{-1}$ \Comment{\small{Schur complement of KKT matrix }$\left[\begin{array}{cc}
Q & E^T\\
E & 0
\end{array}
\right]$}

\State \textbf{Return}: $$\vec\delta^* = \left(Q^{-1} + Q^{-1} E^T S^{-1} E Q^{-1} \right)(-l) - Q^{-1} E^T S^{-1} \vec{1}$$

\end{algorithmic}
\end{algorithm}
\end{proposition}

\begin{proof}{Proof of Proposition \ref{prop:vertical-monopoly-multi-multi}.}
The solution to the vertical monopoly (VM) bargaining solution is given by the following optimization problem:

\begin{eqnarray}
\label{eq:opt-convex}
\arg\max_{\Vec{\delta}} & U_G + \sum_i U_{D_i} \\
\text{s.t.} & \delta_j \in [0,1]\ \forall j \in \{1,...,n\} \nonumber 
\end{eqnarray}

The proof proceeds almost identically to the Powerful-G bargaining solution, however, this utilitarian solution is markedly `easier' because the objective is \textit{always} concave maximization, so we do not need to do the work to reduce the problem to one that enjoys convexity properties. Here, we show the objective is concave. 
Given this condition, solving the (convex) quadratic program proceeds identically.

\begin{lemma}
\label{lemma:VM-objective-concave-quadratic-multi-multi}
The objective $U_G + \sum_i U_{D_i}$ is a concave quadratic function of $\vec\delta$.
\end{lemma}
\begin{proof}{Proof of Lemma \ref{lemma:VM-objective-concave-quadratic-multi-multi}.}
    We show the quadratic term of the objective is negative definite, which implies the Hessian is negative definite. 

\begin{eqnarray*}
    U_G + \sum_i U_{D_i} = & \frac{1}{4c_0}\vec\delta^T [\vec{1} \vec{1}^T]\vec\delta + \frac{1}{2}\vec{1}^T C^{-1}\vec\delta - \frac{1}{2}\vec{\delta}^T C^{-1}\vec{\delta} + \sum_i \left(\frac{1}{2c_0}\right)\left(\sum_j \delta_j\right)\left(1-\delta_i\right) + \frac{1}{4c_i}\left(1-\delta_i\right)^2\\
    = & \vec\delta^T\left[-\frac{1}{4}C^{-1} - \frac{1}{4c_0}[\vec{1}\vec{1}^T]\right]\vec\delta + \vec{1}^T\left[\frac{n}{2c_0}\right]\vec\delta + \vec{1}^T\left[\frac{1}{4}C^{-1}\right]\vec{1}
\end{eqnarray*}
The Hessian is $\frac{1}{2}$ the quadratic matrix term: $H= [-\frac{1}{2}C^{-1} - \frac{1}{2c_0}[\vec{1}\vec{1}^T]]$. The full specification of each entry in the Hessian matrix is given by:
$$H_{ii}=-\frac{1}{2c_0}-\frac{1}{2c_i}\ ;\ H_{ij}=-\frac{1}{2c_0}.$$
By identical logic used in Lemma \ref{lemma:powerful-g-multi-multi-concave}, notice that $-H$ is a diagonal-dominant Hermitian matrix with real and positive diagonals, which must be positive definite. Thus $H$ is negative definite. 
\end{proof}

    


The remainder consists in identifying the boundary cases for which $\delta_i^{VM}=1$ and then solving the equality-constrained, convex optimization exactly as demonstrated in the Powerful-G solution.
\end{proof}

\subsection{Theorem on the Three Specialist Regimes}
\label{app:regimes-proof}

\begin{proof}{Proof of Theorem \ref{thm:specialist-regimes}.}
    We prove this theorem in a sequence of Lemmas. The proof follows for any given specialist $D_i$ and revenue-sharing parameter $\delta$.

    \begin{lemma}
    \label{lemma:regimes1}
        If fixed costs are under control, meaning $r_i(\alpha_0)>\frac{1}{1-\delta}\phi_i(\alpha_0)$, then $D_i$ will not abstain -- instead, $D_i$ would always prefer to free-ride.
    \end{lemma}
    
    \begin{proof}{Proof.} If $r_i(\alpha_0)>\frac{1}{1-\delta}\phi_i(\alpha_0)$, then $U_{D_i}\big\vert_{\alpha_i=\alpha_0} = r_i(\alpha_0)-\frac{1}{1-\delta}\phi_i(\alpha_0)$ is simply the RHS minus the LHS of the inequality. This means $U_{D_i}$ must be positive at $\alpha_i=\alpha_0$. Thus, as long as fixed costs are under control, the specialist prefers free-riding to abstaining.
    \end{proof}

    \begin{lemma}
    \label{lemma:regimes2}
        If fixed costs are not under control, meaning $r_i(\alpha_0)<\frac{1}{1-\delta}\phi_i(\alpha_0)$, then $D_i$ will not free-ride -- instead, $D_i$ would always prefer to abstain.
    \end{lemma}

     \begin{proof}{Proof.} If $r_i(\alpha_0)<\frac{1}{1-\delta}\phi_i(\alpha_0)$, then $U_{D_i}\big\vert_{\alpha_i=\alpha_0} = r_i(\alpha_0)-\frac{1}{1-\delta}\phi_i(\alpha_0)$ is simply the RHS minus the LHS of the inequality. This means $U_{D_i}$ must be negative at $\alpha_i=\alpha_0$. Thus, as long as fixed costs are not under control, the specialist prefers abstaining to free-riding.
     \end{proof}

    \begin{lemma}
    \label{lemma:regimes3}
        If it is marginally profitable to invest in the technology, meaning $r_i'(\alpha_0) > \frac{1}{1-\delta}\phi_i'(\alpha_0)$, then $D_i$ will not free-ride -- instead, $D_i$ would always prefer to contribute.
    \end{lemma}

    \begin{proof}{Proof.} If $r_i'(\alpha_0) > \frac{1}{1-\delta}\phi_i'(\alpha_0)$, then $\frac{\partial U_{D_i}}{\partial \alpha_i}\big\vert_{\alpha_i=\alpha_0} = r_i'(\alpha_0)-\frac{1}{1-\delta}\phi_i'(\alpha_0)$ is simply the RHS minus the LHS of the inequality.  This means $U_{D_i}$ is increasing at $\alpha_i=\alpha_0$. Thus, as long as it is marginally profitable to improve the technology, the specialist prefers contributing to free-riding.
    \end{proof}
    
    \begin{lemma}
    \label{lemma:regimes4}
        If it is marginally costly to invest in the technology, meaning $r_i'(\alpha_0) < \frac{1}{1-\delta}\phi_i'(\alpha_0)$, then $D_i$ will not contribute -- instead, $D_i$ would always prefer to free-ride.
    \end{lemma}

    \begin{proof}{Proof.} If $r_i'(\alpha_0) < \frac{1}{1-\delta}\phi_i'(\alpha_0)$, then $\frac{\partial U_{D_i}}{\partial \alpha_i}\big\vert_{\alpha_i=\alpha_0} = r_i'(\alpha_0)-\frac{1}{1-\delta}\phi_i'(\alpha_0)$ is simply the RHS minus the LHS of the inequality.  This means $U_{D_i}$ is decreasing at $\alpha_i=\alpha_0$. Thus, as long as it is marginally costly to improve the technology, the specialist prefers free-riding to contributing.
    \end{proof}

    Taken together, we can conclude the following about combinations of conditions:
    \begin{itemize}
        \item Fixed costs under control, marginally profitable investment: A<F, F<C (Lemmas \ref{lemma:regimes1} and \ref{lemma:regimes3}). Thus the specialist would contribute.
        \item Fixed costs under control, marginally costly: A<F, C<F  (Lemmas \ref{lemma:regimes1} and \ref{lemma:regimes4}). Thus the specialist would free-ride.
        \item Fixed costs not under control, marginally profitable: F<A, F<C (Lemmas \ref{lemma:regimes2} and \ref{lemma:regimes3}). Thus the specialist would either abstain or contribute.
        \item Fixed costs not under control, marginally costly: F<A, C<F (Lemmas \ref{lemma:regimes2} and \ref{lemma:regimes4}). Thus the specialist would abstain.
    \end{itemize}
    Above, the short-hand notation `A,' `F,' and `C' refer to the strategies of abstaining, free-riding, and contributing, respectively. The optimal strategies follow from the two marginal conditions. This completes the proof.
\end{proof}
\section{Proofs for Deferred Lemmas}
\label{deferred-lemmas}

\subsection{Lemmas for Proving Proposition \ref{prop:multi-multi-powerful-g}}
\label{subsec:deferred-lemmas-multi-multi-powerful-G}

\xhdr{Proof of Lemma \ref{lemma:powerful-g-multiple-bargaining-individual-soln}}
\begin{proof}
        We solve for the powerful-G solution in domain $i$, $\delta_i^{P. G}$, for fixed values of all other bargains $\delta_{k}$.

        \begin{eqnarray*}
            & \delta_i^{P. G} = \arg\max_{\delta_i} U_G \\
            \Rightarrow & \frac{d U_G}{d \delta_i} =0 \\
            \Rightarrow & \frac{d}{\partial \delta_i}\left(\frac{1}{4c_0}\left(\sum_k \delta_k\right)^2 + \sum_k \frac{\delta_k(1-\delta_k)}{2c_k}\right) =0 \\
            \Rightarrow & \frac{1}{2c_0}\sum_k \delta_k+\frac{1-2\delta_i}{2c_i} = 0 \\
            \Rightarrow & \left(\frac{1}{2c_0}-\frac{1}{c_i}\right)\delta_i^* + \frac{1}{2c_0}\sum_{k:k\neq i} \delta_k + \frac{1}{2c_i} = 0 \\
            \Rightarrow & \delta_i^* = \frac{2c_0c_i}{2c_0-c_i} \left(\frac{1}{2c_0}\sum_{k:k\neq i} \delta_k + \frac{1}{2c_i}\right) \\
            \Rightarrow & \delta_i^* = \frac{1}{2c_0-c_i} \left({c_i}\sum_{k:k\neq i} \delta_k + {c_0}\right) \\
            \Rightarrow & \delta_i^* = \frac{1}{2-\frac{c_i}{c_0}} \left(\frac{c_i}{c_0}\sum_{k:k\neq i} \delta_k + 1\right)
        \end{eqnarray*}

        Notice the second derivative is $\frac{d^2 U_G}{d \delta_i^2} = \frac{1}{2c_0}-\frac{1}{c_i}$ which is positive as long as $c_i>2c_0$. Thus we have three possible cases:
        \begin{itemize}
            \item If $c_i>2c_0$, the critical point solved above represents a global minimum. For these values, the critical point is less than or equal to 0 because $\frac{1}{2c_0-c_i}<0$, so the optimum $\delta_i^{P. G}=1$.
            \item If $2c_0 \geq c_i > c_0$, then the critical point represents a maximum. The critical point $\delta_i^* = \frac{1}{2-\frac{c_i}{c_0}} \left(\frac{c_i}{c_0}\sum_{k:k\neq i} \delta_k + 1\right)$ is the product of two terms, both of which must be greater than or equal to 1. Thus, the powerful-G solution in this case is $\delta_i^{P. G}=1$.
            \item The remaining case is $c_i<c_0$. In this case (like the one above) the critical point $\delta_i^*$ represents a global maximum of a concave quadratic function. Within this condition, if the critical point $\delta_i^* \geq 1$, this would imply $\delta_i^{P. G}=1$. So, taking this inequality as a starting point: 
            \begin{eqnarray*}
                & \delta_i^* \geq 1 \\
                \Rightarrow & \frac{1}{2-\frac{c_i}{c_0}} \left(\frac{c_i}{c_0}\sum_{k:k\neq i} \delta_k + 1\right) \geq 1 \\
                \Rightarrow & \frac{c_i}{c_0}\sum_{k:k\neq i} \delta_k + 1 \geq 2-\frac{c_i}{c_0} \\
                \Rightarrow & \frac{c_i}{c_0}\sum_{k:k\neq i} \delta_k \geq 1-\frac{c_i}{c_0} \\
                \Rightarrow & \frac{c_i}{c_0}\left(\sum_{k:k\neq i}\delta_k +1\right) \geq 1 \\
                \Rightarrow & c_i \geq \frac{1}{\sum_{k:k\neq i}\delta_k +1} c_0
            \end{eqnarray*}
        \end{itemize}
    The above condition thus represents the cutoff above which the powerful-G solution is $1$, and below which the solution is given by the critical point. This completes the proof.
\end{proof}

\xhdr{Proof of Lemma \ref{lemma:powerful-G-non-decreasing-ci}}
\begin{proof}
    We prove this Lemma by separating the domain into two intervals: 
    
    \begin{itemize}
        \item First, in cases where $c_i \geq \frac{1}{\sum_{k:k\neq i}\delta_k + 1}c_0$, then $\delta_i^{P. G}=1$, which is non-decreasing.
        \item For all other values of $c_i$, namely $0 \leq c_i \leq \frac{1}{\sum_{k:k\neq i}\delta_k + 1}c_0$, we can complete the proof by showing $\frac{d \delta_i^{P. G}}{d c_i}\geq0$:
            $$\frac{d \delta_i^{P. G}}{d c_i} = \frac{d}{d c_i} \left(\frac{1}{2-\frac{c_i}{c_0}} \left(\frac{c_i}{c_0}\sum_{k:k\neq i} \delta_k + 1\right)\right) = \frac{\frac{1}{c_0}(2\sum_{k:k\neq i}\delta_k + 1)}{(2-\frac{c_i}{c_0})^2}$$
        The above quantity has a positive numerator in all cases. The denominator is positive for all values except $c_i=2c_0$, but this is outside the given interval: $c_i \leq \frac{1}{\sum_{k:k\neq i}\delta_k + 1}c_0 \leq c_0 < 2c_0$. 
    \end{itemize}
 \end{proof}

\xhdr{Proof of Lemma \ref{lemma:powerful-G-ci-at-least-half}}
\begin{proof}
        We've already established in Lemma \ref{lemma:powerful-G-non-decreasing-ci} that in any domain $i$, the powerful-G solution is a non-decreasing function of $c_i$. Thus, it suffices to show that $\lim_{c_i\rightarrow0}\delta_k^{P. G} = \lim_{c_i\rightarrow0} \frac{1}{2-\frac{c_i}{c_0}} \left(\frac{c_i}{c_0}\sum_{k:k\neq i} \delta_k + 1\right) = \frac{1}{2}(0+1)=\frac{1}{2}$.
\end{proof}

\xhdr{Proof of Lemma \ref{lemma-boundary-procedure-pg-multi-multi}}
\begin{proof}

For ease and without loss of generality, we assume the set of domains is in descending order of value $c_i$, and we'll consider each domain one at a time to determine whether the domain has $\delta_i=1$ using the condition derived in Lemma \ref{lemma:powerful-g-multiple-bargaining-individual-soln}. We start by analyzing the case where $i=n$, that is, we're at the \textit{last possible} domain, and for all others we'll derive a recursive relation to prove the Lemma.
\begin{itemize}
    \item Final state $i=n$. We're given $\delta_{i-1}=1$. So the condition in Lemma \ref{lemma:powerful-g-multiple-bargaining-individual-soln} evaluates as follows:
    \begin{eqnarray*}
        \delta_i=1 \iff &  \frac{c_i}{c_0} \geq \frac{1}{1+\sum_{k\neq i} \delta_k} \\
        \iff & \frac{c_i}{c_0} \geq \frac{1}{1+(n-1)(1)} \\
        \iff & \frac{c_i}{c_0} \geq \frac{1}{n} \\
        \iff & \frac{c_i}{c_0} \geq \frac{1}{n} \\
        \iff & \frac{c_i}{c_0} \geq \frac{2 - (0)}{2n} \\
        \iff & \frac{c_i}{c_0} \geq \frac{2 - \sum_{j>i}\frac{c_j}{c_n}}{n+i}, i=n \\ 
        \iff & c_i(n+i) \geq 2c_0- \sum_{j>i}c_j.
    \end{eqnarray*}
    We're left with the formula in our Lemma statement; note that there are no values $j:j>i$ so the summation is equal to 0.
    \item $i \in \{1,...,n-1\}$. Here there are two options: (1) $\delta_{i+1}=1$, in which case it is easy to show that $\delta_i=1$ because $\frac{c_i}{c_0}\geq \frac{c_{i+1}}{c_0} \geq \frac{1}{1+\sum_{k:k\neq i+1} \delta_k} \geq \frac{1}{1+\sum_{k:k\neq i} \delta_k}$. (2) $\delta_{i+1} < 1$, in which case we can plug in the formula in Lemma \ref{lemma:powerful-g-multiple-bargaining-individual-soln} and evaluate to define the condition without using any values $\delta$. 
     \begin{eqnarray*}
        \delta_i=1 \iff &  \frac{c_i}{c_0} \geq \frac{1}{1+\sum_{k\neq i} \delta_k} \\
         \iff & \frac{c_i}{c_0} \geq \frac{1}{i + \sum_{j>i} \delta_j} \\
    \end{eqnarray*}
    Notice that, even though the summation in the denominator looks nasty, thinking recursively allows us to simplify:
    \begin{eqnarray*}
    &  \delta_j = \frac{1}{2-\frac{c_j}{c_0}}\left(\frac{c_j}{c_0}\sum_{k:k\neq j}\delta_k + 1\right) \\
    \iff & \delta_j =  \frac{1}{2-\frac{c_j}{c_0}} +  \frac{\frac{c_j}{c_0}}{2-\frac{c_j}{c_0}}\left(i + \sum_{k=n+1}^n \delta_k - \delta_j\right) \\
    \iff & \left(1+\frac{\frac{c_j}{c_0}}{2-\frac{c_j}{c_0}}\right)\delta_j = \frac{1}{2-\frac{c_j}{c_0}} +  \frac{\frac{c_j}{c_0}}{2-\frac{c_j}{c_0}}\left(i + \sum_{k=n+1}^n \delta_k\right) \\
    \iff & \delta_j = \left(1-\frac{c_j}{2c_0}\right)\left(\frac{1}{2-\frac{c_j}{c_0}} +  \frac{\frac{c_j}{c_0}}{2-\frac{c_j}{c_0}}\left(i + \sum_{k=n+1}^n \delta_k\right)\right) \\
    \iff & \delta_j = \frac{1}{2} + \frac{c_j}{2c_0}\left(i + \sum_{k=n+1}^n \delta_k\right) \\
    \end{eqnarray*}
    Taking the sum, we evaluate another recursive relation:
    \begin{eqnarray*}
    & \sum_{j=i+1}^n\delta_j = \sum_{j=i+1}^n \left(\frac{1}{2} + \frac{c_j}{2c_0}\left(i + \sum_{k=n+1}^n \delta_k\right)\right) \\ 
    \iff & \sum_{j=i+1}^n\delta_j = \sum_{j=i+1}^n \frac{1}{2} + \sum_{j=i+1}^n \frac{c_j}{2c_0}\left(i + \sum_{k=n+1}^n \delta_k\right) \\ 
    \iff & \left(1 - \sum_{j=i+1}^n \frac{c_j}{2c_0}\right)\sum_{j=i+1}^n\delta_j = \frac{n-i}{2} + i \sum_{j=i+1}^n \frac{c_j}{2c_0} \\ 
    \iff & \sum_{j=i+1}^n\delta_j = \frac{n-i+i\sum_{j=i+1}^n \frac{c_j}{c_0}}{2-\sum_{j=i+1}^n \frac{c_j}{c_0}}
    \end{eqnarray*}
    Finally, we're in a position to plug into the original condition provided above for determining whether $\delta_i=1$:
    \begin{eqnarray*}
    \delta_i=1 \iff & \frac{c_i}{c_0} \geq \frac{1}{i + \sum_{j>i} \delta_j} \\
    \iff & \frac{c_i}{c_0} \geq \frac{1}{i + \frac{n-i+i\sum_{j=i+1}^n \frac{c_j}{c_0}}{2-\sum_{j=i+1}^n \frac{c_j}{c_0}}} \\
    \iff & \frac{c_i}{c_0} \geq  \frac{2-\sum_{j=i+1}^n \frac{c_j}{c_0}}{2i-i\sum_{j=i+1}^n \frac{c_j}{c_0} + n-i+i\sum_{j=i+1}^n \frac{c_j}{c_0}}
    \\
    \iff & \frac{c_i}{c_0} \geq  \frac{2-\sum_{j=i+1}^n \frac{c_j}{c_0}}{i+n}
    \\
    \iff & c_i (i+n) \geq 2c_0-\sum_{j=i+1}^n c_j
    \end{eqnarray*}
    \end{itemize}
    Notice we're left with the same closed-form expression in both cases. This concludes the Lemma's proof.
\end{proof}

\xhdr{Proof of Lemma \ref{lemma:powerful-g-multi-multi-concave}}
\begin{proof}
    We show the Hessian Matrix
$$
H := \left[
\begin{array}{ccc}
    \frac{\partial^2 U_G}{{\partial \delta_{i^* + 1}}^2} & \cdots & \frac{\partial^2 U_G}{\partial \delta_{n}\partial \delta_{i^* + 1}}  \\
    \vdots & \ddots & \vdots \\
    \frac{\partial^2 U_G}{\partial \delta_{i^* + 1}\partial \delta_{n}} & \cdots & \frac{\partial^2 U_G}{{\partial \delta_{n}}^2} \\
\end{array}
\right]
$$
is negative definite, meaning that the solutions representing the objective's critical points represent global maxima. (In our case these are attainable by solving a system of linear equations.)

The Hessian is solved below:
$$ H_{jj} = \frac{\partial^2 U_G}{{\partial \delta_{j}}^2} = \frac{1}{2c_0}-\frac{1}{c_j} \ ;\ H_{jk}\frac{\partial^2 U_G}{\partial \delta_{j}\partial \delta_{k}} = \frac{1}{2c_0} \text{for all } k \neq j.$$

Notice the Hessian is symmetric (by definition), real-valued, and its diagonal entries are negatively diagonal-dominant. It is established that a Hermitian (in our case, symmetrical matrix) that is strictly diagonal-dominant (meaning diagonal entries are greater in magnitude than the sum of non-diagonal members of the same row) with real positive diagonal entries is positive definite. It follows that our matrix is negative definite so long as $-H$ is strictly diagonal-dominant with real positive diagonal entries, that is: 
\begin{eqnarray*}
    -H \text{ is P.D.} \iff & \left|-H_{jj}\right|>\sum_{k \neq j} -H_{jk} \text{ and } -H_{jj}>0 \ \forall j \in \{i^*+1,...,n\}\\
    \iff & \frac{1}{c_j}-\frac{1}{2c_0} > \sum_{k \neq j} \frac{1}{2c_0} \text{ and } \frac{1}{c_j}-\frac{1}{2c_0}>0 \ \forall j  \\
    \iff & {c_j} < \frac{2c_0}{n} \text{ and } c_j<2c_0 \ \forall j  \\
    \iff & {c_j} < \frac{2c_0}{n} \ \forall j  \\
    \iff & {c_{i^*+1}} < \frac{2c_0}{n}. \\
\end{eqnarray*}
Given that the costs $\vec{c}$ are given in descending order, it suffices to prove this inequality for $j=i^* + 1$, the unsolved domain with the greatest cost parameter. Recall that we arrived at a value for $i^*$ through a procedure, justified in Lemma \ref{lemma-boundary-procedure-pg-multi-multi}, where we establish that $c_{i^*}(i^*+n)\geq 2c_0 - \sum_{j>i^*}c_j$ whereas $c_{i^* + 1}(i^*+1+n)< 2c_0 - \sum_{j>i^*+1}c_j$. We only need this final inequality to establish $H$ is negative definite:
\begin{eqnarray*}
& c_{i^* + 1}(i^*+1+n)< 2c_0 - \sum_{j>i^*+1}c_j \\ 
\Rightarrow & c_{i^* + 1} < \frac{2c_0-\sum_{j>i^*+1}c_j}{n+i^* + 1} \\
\Rightarrow & c_{i^* + 1} < \frac{2c_0}{n+i^*+1}-\frac{\sum_{j>i^*+1}c_j}{n+i^* + 1} \geq \frac{2c_0}{n+i^* + 1} \geq \frac{2c_0}{n}. 
\end{eqnarray*}
\end{proof}

\subsection{Lemmas for Proving Proposition \ref{prop:multi-multi-powerful-d}}
\label{subsec:deferred-lemmas-multi-multi-powerful-D}

\xhdr{Proof of Lemma \ref{lemma-u-d-non-decreasing}}
\begin{proof}
    It suffices to show that for any given bargain $k \neq i$, the partial derivative of the utility of $U_{D_i}$ with respect to $\delta_k$ is non-negative for all feasible values of $\delta_k$. Our job is made easier because $U_{D_i}$ is linear as a function of $\delta_k$, so the partial derivative is a constant value:
    $$\frac{\partial}{\partial \delta_k}U_{D_i}=\frac{\partial}{\partial \delta_k}\left[\frac{1-\delta_i}{2c_0}\sum_j\delta_j + \frac{(1-\delta_i)^2}{4c_i}\right] = \frac{1-\delta_i}{2c_0} \geq 0$$
    The last step follows from given requirements that $c_0>0$ and $\delta_i\leq 1$.
\end{proof}

\xhdr{Proof of Lemma \ref{powerful-d-multi-quadratic-fixed-j}}
\begin{proof}
$\delta_i^{P.D_i}$ is the solution that maximizes $D_i$'s utility. Thus, below we find the critical points and perform a second-derivative test to ensure that our candidate solution indeed represents a global optimum. 
\begin{eqnarray*}
& \delta^{P. G}_i := \arg\max_{\delta_i} U_{D_i} \\
\Rightarrow & \frac{\partial U_{D_i}}{\partial \delta_i} = 0 \\
\Rightarrow & \frac{\partial}{\partial \delta_i} \left[\frac{1-\delta_i}{2c_0}\sum_j\delta_j + \frac{(1-\delta_i)^2}{4c_i}\right] = 0 \\
\Rightarrow & \frac{1}{2c_0} - \frac{1}{c_0}\delta_i - \frac{1}{2c_0}\sum_{j:j\neq i}\delta_j-\frac{(1-\delta_i)}{2c_i} = 0 \\
\Rightarrow & \left(\frac{1}{2c_i}-\frac{1}{c_0}\right)\delta_i^* = -\frac{1}{2c_0} +\frac{1}{2c_0}\sum_{j:j\neq i}\delta_j +\frac{1}{2c_i} \\
\Rightarrow & \delta_i^* = \frac{c_0 + c_i\sum_{j:j\neq i}\delta_j - c_i}{c_0-2c_i}.
\end{eqnarray*}

The second derivative $\frac{\partial^2 U_{D_i}}{\partial \delta_i^2} = -\frac{1}{c_0}+\frac{1}{2c_i}$ is negative when $c_0 < 2c_i$. 

Thus we have two cases:
\begin{itemize}
    \item If $c_0 < 2c_i$, the critical point above represents a global maximum of a concave quadratic function, meaning it represents the solution if it is feasible, and otherwise the nearest boundary of the feasible set $\delta_i \in \{0,1\}$ represents the solution. The value of the critical point is less than 1 in all cases so $\delta_i^{P.D_i}= \max \left[0,  \frac{c_0-c_i\left(1-\sum_{j:j\neq i}\delta_j\right)}{c_0-2c_i}\right]$.
    \item Else $c_0 \geq 2c_i$. In this case, the critical point above is greater than 1 and represents a global minimum of a concave quadratic function, meaning the solution in this case must be the feasible point at greatest distance from the critical point, which is $\delta_i^{P.D_i}=0$. 
\end{itemize}
\end{proof}

\end{document}